\documentclass[onecolumn,pra,longbibliography,nofootinbib]{revtex4-2}

\usepackage[dvips]{graphicx} 
\usepackage{amsfonts}
\usepackage{amssymb}
\usepackage{amscd}
\usepackage{amsmath}    
\usepackage{enumerate}
\usepackage{epsfig}
\usepackage{subfigure}
\usepackage{subfloat}
\usepackage{xcolor}
\usepackage{amsthm}
\usepackage{physics}
\usepackage[most]{tcolorbox}
\usepackage[]{qcircuit}
\usepackage{tabularx}
\usepackage{float}
\usepackage{appendix}
\usepackage{makecell}
\usepackage{algorithm}
\usepackage{algorithmicx}
\usepackage{algpseudocode}
\usepackage{dsfont}
\usepackage[colorlinks, linkcolor=red, anchorcolor=blue, citecolor=green]{hyperref}
\usepackage{MnSymbol}
\newtcolorbox[auto counter]{mybox}[2]{
enhanced,
breakable,
label=#1,
colback=blue!5!white,
colframe=blue!75!black,
fonttitle=\bfseries,
title=Box \thetcbcounter: #2
}

\newtheorem*{rep@theorem}{\rep@title}
\newcommand{\newreptheorem}[2]{%
\newenvironment{rep#1}[1]{%
\def\rep@title{#2 \ref{##1}}%
\begin{rep@theorem}}%
{\end{rep@theorem}}}
\makeatother

\newtheorem{theorem}{Theorem}
\newtheorem{lemma}{Lemma}
\newtheorem{question}{Question}
\newreptheorem{theorem}{Theorem}
\newreptheorem{lemma}{Lemma}
\newreptheorem{question}{Question}

\newtheorem{corollary}{Corollary}

\newtheorem{definition}{Definition}
\newtheorem{proposition}{Proposition}

\newtheorem{example}{Example}

\newcommand{\lket}[1]{\vert #1 \rangle\!\rangle}
\newcommand{\lbra}[1]{\langle\!\langle #1 \vert}
\newcommand{\lbraket}[2]{\langle\!\langle #1 \vert #2 \rangle\!\rangle}
\newcommand{\lketbra}[2]{\vert #1 \rangle\!\rangle\langle\!\langle #2 \vert}

\begin{document}
\title{Group twirling and noise tailoring for multi-qubit controlled phase gates}
\author{Guoding Liu}
\author{Ziyi Xie}
\author{Zitai Xu}
\author{Xiongfeng Ma}
\email{xma@tsinghua.edu.cn}
\affiliation{Center for Quantum Information, Institute for Interdisciplinary Information Sciences, Tsinghua University, Beijing, 100084 China}

\begin{abstract}
Group twirling is crucial in quantum information processing, particularly in randomized benchmarking and random compiling. While protocols based on Pauli twirling have been effectively crafted to transform arbitrary noise channels into Pauli channels for Clifford gates --- thereby facilitating efficient benchmarking and mitigating worst-case errors --- practical twirling groups for multi-qubit non-Clifford gates are lacking. In this work, we study the issue of finding twirling groups for generic quantum gates within a widely used circuit structure in randomized benchmarking or random compiling. For multi-qubit controlled phase gates, which are essential in both the quantum Fourier transform and quantum search algorithms, we identify optimal twirling groups within the realm of classically replaceable unitary operations. In contrast to the simplicity of the Pauli twirling group for Clifford gates, the optimal groups for such gates are much larger, highlighting the overhead of tailoring noise channels in the presence of global non-Clifford gates.
\end{abstract}

\maketitle

\section{Introduction}\label{sc:intro}
There has been an increased interest in quantum information processing due to its potential revolution in both science and technology. While quantum computing holds the promise of quantum advantages, its practical implementation faces significant challenges, primarily due to the inherent noise of quantum systems. In dealing with quantum noise, the group twirling based noise tailoring is an essential step. Group twirling symmetrizes the noise channel~\cite{Emerson2005Scalable, Emerson2007Characterization}, allowing accurate and efficient extraction of noise channel parameters. Then, it enables efficient noise characterization and subsequent quantum control optimization~\cite{Vandersypen2005control, Chu2002control}. This principle underlies the randomized benchmarking (RB) methodology~\cite{Knill2008RB, Emerson2011prlRB, Emerson2012praRB, Harper2020Efficient, Harper2021Estimation}, which stands out as a major quantum benchmarking technique due to its low sample complexity and resilience against state preparation and measurement (SPAM) errors. Moreover, group twirling is essential in random compiling~\cite{Wallman2016compiling}, which turns generic noise into a specific type and reduces the worst-case error of quantum gates.

While randomized benchmarking and random compiling have seen considerable advancements, efficient noise tailoring protocols for multi-qubit gates mainly focus on the Clifford case, primarily achieved through Pauli group twirling~\cite{Erhard2019CB, Zhang2023Scalable, Wallman2016compiling}. For multi-qubit non-Clifford gates, researchers also made some progress and proposed fascinating benchmarking protocols for two kinds of quantum gate sets: CNOT dihedral group~\cite{Cross2016dihedral} and matchgate group~\cite{Helsen2022matchgate}. Nonetheless, these groups are both large and highly non-local, posing challenges to practical implementation in expansive quantum systems. Experimental undertakings have, so far, been limited to the two-qubit domain for the CNOT dihedral group~\cite{GarionControlledS2021}. This raises the question of whether more compact and easily implementable twirling groups might suffice for twirling and benchmarking tasks. Similar issues exist for random compiling, with no methodologies for tailoring multi-qubit non-Clifford gates. This hinders the application of multi-qubit non-Clifford gates, though many of which are directly implementable in quantum processors~\cite{nguyen2022programmable,Kim2022iToffoli} and possess transversal schemes~\cite{Paetznick2013Transversal,Jochym2021transversal}.


In this letter, we investigate noise tailoring strategies for generic quantum gates within the contexts of RB and random compiling. Central to noise tailoring is the selection of appropriate twirling gates. We study this question within a frequently used circuit structure. The choice of twirling gates is not arbitrary. Introducing twirling gates must maintain the logical circuit, which imposes algebraic constraints between the gate undergoing twirling and the twirling gates themselves. Furthermore, RB brings along an added stipulation: the tailored noise channel should commute with the target gate, which is helpful to get the powers of the twirled noise channel and enable fitting and extracting SPAM-error-free noise parameters. Within this framework, our findings indicate that any quantum gate tailoring demands a twirling gate set comparable in magnitude to the Pauli group, implying the optimality of existing noise tailoring schemes designed for Clifford gates.

In addition to the well-studied Clifford gates, our study emphasizes multi-qubit controlled phase gates, given by
\begin{equation}
\begin{split}
C^nZ_m = \begin{pmatrix}
\mathbb{I}_{2^{n+1}-1} & \mathbf{0} \\
\mathbf{0} & e^{i\frac{2\pi}{m}}
\end{pmatrix},
\end{split}
\end{equation}
where $n$ and $m$ are both positive integers, representing the number of control qubits and the phase, respectively. The matrix $\mathbb{I}_{2^{n+1}-1}$ denotes the identity of dimension $2^{n+1}-1$ and $\mathbf{0}$ denotes the zero matrix. Our results can be further generalized by replacing the phase from $\frac{2\pi}{m}$ to any angle $\theta \in [0, 2\pi]$. These gates are the key components in quantum Fourier transform and are locally equivalent to the operator $2\ketbra{\psi}-\mathbb{I}_{2^{n+1}}$ in Grover's search algorithm where $\ket{\psi}$ denotes the target search state. In principle, $C^{n}Z_m$ can be realized by interactions between $\ket{1}^{\otimes (n+1)}$ and high energy levels~\cite{Fedorov2012Toffoli,Li2019CCZ}. Thus, devising efficient noise tailoring schemes for such gates is crucial in both the theory and experiment. Within our framework, we introduce an optimal noise tailoring scheme for $C^{n}Z_m$ if the twirling group falls within the class of classically replaceable unitary operations (CRU)~\cite{Liu2022CRO} or incoherent unitary operations~\cite{Baumgratz2014Coherence,Eric2016DIO}. CRU comprises gates that can be moved after computational basis measurements, typically set as $Z$-basis measurements, and replaced by classical post-processing. Contrasting with the simplicity in tailoring Clifford gates, the optimal twirling gate set for $C^nZ_m$ grows exponentially with the qubit count. This shows an underlying difficulty in tailoring multi-qubit non-Clifford gates.

The structure of this paper is organized as follows. In Section~\ref{sc:group}, we mathematically describe the problem of finding the twirling group for a target gate in RB and show the main results. In Section~\ref{sc:randomcompiling}, we briefly introduce the results of random compiling. In Section~\ref{sc:simulation}, we show simulation results of the RB protocols based on the new twirling groups. Finally, we conclude and discuss in Section~\ref{sc:outlook}. More details about this work, including the full proof of the main results, are left in the Appendices.

\section{Twirling groups in randomized benchmarking}\label{sc:group}
We start with clarifying the algebraic relations between the twirled gate and the twirling gates in RB, following the ideas of~\cite{Helsen2022RBFramework, Helsen2019characterRB, Magesan2012interleavedRB, Erhard2019CB}.


The task here is estimating the fidelity of an individual gate, $U$, robust to SPAM errors. In general, one can express the noisy quantum gate $\widetilde{\mathcal{U}} = \mathcal{U} \Lambda$ as the composite of the noiseless gate $\mathcal{U}$ and its noise channel $\Lambda$. Here, $\mathcal{U}$ denotes the Pauli-Liouville representation of $U$, and $\widetilde{\cdot}$ represents the noisy version of quantum gates or observables. For a noise channel, we use the same notation $\Lambda$ to represent its Liouville representation and other representations. The key to enable RB is obtaining the powers of the $\mathsf{G}$-twirled noise channel, $\Lambda_G^m$, where $m\in \mathbb{Z}_+$ and $\Lambda_G = \mathbb{E}_{G\in \mathsf{G}} \mathcal{G}\Lambda \mathcal{G}^{\dagger}$. $\mathsf{G}$ is a prefixed twirling group, ensuring that $\Lambda_G$ is diagonal in the Pauli-Liouville representation, up to a unitary transformation independent of $\Lambda$. The diagonalizability ensures $\tr(M\Lambda_G^m(\rho))$ to be multiple exponential decay function $\sum_i A_ip_i^m$ for any state $\rho$ and observable $M$. Then, via the technique of character RB~\cite{Helsen2019characterRB}, one can obtain $A_ip_i^m$ and extract $p_i$ with accurate single-exponential fitting. The linear combination of $p_i$ gives the fidelity estimation robust to SPAM errors. If $\Lambda_G$ is not diagonal, $p_i$ will turn into a matrix, and we will meet the notorious matrix exponential fitting problem~\cite{Helsen2022RBFramework}, which would cause inaccurate fidelity estimation.

To obtain $\Lambda_G^m$, we consider the circuit in Figure~\ref{fig:RBmaintext}. One independently and randomly samples $m$ twirling gates $G_i$ from group $\mathsf{G}$ and implements them interleaved with target gate $U$. The circuit ends with the inverse gate $G_{\text{inv}} = (\prod_{i=1}^m UG_i)^{\dagger}$. This circuit structure is first devised in interleaved RB~\cite{Magesan2012interleavedRB}, but the difference here is that we do not require $U\in \mathsf{G}$. The target gate $U$ is the one acquiring noise tailoring, and its noise channel is twirled, so we call $U$ the twirled gate. The gates from $\mathsf{G}$ are named the twirling gates.

\begin{figure}[htbp!]
\centering
\includegraphics[width=.47\textwidth]{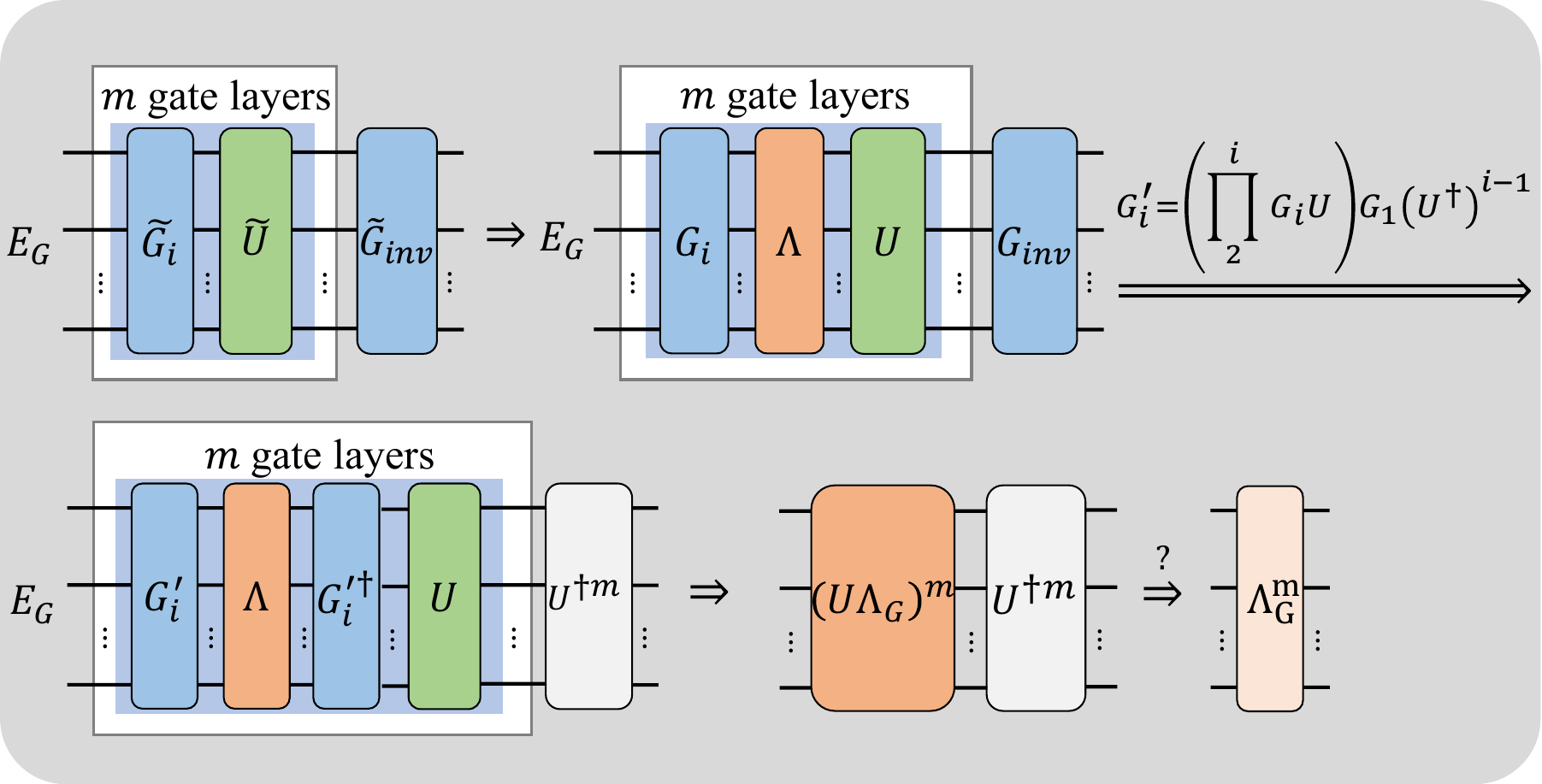}
\caption{Random twirling gates $G_1, G_2,\cdots,G_m$ interleaved with $U$ in RB. The inverse gate $G_{\text{inv}} = (\prod_{i=1}^m UG_{i})^{-1}$. $\widetilde{\cdot}$ represents the noisy version of a quantum gate. $\Lambda$ is the noise channel of $U$ and $\Lambda_G$ is the $G$-twirled noise channel.}
\label{fig:RBmaintext}
\end{figure}

Note that the interleaved circuit actually measures the noise from both $U$ and $\mathsf{G}$. One can invoke the technique of interleaved RB~\cite{Magesan2012interleavedRB} to remove the noise effect from $\mathsf{G}$. Below, we focus on the case that the twirling gates are noiseless, and the results can be extended to more general cases. Mathematically, the circuit is
\begin{equation}\label{eq:RB}
\begin{split}
\widetilde{\mathcal{S}}_m &=  \mathcal{G}_{inv}\prod_{i=1}^m  \mathcal{U}\Lambda\mathcal{G}_{i}\\
&= \mathcal{U}^{\dagger m} \prod_{i=1}^m  \mathcal{U}\mathcal{G}_i^{\prime\dagger}\Lambda\mathcal{G}'_i,
\end{split}
\end{equation}
where $G'_i = (\prod_{j=2}^i G_jU) G_1 U^{\dagger i-1}$ for $1\leq i\leq m$. Below we account for that $U\mathsf{G}U^{\dagger}=\mathsf{G}$ will make $\widetilde{\mathcal{S}}_m$ close to $\Lambda_G^m$ in expectation.

If $U\mathsf{G}U^{\dagger}=\mathsf{G}$, then $\{G'_i, 1\leq i\leq m\}$ are independent and random elements from $\mathsf{G}$. After taking expectation, Eq.~\eqref{eq:RB} would become $\mathcal{U}^{\dagger m} (\mathcal{U}\Lambda_G)^m$. For multi-qubit controlled phase gates and controlled phase gates except for CZ, we found that $U\mathsf{G}U^{\dagger}=\mathsf{G}$ ensures $\Lambda_G \mathcal{U} = \mathcal{U}\Lambda_G$. Then, Eq.~\eqref{eq:RB} would further reduce to $\Lambda_G^m$ as we want. For other quantum gates, one can at least obtain the fidelity of $(\mathcal{U}^{\dagger m}(\mathcal{U}\Lambda_G)^m)^{\frac{1}{m}}$, which we will prove to be a lower bound of the fidelity of $\Lambda$ if $U\mathsf{G}U^{\dagger}=\mathsf{G}$ and $\Lambda_G$ is diagonal. The lower bound of the fidelity is known to be fidelity witness~\cite{Eisert2020benchmarking} and is also useful in quantum benchmarking. It is worth noting that $U^m = \mathbb{I}$ and $U\mathsf{G}U^{\dagger}=\mathsf{G}$ also ensures the inverse gate belonging to $\textsf{G}$. More details are available in Appendices~\ref{appendssc:rbrequire} and~\ref{appendssc:uconjchannel}. We summarize the requirements of the twirling group to tailor $U$ in RB below.


\begin{question}\label{ques:RB}
Given a gate, $U$, find a twirling group, $\mathsf{G}$ such that,
\begin{equation}\label{eq:diag}
\begin{split}
&\text{for any quantum channel } \Lambda, \Lambda_G = \mathbb{E}_{G\in \mathsf{G}} \mathcal{G} \Lambda \mathcal{G}^{\dagger} \text{ is diagonal}\\
&\text{up to a unitary transformation independent of } \Lambda,\\
\end{split}
\end{equation}
and,
\begin{equation}\label{eq:Uconjgroup}
U\mathsf{G} U^{\dagger} = \mathsf{G}.
\end{equation}
\end{question}

Equation~\eqref{eq:diag} means that $\mathsf{G}$ should make the twirled noise channel sufficiently symmetric no matter what the original noise channel is, and Eq.~\eqref{eq:Uconjgroup} ensures that the action of $U$ does not destroy this symmetry. The two requirements lie a fundamental limitation on the size and the structure of the twirling group. A good solution of $\mathsf{G}$ should be as small and easily implementable as possible. To find the optimal group for a gate, we study the diagonalizability of $\Lambda_G$, which puts restriction on $\mathsf{G}$ as below.

\begin{lemma}\label{lemma:cardconstraint}
If a finite $n$-qubit unitary subgroup, $\mathsf{G}$, satisfies Eq.~\eqref{eq:diag}, then the Pauli-Liouville representation of $\mathsf{G}$ is multiplicity-free. As a corollary, the cardinality of the twirling group $\abs{\mathsf{G}}\geq 4^n$.
\end{lemma}

The proof of Lemma~\ref{lemma:cardconstraint} is shown in Appendix~\ref{appendssc:lemmacardconstraint}. We use the arbitrariness of $\Lambda$ and the diagonalizability of $\Lambda_G$ to show $\mathsf{G}$ multiplicity-free, which means any irreducible representation of $\mathsf{G}$ appears at most once in the decomposition of the Liouville representation of $\mathsf{G}$. Then, with Burnside's theorem~\cite{curtis1966representation}, one obtains $\abs{\mathsf{G}}\geq 4^n$.

As the cardinality of the (projective) Pauli group is $4^n$, the above theorem reveals that the twirling group for noise tailoring cannot be smaller than the Pauli group. Pauli gates are also local and easily implementable. Thus, we show the superiority of Clifford gates in noise tailoring, for which the Pauli group is a solution to Question~\ref{ques:RB}. For a non-Clifford gate, the solution cannot be the Pauli group, for which we present a simple but may not be an optimal solution in Appendix~\ref{appendssc:construct}.

Lemma~\ref{lemma:cardconstraint} shows requirements for generic unitary twirling groups. Below, in Lemma~\ref{lemma:cosetcardconstraint}, we consider twirling groups belonging to CRU and show much stronger requirements other than the cardinality constraint for twirling groups. CRU comprises all gates in the product of a permutation matrix, like Pauli $X$ and Toffoli gate, and a diagonal matrix in the computational basis, like Pauli $Z$. This set is large and becomes universal after adding Hadamard gates. Using CRU subgroups for twirling also lets us replace the inverse gate with classical postprocessing when measurements are under the computational bases~\cite{Liu2022CRO}, bringing additional advantages.

\begin{lemma}[Informal Version]\label{lemma:cosetcardconstraint}
If a finite $n$-qubit CRU subgroup, $\mathsf{G}$, satisfies Eq.~\eqref{eq:diag}, then $\mathsf{G}$ can interchange any two computational basis states. That is, for any two computational basis states $\ket{\mathbf{i}}$ and $\ket{\mathbf{j}}$ where $\mathbf{i}, \mathbf{j}\in \{0,1\}^n$, there exists a gate $G\in \mathsf{G}$ such that $\ket{\mathbf{j}} = G\ket{\mathbf{i}}$.
\end{lemma}

To prove Lemma~\ref{lemma:cosetcardconstraint}, we leverage the multiplicity-free condition of $\mathsf{G}$ in the Liouville representation in Lemma~\ref{lemma:cardconstraint} and that CRU transforms diagonal matrices to diagonal matrices. By examining the multiplicity of the trivial representation, we obtain the result. The complete proof and more details on optimizing the group within CRU are available in Appendix~\ref{appendssc:lemmacosetcardconstraint}.

Note that diagonal gates, CNOT gates, and Toffoli gates do not affect $\ket{\mathbf{0}}$, Lemma~\ref{lemma:cosetcardconstraint} implies that the CRU subgroup $\mathsf{G}$ must include a gate set like $\mathsf{X} = \langle X \rangle^{\otimes n}$. Combining this with Eq.~\eqref{eq:Uconjgroup}, we can deduce the necessary quantum gates that $\mathsf{G}$ must contain for tailoring a gate, $U$. Furthermore, when $U$ is a multi-qubit controlled phase gate, we derive its optimal twirling group as detailed in Theorem~\ref{thm:cnzm}, with full proof available in Appendix~\ref{appendssc:thmcnzm}.

\begin{theorem}\label{thm:cnzm}
The optimal twirling group $\mathsf{G}$ in CRU for the multi-qubit controlled phase gate, $U = C^n Z_m$, with $n\geq 1, m\geq 2$, is the smallest group containing $\mathsf{X}$ and normalized by $U$, given by
\begin{equation}\label{eq:twirlgroupCnZm}
\mathsf{G} = \{\Pi  (\prod_{i=1}^t (\Pi_i^{\dagger} U\Pi_i U^{\dagger})^{l_i} ) | \Pi\in\mathsf{X}, t\in \mathbb{N}, \forall i, l_i\in {\pm 1}, \Pi_i \in \mathsf{X} \}.
\end{equation}
\end{theorem}

The optimal group $\mathsf{G}$ is generated by permutation matrices in $\mathsf{X}$ and diagonal matrices $\prod_{i=1}^t U^{l_i}\Pi_i U^{l_i\dagger}\Pi_i^{\dagger}$, which we proved to satisfy Question~\ref{ques:RB}. The optimality here means that any group satisfying Question~\ref{ques:RB} must contain all diagonal matrices in $\mathsf{G}$ and a permutation matrix set no smaller than $\mathsf{X}$, making $\mathsf{G}$ to be the smallest and most easily implementable choice.

Specific forms of twirling groups for $C^nZ$ and $CZ_m$ are summarized in Table~\ref{table:compelxity}. $\mathsf{G}$ reduces to $\langle C^{n-1}Z, C^{n-2}Z, \cdots, CZ, Z, X \rangle$ and $\langle CZ_m, Z_m, X \rangle$ for $U=C^nZ$ and $U=CZ_m$ with odd $m$, respectively. For $U=CZ_m$ with even $m$, $\mathsf{G}\leq \langle CZ_{m/2}, Z_m, X \rangle$. Within group generator $\langle\cdot\rangle$, we use $X$ to denote $\{X_1, X_2, \cdots, X_n\}$, the Pauli $X$ gates on all qubits, and similarly omit subscripts associated with the acting qubits when referring to $C^{n-1}Z, C^{n-2}Z, \cdots, CZ, Z$ and $CZ_m, Z_m$.

\begin{table}[!ht]
\centering
\resizebox{.45\textwidth}{!}{


\begin{tabular}{cccc}%
& \multicolumn{3}{c}{$C^nZ$} \\ \hline
& Group & Size & Computing \\ 
This work & $\langle C^{n-1}Z, C^{n-2}Z, \cdots, CZ, Z, X \rangle$ & $O(N^n)$ & $O(N^{n})$ \\ 
CXD~\cite{Cross2016dihedral} & $\langle CX, Z_{2^{n+1}}, X \rangle$ & $O(N^{n+1})$ & $O(N^{3n+1})$  \\
& \multicolumn{3}{c}{$CZ_m$} \\ \hline
& Group & Size & Computing \\ 
This work & $\langle CZ_m, Z_m, X \rangle$ & $O(N^2\log m)$ & $O(N^2\log m)$ \\ 
CXD & $\langle CX, Z_{2m}, X \rangle$ & $O(N^{\log m})$ & $O(N^{3\log m+1})$ \\ 
\end{tabular}

}
\caption{Scaling of the group size and computational complexity of the twirling group for tailoring $C^nZ$ and $CZ_m$ in our work and~\cite{Cross2016dihedral}. CXD represents the CNOT dihedral group. The complexity is expressed with the qubit number $N$, controlled qubit number $n$, and phase angle index $m$. Unlike~\cite{Cross2016dihedral}, which specifically addresses computational complexity for $m=2^k$, our results apply to generic positive integer $m$.}
\label{table:compelxity}
\end{table}

Note that enabling noise tailoring requires sampling from the twirling group and performing group element multiplication. These tasks introduce considerations for sample complexity, directly related to group size and computational complexity, contingent on the algorithm used for group multiplication. Table~\ref{table:compelxity} provides complexity results for $C^nZ$ and $CZ_m$, with the derivation available in Appendix~\ref{appendssc:complexity}. Beyond noise tailoring, these findings may be useful in investigating the classical simulation of various twirling groups, whose classical simulability forms an interesting computational complexity hierarchy concerning both $n$ and $m$.

When benchmarking an individual multi-qubit controlled phase gate $C^nZ_m$, our method surpasses previous results in~\cite{Cross2016dihedral} regarding group size and computational complexity, as shown in Table~\ref{table:compelxity}. Nonetheless, even with this optimal approach, the twirling group becomes large and highly non-local as the controlled qubit number increases. This leads to unfavorable sample and computational complexity for large quantum gates, highlighting the inherent challenges in benchmarking multi-qubit non-Clifford gates.

\section{Random compiling for multi-qubit non-Clifford quantum gates}\label{sc:randomcompiling}
Below, we briefly illustrate the results of random compiling. The algebraic restrictions between the twirled gate and the twirling gates here are looser than those in RB. The task is turning the noisy quantum gate $\widetilde{\mathcal{U}} = \mathcal{U} \Lambda$ to $\mathcal{U} \Lambda_G$ where ensuring $\Lambda_G$ to be a Pauli channel. To achieve that, one inserts twirling gates $G$ and $G'=UG^{\dagger}U^{\dagger}$ beside $U$ as shown in Figure~\ref{fig:RC} where $G$ is randomly sampled from twirling group $\mathsf{G}$. In reality, there may be multiple sequential target gates $U$ requiring noise tailoring and the twirling gates for different target gates would merge and, in general, belong to $\mathsf{V} = \mathsf{G}U\mathsf{G}U^{\dagger} = \{G_iUG_jU^{\dagger}|G_i, G_j\in \mathsf{G}\}$. Thus, in random compiling, we should optimize the choice of $\mathsf{G}$ to simplify $\mathsf{V}$. Similar to the results in RB, $\mathsf{X}U\mathsf{X}U^{\dagger}\mathsf{Z} = \{\Pi_iU\Pi_jU^{\dagger}W|\Pi_i, \Pi_j\in \mathsf{X}, W\in \mathsf{Z}\}$ would be a nearly optimal choice for multi-qubit controlled phase gates. More details about random compiling are shown in Appendix~\ref{appendsc:RC}.

\begin{figure}[htbp!]
\centering
\includegraphics[width=.47\textwidth]{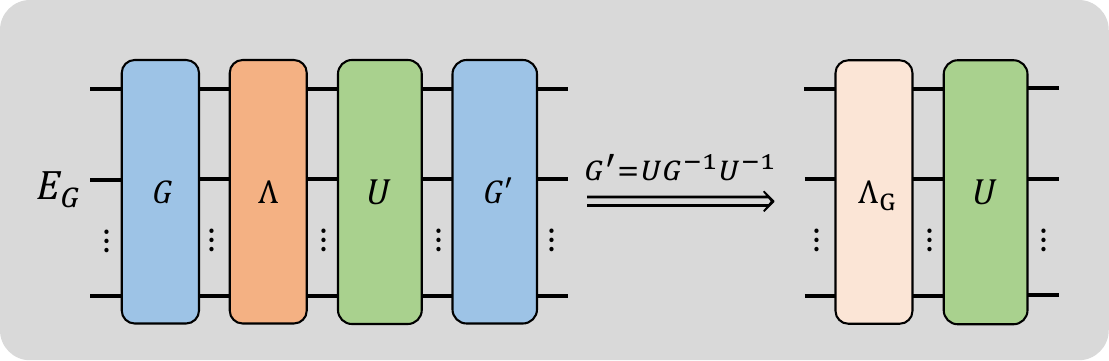}
\caption{Random compiling for $U$ by inserting twirling gates $G$ and $G' = UG^{-1}U^{-1}$ beside it.  $\Lambda_G$ is the $G$-twirled noise channel required to be a Pauli channel.}
\label{fig:RC}
\end{figure}

\section{Randomized benchmarking simulation for multi-qubit controlled phase gates}\label{sc:simulation}
Below, we show the benchmarking results for  $C^nZ$ gates using our scheme and compare it to~\cite{Cross2016dihedral}. In the simulation, to tailor $C^nZ$, we use group $\mathsf{G}^Z = \langle C^{n-1}Z,\cdots, CZ, X, S\rangle$ instead of the optimal twirling group $\langle C^{n-1}Z,\cdots, CZ, X, Z\rangle$. A small overhead of implementing phase gates $S$ would make the twirled noise channel more symmetric and reduce the number of its parameters from $3*(2^{n+1}-1)$ to $2*(2^{n+1}-1)$. This allows us to extract the noise parameters with only two SPAM settings:
\begin{enumerate}
\item\label{item:sm1}
preparing $\ket{0}^{\otimes (n+1)}$ and measuring in $Z$ basis;

\item\label{item:sm2}
preparing $\ket{+}^{\otimes (n+1)}$ and measuring in $X$ basis.
\end{enumerate}
For setting~\ref{item:sm1}, one implements the circuit in Figure~\ref{fig:RBmaintext} and get $\tr(\widetilde{M}\Lambda_{\mathsf{G}^Z}^m(\widetilde{\rho}))$ where $\Lambda_{\mathsf{G}^Z}$ is the twirled noise channel, $\rho = \ketbra{0}^{\otimes (n+1)}$, and $M\in \{\mathbb{I},Z\}^{\otimes (n+1)}$. Each non-identity choice of $M$ would make $\tr(\widetilde{M}\Lambda_{\mathsf{G}^Z}^m(\widetilde{\rho}))$ approximately a single exponential decay function $A\lambda^m$, while $\lambda$ is a noise parameter. Via fitting, one can obtain $2^{(n+1)}-1$ distinct parameters. Same for the setting~\ref{item:sm2}. The linear combination of these parameters provides the fidelity of $\Lambda$.

In~\cite{Cross2016dihedral}, the group to tailor $C^nZ$ is CNOT dihedral group $\mathsf{G}^X = \langle CX, X, Z_{2^n}\rangle$. $\Lambda_{\mathsf{G}^X}$ only has two parameters, which are separately extracted by preparing two distinct states, $\ket{0}^{\otimes (n+1)}$ and $\ket{+}^{\otimes (n+1)}$, measuring the probability back to the initial state, and fitting the probability to function $A\lambda^m+B$. The linear combination of two parameters also evaluates the fidelity of $\Lambda$. Except for the original SPAM setting, we also simulate the benchmarking protocol using the CNOT dihedral group and the SPAM settings~\ref{item:sm1} and~\ref{item:sm2} in our scheme to make a fair comparison.


The results are shown in Figures~\ref{fig:CCZ} and~\ref{fig:CnZ}, demonstrating the results for CCZ versus the number of sampled sequences and the result for $C^nZ$ versus the number of controlled qubits, respectively. More simulation details are presented in Appendix~\ref{appendsc:simulation}.

\begin{figure}[!htbp]
\subfigure[CCZ]{
\begin{minipage}[b]{0.42\textwidth}\label{fig:CCZ}
\includegraphics[scale=0.385]{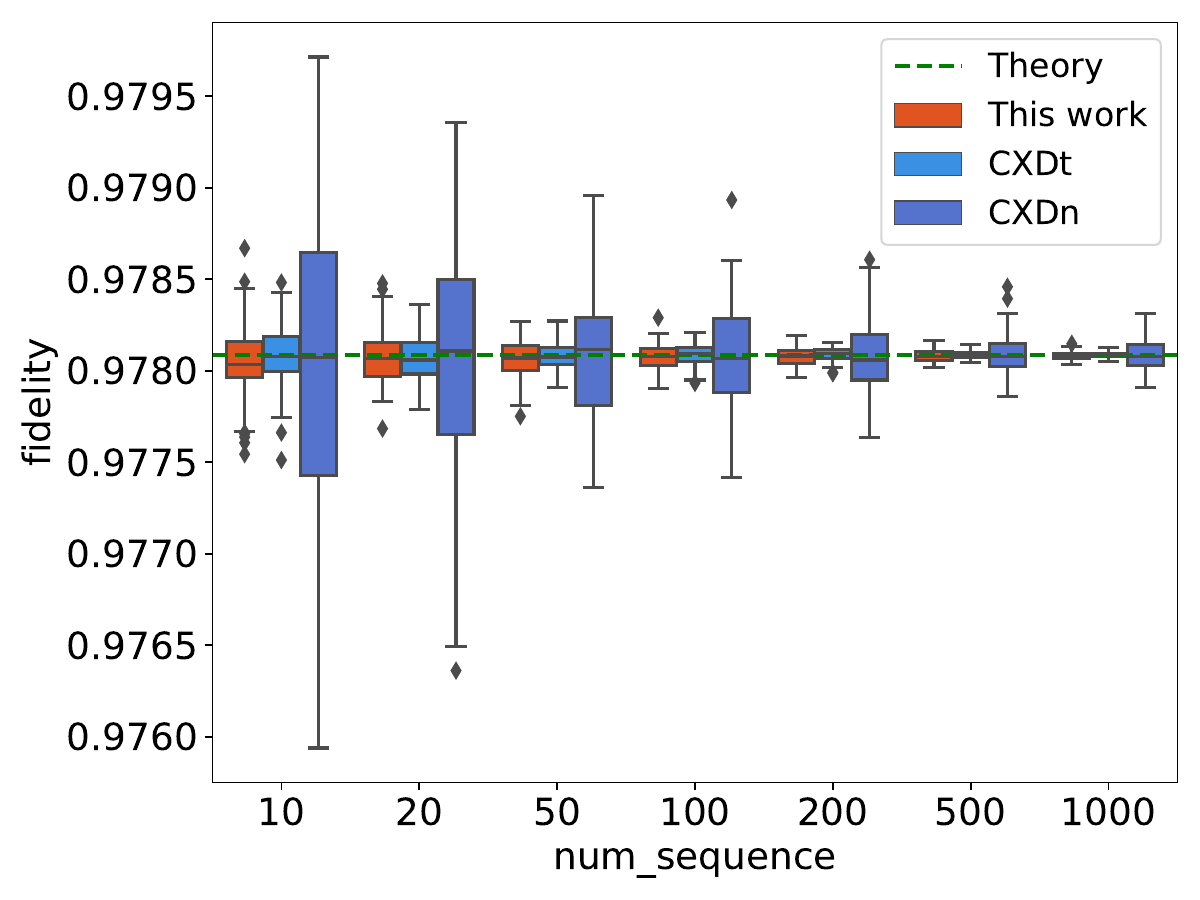}
\end{minipage}
}
\subfigure[$C^nZ$]{
\begin{minipage}[b]{0.48\textwidth}\label{fig:CnZ}
\includegraphics[scale=0.244]{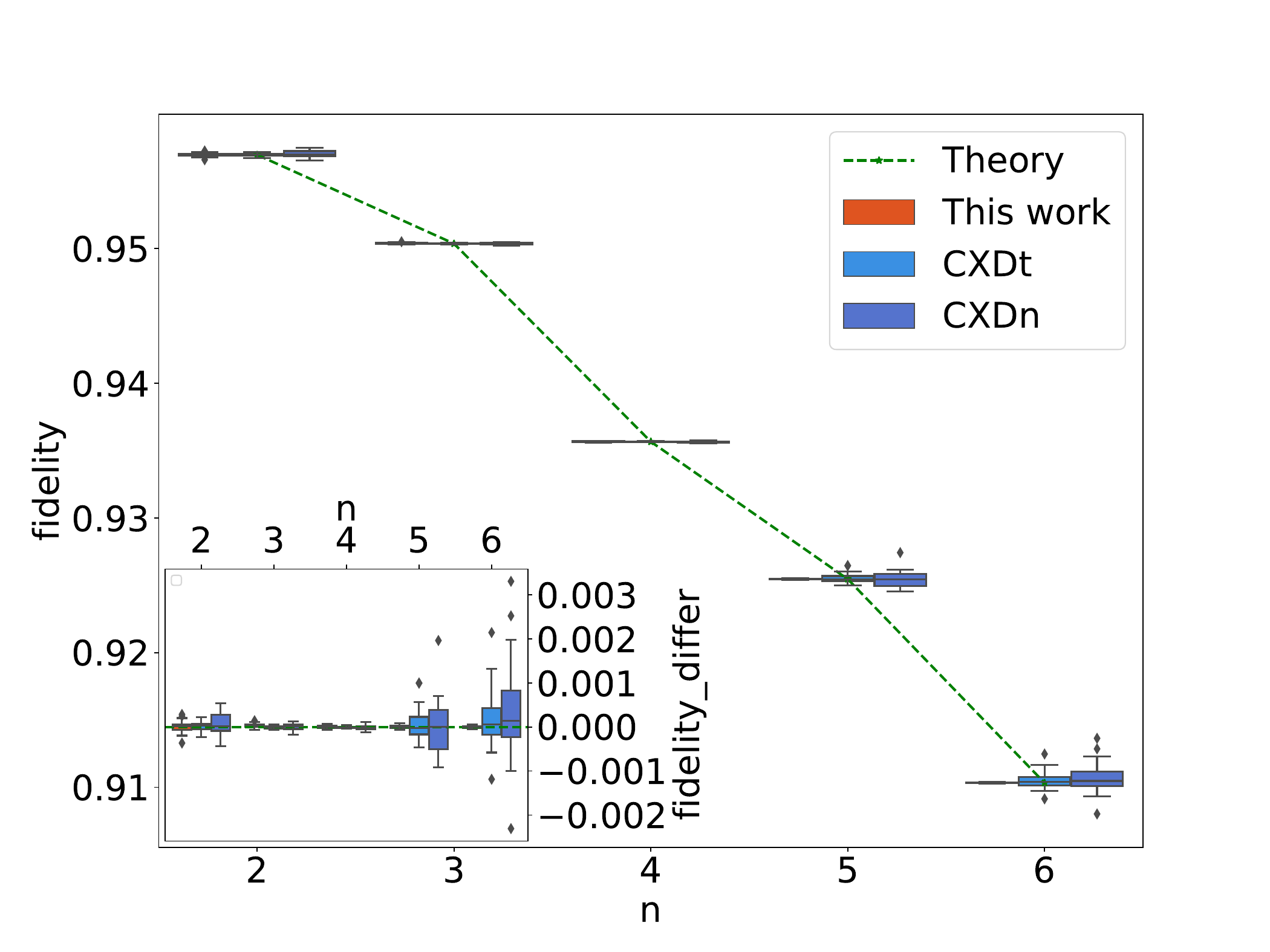}
\end{minipage}
}
\label{fig:simulation}
\caption{Benchmarking results for CCZ and $C^nZ$ with our schemes and CNOT dihedral group. There are two processing procedures for the CNOT dihedral group. The original procedure in~\cite{Cross2016dihedral} is denoted as CXDn, and the method, the same as our scheme, is denoted as CXDt. The green dashed line is the theoretical value of the noise channel fidelity. In (a), each box plot contains 100 fidelities, and each fidelity is estimated with circuit depths $\{3, 6,\cdots, 30\}$ and several sampled sequences specified by the horizontal axis. In (b), each box plot contains 20 fidelities, and each fidelity is estimated with circuit depths $\{1, 2, 5, 10, 20\}$, and 50 sampled sequences for each depth. The small figure in (b) demonstrates the difference between the estimated fidelities and the theoretical fidelity for three benchmarking procedures.}
\end{figure}

From Figure~\ref{fig:CCZ}, we can see that if using the SPAM settings~\ref{item:sm1} and~\ref{item:sm2}, the accuracy and the fluctuation of the estimated fidelities are almost the same for $\mathsf{G}^Z$ and $\mathsf{G}^X$, even $\mathsf{G}^Z$ is smaller and more easily implementable. Both two procedures are better than the original one~\cite{Cross2016dihedral}. Concerning the scalability, we found that $\mathsf{G}^Z$ performs better than $\mathsf{G}^X$ when $n=5$ and $n=6$ as shown in Figure~\ref{fig:CnZ}. An explanation is that as the number of qubits increases, the group size of $\mathsf{G}^X$ is too large. This results in an insufficient number of sampled sequences for $\mathsf{G}^X$ to twirl the noise channel, causing increased fluctuation in performance. The simulation emphasizes the importance of identifying smaller twirling groups and demonstrates the practicality of our findings. It is worth noting that our results can be extended to Toffoli gates, controlled-$\sqrt{Y}$ gates, and various useful non-diagonal gates via the technique of local gauge transformation~\cite{Zhang2023Scalable}.

\section{Outlook}\label{sc:outlook}
In this work, we study noise tailoring from a distinct perspective of finding suitable twirling groups and show nearly optimal results for multi-qubit controlled phase gates $C^nZ_m$. Besides $C^nZ_m$, exploring optimal twirling groups for other quantum gates and investigating the feasibility of solely employing local twirling gates is important to practical noise tailoring. While the exponential growth of optimal twirling groups poses a challenge for large-scale non-Clifford gate benchmarking, avenues for overcoming this limitation exist. Adding preconditions to the noise channel or trying circuit structure outside Figure~\ref{fig:RBmaintext} may allow one to find smaller twirling groups. Benchmarking without group structure, as discussed in the literature~\cite{Helsen2022RBFramework,Chen2022beyond,heinrich2023randomized,Gu2023CSB}, offers an alternative approach. Beyond noise tailoring, the properties of twirling groups explored here can also enhance studying protocols involving these groups, like shadow tomography.

\begin{acknowledgements}
We thank Boyang Chen, Zhuo Chen, Daojin Fan, and Shaowei Li for the helpful discussions. This work was supported by the National Natural Science Foundation of China Grant No.~12174216 and the Innovation Program for Quantum Science and Technology Grant No.~2021ZD0300804.
\end{acknowledgements}

\newpage
\appendix

\section{Preliminary}\label{appendsc:pre}
In this section, we introduce some related works and basic mathematical tools. Below, we first introduce basic notations. The Hilbert space for $n$ qubits is denoted as $\mathcal{H}$ and the set of linear operators on $\mathcal{H}$ is denoted as $\mathcal{L}(\mathcal{H})$. Same in the main text, we denote a quantum gate, or a unitary transformation acting on $\mathcal{H}$, in the standard representation with a capital letter, like $U$ and $G$. A set composed of quantum gates is denoted with sans serif fonts like $\mathsf{S}$. A quantum gate set satisfying group condition under the gate composition is always denoted as $\mathsf{G}$. Note that, in this work, we distinguish two concepts: the twirling gate set and the twirling group. The twirling group is a twirling gate set with a group structure. We normally use $\mathsf{V}$ and $\mathsf{G}$ to represent the twirling gate set and the twirling group, respectively. The Liouville representation of $U$ is denoted as its calligraphic letter, $\mathcal{U}$. In this paper, the Liouville representation refers to the Pauli-Liouville representation, which will be reviewed in detail in the following subsection. Quantum channels are defined as completely positive and trace-preserving (CPTP) linear maps on $\mathcal{L}(\mathcal{H})$. The noise channel of $U$ is normally denoted as $\Lambda$, and we use the same expression for the map representation and the Liouville representation for $\Lambda$. We use a wavy line to represent the noisy version of $U$ like $\widetilde{U}$ and $\widetilde{\mathcal{U}}$. Note that a quantum gate is always a quantum channel, but a quantum channel is generally not a quantum gate. We summarize the frequently-used notations in Table~\ref{table:Notation}. Note that we sometimes reuse part of notations, and the meanings of these notations are relevant to the context.

\begin{center}
\begin{table}[H]
\caption{Notation}\label{table:Notation}
\centering
\begin{tabular}{cc}\hline
Quantity &  Notation\\\hline
$n$ &  number of controlled qubit in $C^nZ_m$ or qubit number \\\hline
$m$ &  phase of controlled qubit in $C^nZ_m$ or circuit depth \\\hline
$N$ &  qubit number \\\hline
$\mathcal{H}$ &  Hilbert space\\\hline
$\mathcal{L}(\mathcal{H})$ & the set of linear operators on $\mathcal{H}$\\\hline
$U$ &  target gate \\\hline
$V$ &  twirling gate \\\hline
$G$ &  twirling gate in group $\mathsf{G}$ \\\hline
$\mathcal{U}$ &  Pauli Liouville representation of $U$ \\\hline
$\Lambda$ &  noise channel \\\hline
$\Lambda_G$ & $\mathsf{G}$-twirled noise channel \\\hline
$\widetilde{U}$ &  noisy version of $U$ \\\hline
$\mathsf{V}$ & twirling gate set \\\hline
$\mathsf{G}$ & twirling gate group \\\hline
$C^nZ_m = \begin{pmatrix}
\mathbb{I}_{2^{n+1}-1} & \mathbf{0} \\
\mathbf{0} & e^{i\frac{2\pi}{m}}
\end{pmatrix}$ & multi-qubit controlled phase gate \\\hline
$X_j$ &  Pauli $X$ on qubit $j$ \\\hline
$Y_j$ &  Pauli $Y$ on qubit $j$ \\\hline
$Z_j$ &  Pauli $Z$ on qubit $j$ \\\hline
$\mathbb{I}$ &  identity operator \\\hline
$I$ &  identity channel \\\hline
$\langle G_1, G_2,... \rangle$ & Group generated by $G_1, G_2,...$ \\\hline
$\textsf{P}_n$ &  $n$-qubit Pauli group $\langle X, Z\rangle^{\otimes n}$ \\\hline
$\textsf{D}_n^m$ &  $n$-qubit local dihedral group $\langle X, Z_m\rangle^{\otimes n}$ \\\hline
$\chi$ &  character function of a group \\\hline
$\Pi$ & projector or permutation matrix \\\hline
\end{tabular}
\end{table}
\end{center}

\subsection{Pauli-Liouville representation}\label{appendssc:PLrep}
Here, we introduce the Pauli-Liouville representation of quantum channels. Normally, a quantum channel, $\Lambda$, is defined within its Kraus representation. For any $O\in \mathcal{L}(\mathcal{H})$, the action of channel $\Lambda$ is defined with:
\begin{equation}
\Lambda(O)=\sum_{l=1}^k K_l O K_l^\dagger,
\end{equation}
where $k\in\mathbb{Z}_+$ and $\{K_l, 1\leq l\leq  k\}$ is the set of Kraus operators satisfying the condition:
\begin{equation}
\sum_{l=1}^k K_l^\dagger K_l=\mathbb{I},
\end{equation}
where $\mathbb{I}$ represents the identity operator.

The Kraus representation is not a matrix representation and, hence not convenient for our work. For further elaboration, we introduce the Liouville representation, which is defined on a set of trace-orthonormal basis elements in $\mathcal{L}(\mathcal{H})$. In this work, we choose the basis to be the normalized and projective Pauli group. In this case, the representation is named Pauli-Liouville representation. The $n$-qubit Pauli group, denoted as $\mathsf{P}_n$, is given by:
\begin{equation}
\mathsf{P}_n = \bigotimes_{j=1}^n \langle X, Z\rangle  = \langle X_1, Z_1, X_2, Z_2, \cdots, X_n, Z_n\rangle = \{\pm 1, \pm i\}\times \{\bigotimes_{j=1}^n P_j | P_j\in \{ \mathbb{I}, X, Y, Z\} \},
\end{equation}
where $\mathbb{I} = \begin{pmatrix}
1 & 0\\0 & 1
\end{pmatrix}$ is the identity operator with dimension 2. $X$, $Y$, and $Z$ are the single-qubit Pauli matrices given by
\begin{equation}
X = \begin{pmatrix}
0 & 1\\
1 & 0
\end{pmatrix},
\quad
Y = \begin{pmatrix}
0 & -i\\
i & 0
\end{pmatrix},
\quad
Z = \begin{pmatrix}
1 & 0\\
0 & -1
\end{pmatrix}.
\end{equation}
$X_j$ is the single-qubit Pauli $X$ matrix acting on $j$-th qubit given by $(\bigotimes_{i=1}^{j-1}\mathbb{I}_i) \otimes X (\bigotimes_{i=j+1}^{n}\mathbb{I}_i)$.
The notation $\langle \cdot \rangle$ denotes the group generated by $\cdot$. The elements in $\langle \rangle$ are called group generators. For instance, $\langle X_j, Z_j\rangle = \{X_j^{m_1} Z_j^{n_1} X_j^{m_2} Z_j^{n_2}\cdots | \forall i, m_i, n_i\in \mathbb{Z}\} = \{\pm 1, \pm i\}\times \{\mathbb{I}_j,X_j,Y_j,Z_j\}$. In the main text and below, for brevity we sometimes use notation $\langle X, Z\rangle^{\otimes n}$ or even $\langle X, Z\rangle$ to represent $\langle X_1, Z_1, X_2, Z_2, \cdots, X_n, Z_n\rangle$ where in this case, $X$ and $Z$ denote Pauli $X$ and Pauli $Z$ gates on all qubits, respectively. We would always omit subscripts, representing which qubits the gates act on, in $\langle \rangle$ when the group generators contain the same gates acting on all qubits.

In literature, people always use another definition for the Pauli group, which is, in fact, the projective Pauli group. Projective Pauli group is the quotient of the Pauli group by its center $\{\pm 1, \pm i\}$, $\langle X_1, Z_1, X_2, Z_2, \cdots, X_n, Z_n\rangle / \{\pm 1, \pm i\}$. Thus, we normally use $\bigotimes_{j=1}^n \{\mathbb{I},X,Y,Z\}$ to represent it. As in quantum systems, the overall phase is not important. The quotient by the phase $\{\pm 1, \pm i\}$ does not influence the quantum state. Below we simply use the Pauli group to denote the projective Pauli group and use the definition below.
\begin{equation}
\mathsf{P}_n = \bigotimes_{j=1}^n \{\mathbb{I}, X, Y, Z\}.
\end{equation}

The normalized Pauli group is obtained by normalizing the elements of the $n$-qubit Pauli group by a factor of $1/\sqrt{2^n}$ as shown below.
\begin{equation}
\mathsf{P}'_n = \{\sigma_{i} = \frac{1}{\sqrt{2^n}} P_{i} | P_{i}\in \mathsf{P}_n\}.
\end{equation}
These normalized Pauli operators are complete and satisfy the orthonormality under the Hilbert-Schmidt inner product,
\begin{equation}
\tr(\sigma_i^\dagger \sigma_j) = \delta_{ij},
\end{equation}
where $\delta_{ij}$ is the Kronecker delta symbol. Thus, we can represent the quantum state and the quantum channel with normalized Pauli operators as a set of bases.

Any $n$-qubit operator $O$ in $\mathcal{L}(\mathcal{H})$ can be decomposed with the $4^n$ normalized Pauli operators,
\begin{equation}\label{eq:pauliexpand}
O =\sum_{\sigma_i\in\mathsf{P}'_n}\tr(\sigma_i^\dagger O)\sigma_i.
\end{equation}
Thus, we can use the expansion coefficients $\tr(\sigma_i^\dagger O)$ to represent $O$ and in Pauli-Liouville representation, we express Eq.~\eqref{eq:pauliexpand} as
\begin{equation}
\lket{O} =\sum_{\sigma_i\in\mathsf{P}'_n}\tr(\sigma_i^\dagger O)\lket{\sigma_i},
\end{equation}
where $\lket{\sigma_i}$ in Pauli-Liouville representation is a length-$4^n$ vector with only one non-zero element 1.

Furthermore, any quantum channel $\Lambda$ can be represented as a matrix in the Liouville representation. The action of the channel on an operator $O$ is given by
\begin{equation}
\lket{\Lambda(O)}=\Lambda\lket{O},
\end{equation}
and the elements of the Liouville representation of $\Lambda$ are given by
\begin{equation}
\Lambda_{ij}=\lbra{\sigma_i}\Lambda\lket{\sigma_j}=\tr(\sigma_i\Lambda(\sigma_j)).
\end{equation}
Specifically, the diagonal terms of Pauli-Liouville representation are named Pauli fidelity~\cite{Zhang2023Scalable}, defined as below.
\begin{equation}
\lambda_i = \lbra{\sigma_i}\Lambda\lket{\sigma_i}=\frac{1}{d} \tr(P_i \Lambda(P_i)),
\end{equation}
where $P_i$ is a Pauli operator, and $d$ is the dimension of the Hilbert space.

In the Liouville representation, the concatenation of two channels can be represented as the product of their corresponding matrices,
\begin{equation}
\lket{\Lambda_2\circ\Lambda_1(O)}=\Lambda_2\lket{\Lambda_1(O)}=\Lambda_2\Lambda_1\lket{O}.
\end{equation}
The Liouville representation also allows us to vectorize the measurement operators using the Liouville bra-notation. The measurement probability of a state $\rho$ using a positive operator-valued measure (POVM) ${F_i}$ is given by
\begin{equation}
p_i=\lbraket{F_i}{\rho}=\tr(F_i^\dagger\rho),
\end{equation}
where $\lbra{F_i}$ is the Liouville bra-notation of $F_i$ and is equal to the conjugate transpose of $\lket{F_i}$.

We also briefly introduce another representation called the $\chi$-matrix representation for quantum channel $\Lambda$. Given an $n$-qubit state, $\rho$, $\Lambda(\rho)$ can be expanded as
\begin{equation}
\Lambda(\rho) = d \sum_{i, j} \chi_{ij} \sigma_i \rho \sigma^\dag_j,
\end{equation}
where the process matrix $\chi$ is uniquely determined by the orthonormal Pauli operator basis ${\sigma_j}$ and $d = 2^n$ represents the dimension of the quantum system. Note that $\sigma_0 = \frac{\mathbb{I}_{d}}{\sqrt{d}}$. If a channel is diagonal in this representation, it is called a Pauli channel. It is easy to verify that a Pauli channel is also diagonal in Pauli-Liouville representation. The diagonal terms of Pauli-Liouville representation and that of $\chi$-matrix representation are related by Walsh-Hadamard transformation as shown in the following equation.
\begin{equation}\label{eq:WalshHtrans}
\lambda_j = \sum_i (-1)^{\langle i, j \rangle }\chi_{ii},
\end{equation}
where $\langle i, j \rangle = 0$ if $P_i$ commutes with $P_j$, and $\langle i, j \rangle = 1$ otherwise. The inverse Walsh-Hadamard transformation is given by,
\begin{equation}\label{eq:invWalshHtrans}
\chi_{jj} = \frac{1}{d^2}\sum_i (-1)^{\langle i, j \rangle } \lambda_i.
\end{equation}

\subsection{Quantum channel fidelity}
To facilitate the understanding of randomized benchmarking (RB), we introduce the concepts of process fidelity and average fidelity. These two fidelities are equivalent to each other by linear transformation. The task of RB is estimating the fidelity of a given quantum gate or a given gate set.

The process fidelity of a channel, $\Lambda$, can be defined with its $\chi$-matrix representation as follows.
\begin{equation}
F(\Lambda) = \chi_{00}(\Lambda).
\end{equation}
The process fidelity can be obtained from Eq.~\eqref{eq:invWalshHtrans} as:
\begin{equation}\label{eq:processfidelity2}
F(\Lambda) = \chi_{00}(\Lambda) = \frac{1}{d^2} \sum_j \lambda_j = \frac{1}{d^2} \tr(\Lambda),
\end{equation}
which is equal to the trace of $\Lambda$ in Liouville representation divided by $d^2$. Eq.~\eqref{eq:processfidelity2} provides an alternative definition of process fidelity.

There exists a relation between the commonly-used average fidelity $F_{ave}$ and the process fidelity~\cite{horodecki1999general},
\begin{equation}
F_{ave} = \frac{dF+1}{d+1}.
\end{equation}
The average fidelity is defined as:
\begin{equation}
F_{ave} = \int d\psi \tr(\ketbra{\psi} \Lambda(\ketbra{\psi})),
\end{equation}
where the integral is taken over the Haar measure, it means the average fidelity of the ideal final state and the realistic final state over all pure states. Both the process fidelity and the average fidelity are well-defined metrics for quantifying the closeness of a quantum channel to the identity. Below, we refer to the fidelity to the process fidelity without further explanation. Also, note that in reality, a quantum gate, $U$, is noisy, and its noisy version can be expressed as the composite channel of $\mathcal{U}$ and its noise channel $\Lambda$. Then, the fidelity of $U$ refers to the fidelity of its noise channel $\Lambda$.

\subsection{Group twirling and representation theory}\label{sec:representation}
Representation theory is an effective tool for analysing abstract groups and is especially useful in randomized benchmarking and random compiling. Below we briefly introduce basic concepts in representation theory and refer to~\cite{steinberg2012representation,fultonRepresentation} for systematic introduction. Let $\mathsf{G}$ be a finite group and $G \in \mathsf{G}$ be an element of the group. The representation of $\mathsf{G}$ is defined as follows.

\begin{definition}[Group representation] \label{def:grouprep}
A map, $\phi$, is called a representation of the group $\mathsf{G}$ on a linear space, $V$, if it is a group homomorphism mapping from $\mathsf{G}$ to $GL(V)$, where
\begin{equation}
\begin{split}
&\phi:\mathsf{G} \rightarrow GL(V), \\
&g \mapsto \phi(G), \ \forall G \in \mathsf{G}.
\end{split}
\end{equation}
Here, $GL(V)$ denotes the general linear group of $V$. The representation $\phi$ satisfies the condition of preserving multiplication that for all $G_1, G_2 \in \mathsf{G}$,
\begin{equation}
\phi(G_1)\phi(G_2) = \phi(G_1 G_2).
\end{equation}
\end{definition}

Intuitively, representation is just using a matrix group to represent $\mathsf{G}$. Below we introduce the concept of irreducible representation. Given a representation, $\phi$, on $V$, a linear subspace, $W \subseteq V$, is referred to as invariant if for all $w \in W$ and for all $G \in \mathsf{G}$,
\begin{equation}
\phi(G)w \in W.
\end{equation}
The restriction of $\phi$ to the invariant subspace $W$ is called the subrepresentation of $\mathsf{G}$ on $W$. Note that any representation has a subrepresentation mapping all elements in $\mathsf{G}$ to $1$ where $W=\{0\}$ and has itself as a subrepresentation where $W=V$. These two subrepresentations are both trivial. With the concept of subrepresentation, we define the irreducible representation as below.

\begin{definition}[Irreducible representation]
A representation $\phi$ of the group $\mathsf{G}$ on the linear space $V$ is said to be irreducible if it only possesses trivial subrepresentations.
\end{definition}

Maschke's theorem provides an interesting property stating that every representation $\phi$ of a finite group, $\mathsf{G}$, can be decomposed into the direct sum of irreducible representations. For all $G \in \mathsf{G}$, the decomposition can be expressed as follows.
\begin{equation}\label{eq:directsum}
\phi(G) \simeq \bigoplus_{\phi_i\in R_{\mathsf{G}}} \mathbb{I}_{n_{i}\times n_{i}}\otimes\phi_i(G),
\end{equation}
Here, $R_{\mathsf{G}}$ denotes the set of all irreducible representations, and $n_i$ represents the multiplicity of the equivalent irreducible representations of $\phi_i$ in $\phi$.

The trace of a representation is called a character, which is defined below.
\begin{definition}[Character function]
Let $\phi$ be a representation over the group $\mathsf{G}$. The character of $\phi$ is defined as the function $\chi_{\phi}: \mathsf{G} \rightarrow \mathbb{C}$ such that for every $G \in \mathsf{G}$,
\begin{equation}
\chi_{\phi}(G) = \tr[\phi(G)].
\end{equation}
\end{definition}

Character function is a powerful tool in representation theory. We can define the inner product of two character functions as below.

\begin{definition}[Inner product of character function]
Let $\chi_1$ and $\chi_2$ be two character functions of a finite group, $\mathsf{G}$, their inner product is defined as
\begin{equation}
\braket{\chi_1}{\chi_2} = \frac{1}{\abs{\mathsf{G}}}\sum_{G\in\mathsf{G}}\chi_1^*(G)\chi_2(G).
\end{equation}
\end{definition}

There is a useful result for the inner product between two characters when $\chi_1$ and $\chi_2$ are both irreducible representations. Then $\braket{\chi_1}{\chi_2} = 1$ if $\chi_1$ and $\chi_2$ are equivalent and $\braket{\chi_1}{\chi_2} = 0$ otherwise. Note that Eq.~\eqref{eq:directsum} tells us that
\begin{equation}
\chi_{\phi}(G) = \tr(\phi(G)) = \sum_{\phi_{i}\in R_{\mathsf{G}}} n_i\chi_i(G),
\end{equation}
where $\chi_i(G) = \tr(\phi_i(G))$. With the orthonormality of irreducible representations, we obtain $n_i = \braket{\chi_i}{\chi_{\phi}}$. Then, we have the following lemma.

\begin{lemma}[Multiplicity of irreducible representations]\label{lemma:multiplicity}
For any irreducible representation $\phi_i$ with character $\chi_i$, the multiplicity of $\phi_i$ in $\phi$ is given by
\begin{equation}
n_i = \braket{\chi_j}{\chi_{\phi}}.
\end{equation}
\end{lemma}

We also introduce a concept named multiplicity-free representation for further elaboration.

\begin{definition}[Multiplicity-free representation]\label{def:multiplicityfree}
Given a representation, $\phi$, of group $\mathsf{G}$, if for any irreducible representation, $\phi_i$, the multiplicity of $\phi_i$ in $\phi$ is 1, we call $\phi$ a multiplicity-free representation.
\end{definition}

Using the character function, we can also obtain the generalized projection formula utilized in character randomized benchmarking \cite{Helsen2019characterRB}.

\begin{lemma}[Generalized projection formula] \label{lemma:project_formula}
Consider a finite group, $\mathsf{G}$, and its representation $\phi$. Let $\phi_i$ be an irreducible representation contained in $\phi$ with character $\chi_i$. The projector onto the support space of $\phi_i$ can be expressed as:
\begin{equation}\label{eq:project_formula}
\Pi_i = \frac{d_i}{|\mathsf{G}|}\sum_{G \in \mathsf{G}} \chi_i(G) \phi(G),
\end{equation}
where $d_i = \dim \phi_i$ denotes the dimension of $\phi_i$.
\end{lemma}

The representation theory is also useful for us to analyze the group twirling on a channel, $\Lambda$, over a group, $\mathsf{G}$.

\begin{definition}[Group Twirling]
For a representation $\phi$ of the group $\mathsf{G}$ on the linear space $V$, the twirling of a linear map, $\Lambda: V \rightarrow V$ over $\mathsf{G}$ is defined as:
\begin{equation}
\Lambda_{\mathsf{G}} = \frac{1}{\abs{\mathsf{G}}} \sum_{G\in\mathsf{G}} \phi(G)^\dagger\Lambda\phi(G).
\end{equation}
\end{definition}

To analyze the result of group twirling, we introduce Schur's lemma, which is essential in our further elaboration.

\begin{lemma}[Schur's lemma]
Let $\phi_1:\mathsf{G}\rightarrow GL(V_1)$ and $\phi_2:\mathsf{G}\rightarrow GL(V_2)$ be two arbitrary irreducible representations of group $\mathsf{G}$ and $A:V_1\rightarrow V_2$ be a linear map from $V_1$ to $V_2$. Suppose for any $G\in\mathsf{G}$,
\begin{equation}
A\phi_1(G) = \phi_2(G)A.
\end{equation}
Then, $A$ equals $0$, mapping all vectors in $V_1$ to zero vector in $V_2$, if $\phi_1$ and $\phi_2$ are inequivalent irreducible representations; $A$ is proportional to the identity operator if $\phi_1$ and $\phi_2$ are equivalent irreducible representations. Note that when $\phi_1$ and $\phi_2$ are equivalent, $V_1$ and $V_2$ are equivalent linear space, so the identity operator is well-defined.
\end{lemma}

Using Schur's Lemma, we can establish the following proposition.
\begin{proposition}\label{prop:twirling}
Let $\phi_1:\mathsf{G}\rightarrow GL(V_1)$ and $\phi_2:\mathsf{G}\rightarrow GL(V_2)$ be two arbitrary irreducible representations of group $\mathsf{G}$ and $A:V_1\rightarrow V_2$ be an arbitrary linear map from $V_1$ to $V_2$. Define
\begin{equation}
A_G = \mathbb{E}_{G\in \mathsf{G}} \phi_2^{\dagger}(G) A \phi_1 (G).
\end{equation}
Then,
\begin{equation}
A_G = 0,
\end{equation}
if $\phi_1$ and $\phi_2$ are inequivalent, and
\begin{equation}
A_G = \frac{\tr(A)}{\dim V_1} \mathbb{I},
\end{equation}
if $\phi_1$ and $\phi_2$ are equivalent, where $\mathbb{I}$ is the identity operator mapping from $V_1$ to $V_2$.
\end{proposition}

For further elaboration and proof, we provide the Burnside theorem in the representation theory.

\begin{lemma}\label{lemma:Burnsidetheorem}
Consider a finite group, $\mathsf{G}$, and its all inequivalent irreducible representations $R_G=\{\phi_i\}$, then the cardinality of the twirling group, $\abs{\mathsf{G}}$ can be given by
\begin{equation}
\abs{\mathsf{G}} = \sum_{\phi_i\in R_G} \dim \phi_i^2.
\end{equation}
\end{lemma}

\subsection{Randomized benchmarking and diagonalizability of twirled noise channel}\label{appendssc:RB}
Below, we introduce randomized benchmarking (RB) and character randomized benchmarking~\cite{Helsen2019characterRB}. Let us start with a brief review of the standard RB, which aims to estimate the fidelity of a quantum gate group. We consider an $n$-qubit gate group, $\mathsf{G}$, where each gate $G\in \mathsf{G}$ is assumed to be with a noise channel $\Lambda$ in implementation. In reality, the noise channels for different gates can be different. Below, we simply consider the first-order approximation~\cite{Emerson2012praRB}, which is valid and useful in experiments, and regard the noise channels for different gates in $\mathsf{G}$ as the same. One can express the noisy quantum gate $\widetilde{\mathcal{G}} = \Lambda\mathcal{G}$ as the composite of the noiseless gate $G$ and its noise channel $\Lambda$ for any gate $G\in\mathsf{G}$. Then, the task of RB is estimating the fidelity of $\Lambda$ robust to SPAM error.

To achieve that, one always performs random gate sequences composed of $m$ random gates from $\mathsf{G}$ and an inverse quantum gate where $m$ is a prefixed integer called circuit depth,
\begin{equation}\label{eq:randomsequence}
\begin{split}
\widetilde{S} &= \widetilde{\mathcal{G}}_{inv} \prod_{i=1}^{m} \widetilde{\mathcal{G}}_i\\
&= \Lambda \mathcal{G}_{inv} \prod_{i=1}^{m} \Lambda \mathcal{G}_i,
\end{split}
\end{equation}
where $G_{inv} = (\prod_{i=1}^m G_i)^{-1}$, $\mathcal{G}_i$ denotes the Liouville representation of $G_i$, and $\widetilde{\cdot}$ means the quantum gate is noisy.

In the sense of expectation, the random gate sequence $\widetilde{S}$ is equal to
\begin{equation}
\begin{split}
\mathbb{E}_{\forall i, G_i\in\mathsf{G}} \widetilde{S} &= \Lambda \mathbb{E}_{\forall i, G_i\in\mathsf{G}} \prod_{i=1}^{m} (\prod_{j=1}^i\mathcal{G}_j)^{\dagger} \Lambda (\prod_{j=1}^i\mathcal{G}_j)\\
&=\Lambda (\mathbb{E}_{G\in \mathsf{G}}\mathcal{G}^{\dagger}\Lambda \mathcal{G})^m\\
&= \Lambda\Lambda_G^m.
\end{split}
\end{equation}
The second line utilizes the fact that $\mathsf{G}$ is a group. Thus, $\prod_{j=1}^i\mathcal{G}_j$ independently
and identically satisfy the uniform distribution on $\mathsf{G}$ for any $1\leq j\leq m$.

If we use the same state preparation and measurement for all random sequences, then in the expectation, we can get
\begin{equation}
\begin{split}
f(m) &= \mathbb{E}_{\forall i, G_i\in\mathsf{G}} \lbra{M} \widetilde{S} \lket{\rho}\\
&= \lbra{M}\Lambda\Lambda_G^m \lket{\rho}\\
&= \lbra{M'}\Lambda_G^m \lket{\rho},
\end{split}
\end{equation}
where $\rho$ and $M$ are prefixed initial state and measurement, respectively, and $M'$ is the measurement absorbing the noise of the inverse gate, satisfying $\lbra{M'} = \lbra{M}\Lambda$.

The next step is obtaining the trace of $\Lambda_G$ robust to state preparation and measurement (SPAM) error. Note that the fidelity $F(\Lambda) = \frac{1}{d^2}\tr(\Lambda)=\frac{1}{d^2}\tr(\Lambda_G)$. Normally, $\Lambda_G$ has few parameters as it has been twirled and becomes symmetric. These parameters can be evaluated robustly to SPAM error via obtaining $f(m)$ with different circuit depths $m$ and employing exponential fitting. Then, the trace of $\Lambda_G$, or the fidelity, can be evaluated with the parameters of $\Lambda_G$.

Before the work of character RB~\cite{Helsen2019characterRB}, researchers mainly focus on a large group, $G$, like the Clifford group, to make $\Lambda_G$ highly symmetric with few parameters. For Clifford RB~\cite{Emerson2011prlRB,Emerson2012praRB}, $\Lambda_G$ is a depolarizing channel with a single parameter $p$ satisfying $\Lambda_G(\rho) = p\rho + (1-p)\frac{\mathbb{I}_d}{d}$. In Pauli-Liouville representation, $\Lambda_G = \lketbra{\sigma_0}{\sigma_0} + p\sum_{i\geq 1} \lketbra{\sigma_i}{\sigma_i}$. Then, $f(m)$ has a single exponential decay expression
\begin{equation}
f(m) = Ap^m + B,
\end{equation}
where $p$ is only relevant to $\Lambda$, and $A$ and $B$ are only relevant to SPAM. Thus, via single exponential fitting, one can get $p$ along with the trace of $\Lambda_G$ and the fidelity of $\Lambda$.

For CNOT dihedral RB~\cite{Cross2016dihedral}, $\Lambda_G = \Pi_0 + p_Z \Pi_Z + p_X \Pi_X$ has two undetermined parameters, $p_Z$ and $p_X$ where $\Pi_0 = \lketbra{\sigma_0}{\sigma_0}$, $\Pi_Z = \sum_{\sigma_z \in \{\frac{\mathbb{I}}{\sqrt{2}},\frac{Z}{\sqrt{2}}\}^{\otimes n}/\sigma_0} \lketbra{\sigma_z}{\sigma_z}$, and $\Pi_X = \sum_{\sigma_x\in \mathsf{P}'_n/\{\frac{\mathbb{I}}{\sqrt{2}},\frac{Z}{\sqrt{2}}\}^{\otimes n}}\lketbra{\sigma_x}{\sigma_x}$, then $f(m)$ is a double exponential decay function with parameters $p_Z$ and $p_X$. One has to employ double exponential fitting to obtain $\tr(\Lambda_G)$ and its fidelity.

For a general group $\mathsf{G}$, $f(m)$ is complex and generally not an exponential decay function~\cite{Helsen2022RBFramework}. But if $\Lambda_G$ can be written as
\begin{equation}\label{eq:projectordecompose}
\Lambda_G = \sum_{i} p_i \Pi_i,
\end{equation}
where $\{\Pi_i\}$ are mutually orthogonal projectors only dependent on $\mathsf{G}$ in the space of Liouville representation, we can obtain all decay parameters $p_i$ with only single exponential fitting via the technique of character RB~\cite{Helsen2019characterRB}. Note that Eq.~\eqref{eq:projectordecompose} can also be interpreted as that $\Lambda_G$ is diagonal in the Pauli-Liouville representation up to a unitary transformation $\mathcal{T}$,
\begin{equation}\label{eq:projectordecompose2}
\Lambda_G = \mathcal{T} (\sum_i p_i \Pi^P_i)
\mathcal{T}^{\dagger},
\end{equation}
where $\Pi^P_i$ is a projector, equaling the summation of the Pauli operator bases. The projector $\Pi^P_i$ and the unitary transformation $\mathcal{T}$ are independent of the channel $\Lambda$ and are only related to group $\textsf{G}$. We consider the case that $\mathcal{T}$ is a Liouville representation of a unitary gate $T$.

Note that in this case,
\begin{equation}\label{eq:survprob}
\begin{split}
f(m) &= \lbra{M'}\Lambda_G^m \lket{\rho}\\
&= \sum_i \lbra{M'}\Pi_i \lket{\rho} p_i^m,
\end{split}
\end{equation}
is a multiple exponential decay function. If we directly fit $f(m)$ with a multiple exponential decay function, $\sum_i A_ip_i^m$, the fitting process would consume massive computational resources, and the result is normally inaccurate. The key step of character RB is utilizing Lemma~\ref{lemma:project_formula}. Instead of implementing the gate sequence in Eq.~\eqref{eq:randomsequence}, in character RB, we select a projector $\Pi_j$ from $\{\Pi_i\}$ and implement
\begin{equation}
\begin{split}
\widetilde{S}' &= \widetilde{\mathcal{G}}_{inv} (\prod_{i=1}^{m} \widetilde{\mathcal{G}}_i) \widetilde{\mathcal{G}'}\\
&= \Lambda \mathcal{G}_{inv} (\prod_{i=1}^{m} \Lambda \mathcal{G}_i) \mathcal{G}',
\end{split}
\end{equation}
where $G'$ is named character gate and is independently sampled from a predetermined gate set, named character group, $\mathsf{G}'$. We select a character function $\chi'$ associated with an irreducible representation $\phi'$ of $\mathsf{G}'$ and obtain a projector according to Lemma~\ref{lemma:project_formula},
\begin{equation}
\Pi' = \frac{\dim \phi'}{\abs{\mathsf{G}'}}\sum_{G'\in \mathsf{G}'}\chi'(G') \mathcal{G}',
\end{equation}
such that for any projector $\Pi_i\in \{\Pi_i\}$,
\begin{equation}\label{eq:projection}
\Pi'\Pi_i = \delta_{ij}\Pi'.
\end{equation}

Meanwhile, we consider $\lbra{M} \widetilde{S}' \dim \phi'\chi'(G') \lket{\rho}$ and in the expectation, this quantity equals
\begin{equation}
\begin{split}
f_j(m) &= \mathbb{E}_{\forall i, G_i\in\mathsf{G}, G'\in \mathsf{G}'} \lbra{M} \widetilde{S}' \dim \phi'\chi'(G') \lket{\rho}\\
&= \lbra{M}\Lambda\Lambda_G^m \Pi'\lket{\rho}\\
&= \lbra{M'}\Lambda_G^m \Pi'\lket{\rho}\\
&= \sum_i p_i^m \lbra{M'} \Pi_i\Pi'\lket{\rho}\\
&= \lbra{M'} \Pi'\lket{\rho} p_j^m.\\
\end{split}
\end{equation}
Thus, $p_j$ can be obtained via single exponential fitting and due to the arbitrariness of $\Pi_j$, all parameters of $\Lambda_G$ can be obtained via single exponential fitting, and hence the fidelity can be evaluated accurately.

Note that in character RB, the key point is realizing projector $\Pi'$. Below we show that we can always realize $\Pi'$ satisfying Eq.~\eqref{eq:projection}. First, each projector can be decomposed as
\begin{equation}
\Pi_j = \sum_{k=1}^{\tr(\Pi_j)} \lketbra{\mu_{j_k}}{\mu_{j_k}},
\end{equation}
where $\{\mu_{j_k}, 1\leq k\leq \tr(\Pi_j), \forall \Pi_j\in \{\Pi_i\}\}$ forms an orthonormal operator basis in Liouville representation. As normalized Pauli operators $\mathsf{P}_n'$ is also an orthonormal basis, there exists a unitary transformation linking the two bases. Then, given $\Pi_j$ from $\{\Pi_i\}$, we select $\lketbra{\mu_{j_1}}{\mu_{j_1}}$, which is equal to $\mathcal{T}\lketbra{\sigma}{\sigma}\mathcal{T}^{\dagger} = \lketbra{T\sigma T^{\dagger}}{T\sigma T^{\dagger}}$ where $\lketbra{\sigma}{\sigma}$ is a Pauli operator basis and $T$ is the unitary transformation linking two bases. For Pauli group $\mathsf{P}_n$, we have the following equation,
\begin{equation}
\lketbra{\sigma}{\sigma} =  \mathbb{E}_{P\in \mathsf{P}_n} \chi_{\sigma}(P)\mathcal{P},
\end{equation}
where $\chi_{\sigma} = (-1)^{\langle P, \sigma \rangle}$ equals $1$ when $P$ and $\sigma$ commute and $-1$ otherwise. Then, we only need to choose $\mathsf{G}' = T\mathsf{P}_n T^{\dagger}$ and $\chi'=\chi_{\sigma}$ to realize
\begin{equation}
\Pi' = \mathcal{T}\lketbra{\sigma}{\sigma}\mathcal{T}^{\dagger} = \mathbb{E}_{P\in \mathsf{P}_n} \chi'(P)\mathcal{T}\mathcal{P}\mathcal{T}^{\dagger}.
\end{equation}
Then, $\Pi'$ satisfies Eq.~\eqref{eq:projection}.

In summary, as long as $\Lambda_G$ has an expression of Eq.~\eqref{eq:projectordecompose} or Eq.~\eqref{eq:projectordecompose2}, then with the technique of character RB, one can obtain $\tr(\Lambda_G)$ and the fidelity of $\Lambda$ accurately with only single exponential fitting. In Appendix~\ref{appendsc:group}, we would prove that if the Liouville representation of $\mathsf{G}$ is not multiplicity-free, as defined in Definition~\ref{def:multiplicityfree}, then $\Lambda_G$ cannot be diagonal for \textit{arbitrary} noise channel $\Lambda$. If the Liouville representation of $\mathsf{G}$ is multiplicity-free, then~\cite{Helsen2019characterRB} has shown that the character group $\mathsf{G}'$ can be chosen as a subgroup of $\mathsf{G}$. In this case, each projector $\Pi_i$ in Eq.~\eqref{eq:projectordecompose} relates to an irreducible representation subspace of $\mathsf{G}$. Denote this irreducible representation as $\phi_i$ with character $\chi_i$, then $\Pi_i$ can be realized with
\begin{equation}
\Pi_i = \frac{\dim \phi_i}{\abs{\mathsf{G}}}\sum_{G'\in \mathsf{G}}\chi_i(G') \mathcal{G}'.
\end{equation}
Thus, implementing character gate $G'$ is not harder than the twirling group.

Besides adding character gates, there is another method to effectively realize the projector $\Pi_i$ in Eq.~\eqref{eq:projectordecompose}. If one can realize a measurement $M$ such that $\sum_i p_i\lbra{M'}\Pi_i = p_j\lbra{M'}\Pi_j$, then Eq.~\eqref{eq:survprob} will also reduce to $p_j^m\lbra{M'}\Pi_j\lket{\rho}$. An accurate initial state $\rho$, satisfying $\sum_i p_i\Pi_i\lket{\rho} = p_j\Pi_j\lket{\rho}$, can also achieve that. Thus, if we have some information about SPAM, we can realize the projection without the need to implement character gates.

\subsection{Interleaved randomized benchmarking}\label{appendssc:interleavedRB}
Above, we only introduce how to evaluate the fidelity of a quantum gate group $\mathsf{G}$ via randomized benchmarking. In order to obtain the gate fidelity of an individual target gate, $U$, one needs to utilize the technique of interleaved RB~\cite{Magesan2012interleavedRB}. In~\cite{Magesan2012interleavedRB}, the target gate $U$ is embedded into a group, $\mathsf{G}$. Below we call them the twirled gate and the twirling group, respectively. To enable interleaved RB, one implements two kinds of circuits. The first is just a random gate sequence in regular RB, extracting the average gate fidelity of the twirling group $\mathsf{G}$. The second type of circuit is composed of random twirling gates from $\mathsf{G}$ interleaved with the target gate $U$, extracting the composition gate fidelity of the twirling group and the target gate. Comparing the two results, one can get the individual gate fidelity of the target gate. Specifically, suppose the noise channel for twirling group $\mathsf{G}$ is $\mathcal{E}$ and the noise channel for the target gate is $\Lambda$. Then with the RB methods introduced before one can obtain $F(\mathcal{E}) = \frac{\tr(\mathcal{E})}{d^2}$. After that, one implements an interleaved random gate sequence,
\begin{equation}\label{eq:interleavedrandomsequence}
\begin{split}
\widetilde{S}_i &= \widetilde{\mathcal{G}}_{inv} \prod_{i=1}^{m} \widetilde{\mathcal{U}}\widetilde{\mathcal{G}}_i\\
&= \mathcal{E} \mathcal{G}_{inv} \prod_{i=1}^{m} \mathcal{U} \Lambda \mathcal{E} \mathcal{G}_i,
\end{split}
\end{equation}
where $G_{inv} = (\prod_{i=1}^m UG_i)^{-1}$, $\mathcal{G}_i$ and $\mathcal{U}$ denote the Liouville representation of $G_i$ and $U$, respectively, and $\widetilde{\cdot}$ means the quantum gate is noisy. Note that $\widetilde{\mathcal{U}} = \mathcal{U}\Lambda$ while $\widetilde{\mathcal{G}}=\mathcal{E}\mathcal{G}$. The noise channels for $U$ and gates in $\mathsf{G}$ are put in different positions, but due to the arbitrariness of $\Lambda$ and $\mathcal{E}$, the difference does not put any restriction on the noise. In~\cite{Magesan2012interleavedRB}, the target gate $U$ belongs to twirling group $\mathsf{G}$, then under the expectation of sampling of $G_i$, Eq.~\eqref{eq:interleavedrandomsequence} is equal to,
\begin{equation}
\begin{split}
\mathbb{E}_{\forall i, G_i\in\mathsf{G}} \widetilde{S}_i &= \mathcal{E} \mathbb{E}_{\forall i, G_i\in\mathsf{G}} \prod_{i=1}^{m} ((\prod_{j=2}^i\mathcal{G}_j\mathcal{U})\mathcal{G}_1)^{\dagger} \Lambda\mathcal{E} ((\prod_{j=2}^i\mathcal{G}_j\mathcal{U})\mathcal{G}_1)\\
&=\mathcal{E} (\mathbb{E}_{G\in \mathsf{G}}\mathcal{G}^{\dagger}\Lambda\mathcal{E} \mathcal{G})^m\\
&= \mathcal{E}(\Lambda\mathcal{E})_G^m.
\end{split}
\end{equation}
The second line utilizes the fact that $\mathsf{G}$ is a group. Thus, same as the discussion in the previous subsection, interleaved quantum circuits allow us to estimate the fidelity of $\Lambda \mathcal{E}$, that is, $F(\Lambda\mathcal{E}) = \frac{\tr(\Lambda \mathcal{E})}{d^2}$. In~\cite{Magesan2012interleavedRB}, the authors show how to estimate $F(\Lambda)$ from $F(\Lambda\mathcal{E})$ and $F(\mathcal{E})$ and we omit the detail here. In~\cite{Magesan2012interleavedRB}, the target gate $U$ belongs to $\mathsf{G}$ so the noisy levels, or the fidelities, of $\Lambda$ and $\mathcal{E}$ are close. But generally speaking, implementing $\mathsf{G}$ is easier than implementing $U$, then $\mathcal{E}$ can be viewed as identity $\mathcal{I}$ compared to $\Lambda$. In our work, we simply take $\mathcal{E} = \mathcal{I}$ and focus on estimating $F(\Lambda\mathcal{E}) = F(\Lambda)$.

\subsection{Dihedral group and classically replaceable unitary operations}
For further elaboration, in this part, we introduce the local dihedral group and classically replaceable unitary operations (CRU)~\cite{Liu2022CRO} along with their properties. The local dihedral group on $n$ qubits is similar to the $n$-qubit Pauli group, which is defined as below.
\begin{equation}
\mathsf{D}_n^m = \langle X, Z_m \rangle^{\otimes n} = \langle X_1, (Z_m)_1, X_2, (Z_m)_2, \cdots, X_n, (Z_m)_n \rangle,
\end{equation}
where $X = \begin{pmatrix}
0 & 1\\
1 & 0
\end{pmatrix}$ is the Pauli $X$ gate and $Z_m = \begin{pmatrix}
1 & 0\\
0 & e^{i\frac{2\pi}{m}}
\end{pmatrix}$ is a phase gate with phase $\frac{2\pi}{m}$. Here, $m$ is a positive integer. In general, the phase on different qubits can be different and we can define $\langle X_1, (Z_{m_1})_1, X_2, (Z_{m_2})_2, \cdots, X_n, (Z_{m_n})_n \rangle$, but below we only consider a simple case that $m_1=m_2=\cdots=m_n=m$. Same with the discussion of the Pauli group, we can define the projective local dihedral group via quotient by the center of $\langle X, Z_m \rangle^{\otimes n}$,
\begin{equation}
\mathsf{D}_{\prime n}^m = \langle X, Z_m \rangle^{\otimes n} / \langle \omega_m \rangle,
\end{equation}
where $\omega_m = e^{i\frac{2\pi}{m}}$ and $\langle \omega_m \rangle = \{e^{i\frac{2k\pi}{m}}, 0\leq k\leq m-1\}$. As the overall phase is not important, below we do not distinguish the local dihedral group and the projective local dihedral group and use the definition $\mathsf{D}_n^m = \langle X, Z_m \rangle^{\otimes n} / \langle \omega_m \rangle$.

Classically replaceable unitary operations, or incoherent unitary operations, are all unitary gates that can be moved after computational basis measurements and be replaced with classical post-processing. Given the computational basis $\ket{i}$ on the Hilbert space $\mathcal{H}$, we can define the dephasing operation $\Delta$. For any $\rho\in \mathcal{D}(\mathcal{H})$,
\begin{equation}
\Delta(\rho)  = \sum_{i} \ketbra{i} \rho \ketbra{i}.
\end{equation}
Then, the CRU gate set has an expression,
\begin{equation}
\{U \text{ is unitary}| \mathcal{U}\Delta = \Delta \mathcal{U}\}.
\end{equation}
Note that we use the same notation $\Delta$ to represent its map representation and Liouville representation. It is shown that a CRU can be given by~\cite{Liu2022CRO,Winter2016Coherence}
\begin{equation}\label{eq:CRU}
U = \sum_j e^{i\theta_j} \ketbra{\sigma(j)}{j},
\end{equation}
where $\sigma$ is a permutation over computational basis and $\theta_j$ is a phase in $[ 0, 2\pi )$. And any unitary gate having an expression of~\eqref{eq:CRU} is a CRU. We can decompose $U$ as
\begin{equation}
U = \sum_j \ketbra{\sigma(j)}{j} \sum_j e^{i\theta_j} \ketbra{j}.
\end{equation}
Set $\Pi = \sum_j \ketbra{\sigma(j)}{j}$ and $W = \sum_j e^{i\theta_j} \ketbra{j}$ and we can see any CRU can be written as the multiplication of a permutation matrix, $\Pi$, and a diagonal matrix, $W$, $U = \Pi W$. Note that any matrix with the form $\sum_j \ketbra{\sigma(j)}{j}$ for arbitrary permutation $\sigma$ is defined to be a permutation matrix. Also, the multiplication of a permutation matrix and a diagonal matrix is always expressed as Eq.~\eqref{eq:CRU} and is a CRU.

\begin{lemma}
Any classically replaceable unitary operation $U$ can be decomposed as the multiplication of a permutation matrix, $\Pi$, and a diagonal matrix, $W$, $U = \Pi W$. Vice versa.
\end{lemma}

From the above lemma, we can obtain that the computational basis gate set is invariant under the action of a CRU.

\begin{lemma}\label{lemma:ZinvariantCRU}
Diagonal matrices set or computational basis gate set is invariant under the action of any CRU.
\end{lemma}
\begin{proof}
For any CRU $U = \sum_j e^{i\theta_j} \ketbra{\sigma(j)}{j}$ and diagonal matrix $W = \sum_j e^{\phi_j} \ketbra{j}$,
\begin{equation}
\begin{split}
UWU^{-1} &= \sum_{jj'k}e^{i(\phi_k + \theta_j - \theta_{j'})} \ket{\sigma(j)} \braket{j}{k} \braket{k}{j'}\bra{\sigma(j')}\\
&= \sum_j e^{i\phi_j} \ketbra{\sigma(j)}\\
&= \sum_j e^{i\phi_{\sigma^{-1}(j)}} \ketbra{j},
\end{split}
\end{equation}
which is a diagonal matrix. Proof is done.
\end{proof}

In fact, for any CRU subgroup $\mathsf{G}\leq \text{CRU}$, all diagonal matrices in $\mathsf{G}$ forms a subgroup, $\mathsf{G}_Z$ of $\mathsf{G}$. Then, $\mathsf{G}_Z$ is invariant under the action of any gate in $\mathsf{G}$, or equivalently, $\mathsf{G}_Z$ is a normal subgroup of $\mathsf{G}$.

\begin{lemma}\label{lemma:ZinvariantCRU2}
Given an $n$-qubit CRU subgroup $\mathsf{G}$, set the computational basis subgroup of $\mathsf{G}$ as $\mathsf{G}_Z = \{U\in \mathsf{G} | U \text{ is diagonal} \}$. Then, $\mathsf{G}_Z$ is a normal subgroup of $\mathsf{G}$.
\end{lemma}
\begin{proof}
Obviously, $\mathsf{G}_Z$ is a subgroup of $\mathsf{G}$. We only need to prove that under the conjugate action of any quantum gate $G$ in $\mathsf{G}$, $\mathsf{G}_Z$ is invariant. For any gate $G\in \mathsf{G}$, $G\in \text{CRU}$, then $G\mathsf{G}_Z G^{\dagger}$ only contains diagonal matrices. As all diagonal matrices in $\mathsf{G}$ are contained in $\mathsf{G}_Z$, $G\mathsf{G}_Z G^{\dagger} \subseteq \mathsf{G}_Z$. Also, $\abs{G\mathsf{G}_Z G^{\dagger}} = \abs{\mathsf{G}_Z}$, so $G\mathsf{G}_Z G^{\dagger} = \mathsf{G}_Z$. As $G$ is an arbitrary gate in $\mathsf{G}$, we prove that $\mathsf{G}_Z$ is the normal subgroup of $\mathsf{G}$. Here we complete the proof.
\end{proof}

Below we present a result about permutation matrices for further elaboration. The lemma tells us the structure of the permutation matrices, which can always be decomposed as the multiplication of Toffoli gates, CNOT gates, and Pauli $X$ gates.

\begin{lemma}\label{lemma:permutation}
All permutation matrices on $n$ qubits can be generated by Pauli $X$ gate on each qubit and all Toffoli gates $C^{n-1}X$ with $n-1$ control qubits and 1 target qubit.
\end{lemma}

\begin{proof}
The expression of permutation matrices on $n$ qubits can be unified as below.
\begin{equation}
\Pi_{\sigma} = \sum_{j\in \{0,1\}^n} \ketbra{\sigma(j)}{j},
\end{equation}
where $\sigma$ is an arbitrary permutation on $\{0,1\}^n$. The permutations themself form a permutation group with cardinality $2^n !$. As transpositions can generate the permutation group, we only need to prove that any transposition $(s_1, s_2)$ can be generated by Pauli $X$ gates and $C^{n-1}X$ gates for $s_1\neq s_2\in \{0,1\}^n$.

Denote $X_k$ to be the Pauli $X$ gate acting on $k$-th qubit and set $X^s = \bigotimes_{i=1}^n X_i^{s_i}$ for $s\in \{0,1\}^n$. We also denote $C^{n-1}X_k$ to be the Toffoli gate with $k$-th qubit as the target qubit and other qubits as the control qubits. Then, for any $s\in \{0,1\}^n$, $X^s C^{n-1}X_k (X^s)^{\dagger}$ would swap the basis $\ket{1^n\oplus s}$ and $\ket{1^n\oplus 1_k \oplus s}$ while keeping other bases fixed. Here, $1^n = 11\cdots 1$ is the all-1 bit string and $1_k = 0^{k-1}10^{n-k}$ is the bit string with 1 in $k$-th position and 0 in other positions. Thus, by taking $k$ over 1 to $n$ and $s$ over $\{0,1\}^n$ we would obtain any transposition $(s, s\oplus 1_k)$. As $(s, s\oplus 1_{k1})(s, s\oplus 1_{k2})(s, s\oplus 1_{k1}) = (s\oplus 1_{k1}, s\oplus 1_{k2})$, we would also obtain any transposition $(s, s\oplus 1_{k1} \oplus 1_{k2} ... \oplus 1_{kl})$, which suffices to produce any transposition $(s_1, s_2)$ for $s_1\neq s_2\in \{0,1\}^n$.

Note that the Toffoli gates with less than $n-1$ control qubits are also permutation matrices. It means that with single qubit Pauli $X$ gates and $C^{n-1}X$ gates, one can construct any controlled-$X$ gates $C^kX$ with $1\leq k\leq n-2$.
\end{proof}

\section{Optimal twirling groups for multi-qubit controlled phase gates in Randomized Benchmarking}\label{appendsc:group}
In original interleaved RB~\cite{Magesan2012interleavedRB}, to benchmark an individual target quantum gate, $U$, one always embeds $U$ in a large group, $\mathsf{G}$, with a number of global quantum gates so that $\mathsf{G}$ has a strong twirling effect. However, the large size of the twirling group would make group sampling and computing inverse gates difficult, and plenty of global quantum gates in $\mathsf{G}$ would make the twirling gates hard to realize. The former difficulty is a classical computational problem, and the latter is a gate implementability problem. In this work, we focus on finding the smallest and the most easily implementable twirling group for a target gate.

Fortunately,~\cite{Erhard2019CB,Helsen2019characterRB,Zhang2023Scalable} shows that embedding the target gate $U$ in the twirling group $\mathsf{G}$ is not necessary to characterize $U$. A Clifford gate can be effectively characterized with local Clifford twirling or Pauli twirling instead of global Clifford twirling. It gives hope that maybe we can choose a small twirling group for benchmarking a generic target quantum gate. Normally, we would like to choose a twirling group as small as possible, but the twirling group can not be arbitrarily small as well. Otherwise, the twirled noise channel is not symmetric enough, resulting in hard post-processing and inaccurate fidelity estimation. In the main text, we have briefly introduced the requirements for twirling groups in RB and proposed Question~\ref{ques:RB}. Below, we will present the requirements for twirling groups in RB more detailedly and completely and rederive Question~\ref{ques:RB}. Moreover, in this section, we will provide the formal versions of the theorems and lemmas in the main text along with their proof. It is worth mentioning that when we write `a channel is diagonal', it always means that `the channel is diagonal up to a unitary transformation'.

\subsection{Requirements for twirling groups in randomized benchmarking}\label{appendssc:rbrequire}
Below, we focus on how to estimate the fidelity of a target gate, $U$, with methods of RB. Same with the notations in the main text and in Appendix~\ref{appendsc:pre}, we express the noisy quantum gate $\widetilde{\mathcal{U}} = \mathcal{U} \Lambda$ as a composite of the noiseless gate $U$ and its noise channel $\Lambda$, where $\mathcal{U}$ denotes the Pauli-Liouville representation of $U$, and $\widetilde{\cdot}$ represents the noisy version of a quantum gate. As mentioned in Appendix~\ref{appendssc:RB}, the key to enable RB is obtaining the powers of the $\mathsf{G}$-twirled noise channel, $\Lambda_G^m$, where $m\in \mathbb{Z}+$, and ensuring $\Lambda_G$ is diagonal up to a unitary transformation, or can be written as Eq.~\eqref{eq:projectordecompose2}. Then, with the technique of character RB, one can obtain $\tr(\Lambda_G)$ and the fidelity of $\Lambda$. Here, $\mathsf{G}$ is the twirling group for tailoring $\Lambda$.

Different from the elaboration in the main text, below we distinguish two concepts of twirling gate set $\mathsf{V}$ and twirling group $\mathsf{G}$. In the main text, we directly select a twirling group, $\mathsf{G}$, and implement a random gate sequence composed of gates from $\mathsf{G}$ interleaved with target gate $U$ to obtain $\Lambda_G^m$. But in fact, we have another way to obtain $\Lambda_G^m$ as shown below.

To obtain $\Lambda_G^m$, we implement $m$ twirling gates $V_i = G_{i}UG_{i-1}^{\dagger}U^{\dagger}$, sampled from the twirling gate set $\mathsf{V} = \mathsf{G}U\mathsf{G}U^{\dagger}$, interleaved with the target gate $U$ as shown in Fig.~\ref{fig:RBappend}. Here, $G_i$ is uniformly and randomly sampled from the group $\mathsf{G}$, and we set $G_0=\mathbb{I}$. The circuit ends with the inverse gate $V_{\text{inv}} = (\prod_{i=1}^m UV_i)^{\dagger} = U^{\dagger m-1}G_m^{\dagger} U^{\dagger}$. The circuit ideally equals the identity, corresponding to $\Lambda_G=I$. Note that the twirling group $\mathsf{G}$ and the twirling gate set $\mathsf{V}$ are in general different and $\mathsf{G}\subseteq \mathsf{V}$. Although we only require $\Lambda$ to be twirled by $\mathsf{G}$, we need to realize gates in $\mathsf{V}$ to achieve that since the circuit involves gate $U$. The twirling gates should first eliminate the effect of the twirled gate $U$ and then influence the noise channel $\Lambda$. Note that the inverse gate $V_{inv}$ can always belong to $\mathsf{V}$ as long as we choose $m$ such that $U^m=\mathbb{I}$. The cost to implement the inverse gate is nearly the same as other twirling gates.

\begin{figure}[htbp!]
\centering
\includegraphics[width=.5\textwidth]{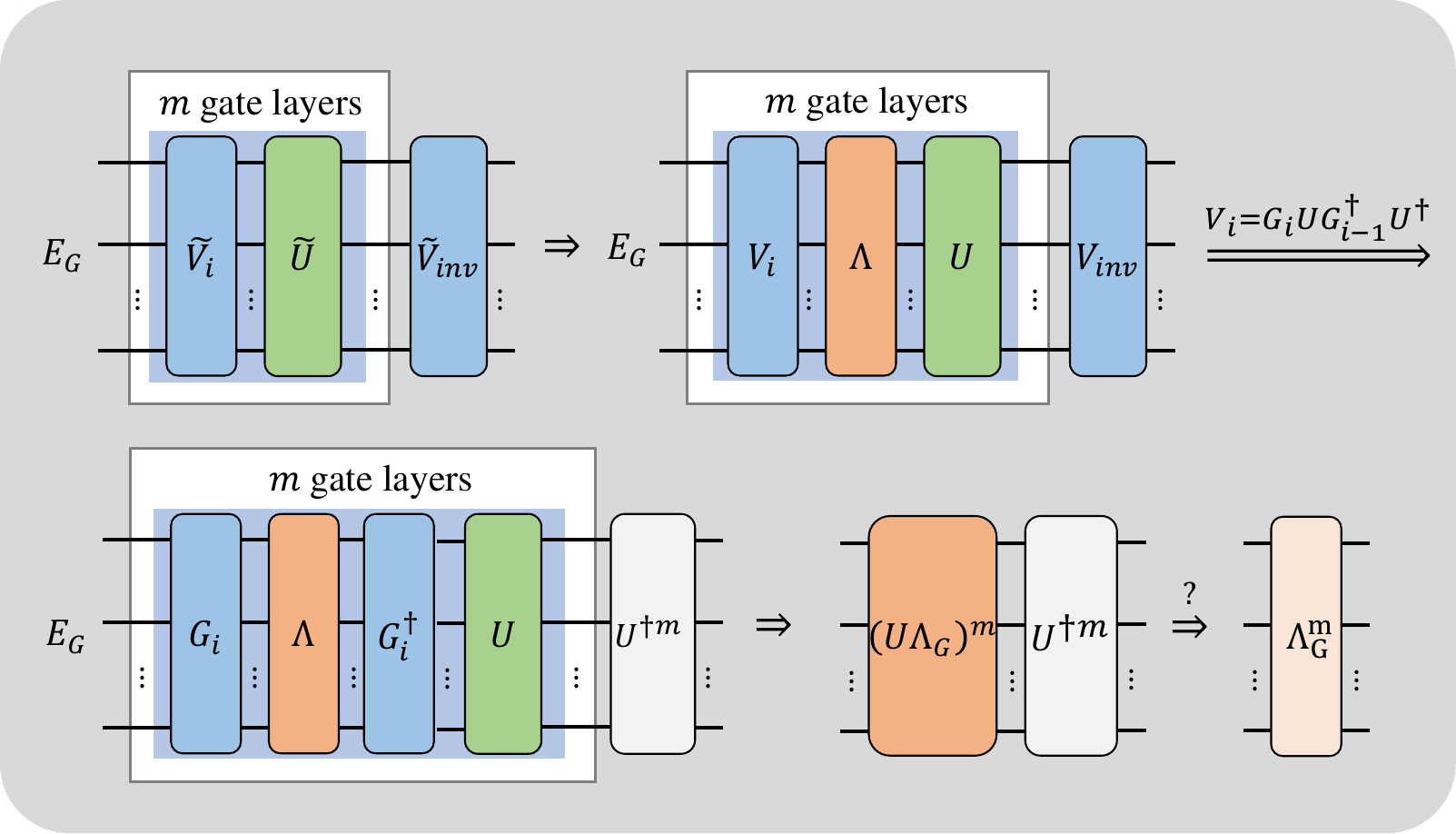}
\caption{Random twirling gates $V_1, V_2,\cdots,V_m$ interleaved with $U$ in RB. The inverse gate $V_{\text{inv}} = (\prod_{i=1}^m UV_{i})^{-1}$. $\widetilde{\cdot}$ represents the noisy version of a quantum gate. $\Lambda$ is the composite noise channel of $V$ and $U$, and $\Lambda_G$ is the $G$-twirled noise channel.}
\label{fig:RBappend}
\end{figure}

Same with the main text and the arguments in Appendix~\ref{appendssc:interleavedRB}, we omit the noise from the twirling gates. Then, the circuit is represented as
\begin{equation}\label{eq:RB2}
\begin{split}
\widetilde{\mathcal{S}}_m &= \mathcal{V}_{\text{inv}}\prod_{i=1}^m \mathcal{U}\Lambda\mathcal{V}_{i}\\
&= \mathcal{U}^{\dagger m-1}\mathcal{G}_m^{\dagger} \mathcal{U}^{\dagger} \prod_{i=1}^m \mathcal{U}\Lambda\mathcal{G}_i\mathcal{U}\mathcal{G}^{\dagger}_{i-1}\mathcal{U}^{\dagger}\\
&= \mathcal{U}^{\dagger m} \prod_{i=1}^m \mathcal{U} \mathcal{G}_{i}^{\dagger}\Lambda\mathcal{G}_{i}.
\end{split}
\end{equation}
In terms of expectation over sampling of $G_i$, Eq.~\eqref{eq:RB2} becomes $\mathcal{U}^{\dagger m} (\mathcal{U}\Lambda_G)^m$. Compared to $\Lambda_G^{m}$, this formulation only requires the commutation relation,
\begin{equation}\label{eq:Ucommute}
\Lambda_G \mathcal{U} = \mathcal{U}\Lambda_G,
\end{equation}
which means that the symmetry of the group $\mathsf{G}$ is preserved under the action of $U$ in the Liouville representation. Thus, to tailor $U$ in RB, our task is finding the twirling group $\mathsf{G}$ satisfying Eq.\eqref{eq:Ucommute} and ensuring the diagonalizability of $\Lambda_G$ while minimizing the size of the twirling gate set $\mathsf{V} = \mathsf{G}U\mathsf{G}U^{\dagger}$. The question is summarized below.

\begin{question}\label{ques:RB2}
Given a gate, $U$, find $\mathsf{V}$ such that $\mathsf{V} = \mathsf{G}U\mathsf{G}U^{\dagger}$ and the twirling group $\mathsf{G}$ satisfies
\begin{equation}
\begin{split}
&\text{for any quantum channel } \Lambda,\\
&\Lambda_G = \mathbb{E}_{G\in \mathsf{G}} \mathcal{G} \Lambda \mathcal{G}^{\dagger} \text{ is diagonal up to a}\\
&\text{unitary transformation independent of } \Lambda,\\
\end{split}
\end{equation}
and,
\begin{equation}
\Lambda_G \mathcal{U} = \mathcal{U}\Lambda_G.
\end{equation}
\end{question}

Question~\ref{ques:RB2} is a more refined problem for finding a twirling group in RB. However, in general, Eq.\eqref{eq:Ucommute} is challenging to characterize, so we substitute it with a more easily handled condition, that is,
\begin{equation}
U\mathsf{G}U^{\dagger} = \mathsf{G}.
\end{equation}
In this case, $\mathsf{V} = \mathsf{G}$ and Question~\ref{ques:RB2} reduces to Question~\ref{ques:RB} in the main text, as shown below.

\begin{repquestion}{ques:RB}
Given a gate, $U$, find a twirling group, $\mathsf{G}$ such that,
\begin{equation}\label{eq:diagappend}
\begin{split}
&\text{for any quantum channel } \Lambda, \Lambda_G = \mathbb{E}_{G\in \mathsf{G}} \mathcal{G} \Lambda \mathcal{G}^{\dagger} \text{ is diagonal}\\
&\text{up to a unitary transformation independent of } \Lambda,\\
\end{split}
\end{equation}
and,
\begin{equation}\label{eq:Uconjgroupappend}
U\mathsf{G} U^{\dagger} = \mathsf{G}.
\end{equation}
\end{repquestion}

As mentioned in the main text, the solution to Question~\ref{ques:RB} allows us to obtain the fidelity of $(\mathcal{U}^{\dagger m}(\mathcal{U}\Lambda_G)^m)^{\frac{1}{m}}$, which is a lower bound of the fidelity of $\Lambda$. Below, we present the proof.

\begin{lemma}
Given a target gate, $U$, and a twirling group, $\mathsf{G}$ satisfying Question~\ref{ques:RB}, then for any positive integer $m$, the fidelity of $(\mathcal{U}^{\dagger m}(\mathcal{U}\Lambda_G)^m)^{\frac{1}{m}}$ is a lower bound of the fidelity of $\Lambda$. Mathematically,
\begin{equation}
F((\mathcal{U}^{\dagger m}(\mathcal{U}\Lambda_G)^m)^{\frac{1}{m}}) \leq F(\Lambda),
\end{equation}
or equivalently,
\begin{equation}
\tr((\mathcal{U}^{\dagger m}(\mathcal{U}\Lambda_G)^m)^{\frac{1}{m}}) \leq \tr(\Lambda).
\end{equation}
\end{lemma}

\begin{proof}
As $\tr(\Lambda) = \tr(\Lambda_G)$, we only need to prove $\tr((\mathcal{U}^{\dagger m}(\mathcal{U}\Lambda_G)^m)^{\frac{1}{m}}) \leq \tr(\Lambda_G)$.

Note that $\Lambda_G$ has the expression of Eq.~\eqref{eq:projectordecompose} and can be written as
\begin{equation}
\Lambda_G = \sum_{i=1}^k p_i \Pi_i.
\end{equation}
Also, the condition $U\mathsf{G}U^{\dagger} = \mathsf{G}$ ensures that $\mathcal{U}^{\dagger}\Lambda_G \mathcal{U} = (\mathcal{U}^{\dagger}\Lambda_G \mathcal{U})_G$ as proven in Lemma~\ref{lemma:Uconjugaterelation}, which means we can express $\mathcal{U}^{\dagger}\Lambda_G \mathcal{U}$ as
\begin{equation}
\mathcal{U}^{\dagger}\Lambda_G \mathcal{U} = \sum_{i=1}^k q_i \Pi_i.
\end{equation}

As $\mathcal{U}$ is unitary and does not change the spectrum of $\Lambda_G$, the value of $q_i$ must be equal to one of elements from $\{p_i, 1\leq i\leq k\}$. Thus, we can construct a map $f_1: \{p_i, 1\leq i \leq k\} \rightarrow \{q_i, 1\leq i \leq k\}$. Similarly, as $\mathcal{U}(\mathcal{U}^{\dagger}\Lambda_G \mathcal{U}) \mathcal{U}^{\dagger} = \Lambda_G$, we can construct a map $f_2: \{q_i, 1\leq i \leq k\} \rightarrow \{p_i, 1\leq i \leq k\}$ so that the compositions of $f_1$ and $f_2$ are identity maps, $f_1\circ f_2(q_i) = q_i$ and $f_2\circ f_1(p_i) = p_i$. Then, $f_1$ is just a permutation on $\{p_i, 1\leq i \leq k\}$ and $\{q_i, 1\leq i \leq k\}$ is equivalent to $\{p_i, 1\leq i \leq k\}$ after permutation. Thus, the diagonal terms of $\mathcal{U}^{\dagger}\Lambda_G \mathcal{U} \Lambda_G$ are $\{p_ip_{\sigma(i)}\}$ where $\sigma$ is a permutation on $\{1,2,\cdots,k\}$.

Similarly, the diagonal terms of $\mathcal{U}^{\dagger m}(\mathcal{U}\Lambda_G)^m = (\mathcal{U}^{\dagger}\cdots (\mathcal{U}^{\dagger}(\mathcal{U}^{\dagger}\Lambda_G \mathcal{U}) \Lambda_G \mathcal{U})\Lambda_G\cdots\mathcal{U})\Lambda_G$ are $\{p_ip_{\sigma_1(i)}\cdots p_{\sigma_{m-1}(i)}\}$ where $\sigma_1,\cdots,\sigma_{m-1}$ are permutations on $\{1,2,\cdots,k\}$. Thus, with rearrangement inequality,
\begin{equation}
\begin{split}
\tr((\mathcal{U}^{\dagger m}(\mathcal{U}\Lambda_G)^m)^{\frac{1}{m}}) &= \sum_{i=1}^k (p_ip_{\sigma_1(i)}\cdots  p_{\sigma_{m-1}(i)})^{\frac{1}{m}}\\
&\leq \sum_{i=1}^k p_i\\
&= \tr(\Lambda_G).
\end{split}
\end{equation}
Here is the proof.
\end{proof}

Thus, solve Question~\ref{ques:RB} and then we can provide a lower bound of the fidelity of the target gate. It is worth noting that if the target gate $U$ is a multi-qubit controlled phase gate, the twirling group satisfying $U\mathsf{G}U^{\dagger} = \mathsf{G}$ would also satisfy $\Lambda_G \mathcal{U} = \mathcal{U}\Lambda_G$, which we will show in Theorem~\ref{thm:ucommute}. As a consequence, for multi-qubit controlled phase gates, we can accurately estimate their fidelities instead of only providing a lower bound.

\begin{lemma}\label{lemma:Uconjugaterelation}
Given a unitary gate, $U$, and a unitary subgroup, $\mathsf{G}$, if $U\mathsf{G}U^{-1} = \mathsf{G}$, then for any quantum channel $\Lambda$, $(\mathcal{U}\Lambda_G \mathcal{U}^{\dagger})_G = \mathcal{U}\Lambda_G \mathcal{U}^{\dagger}$.
\end{lemma}
\begin{proof}
Provided with $U\mathsf{G}U^{-1} = \mathsf{G}$, we have
\begin{equation}
\begin{split}
(\mathcal{U}\Lambda_G \mathcal{U}^{\dagger})_G &= \mathbb{E}_{G\in \mathsf{G}} \mathcal{G} \mathcal{U}\Lambda_G \mathcal{U}^{\dagger} \mathcal{G}^{\dagger}\\
&= \mathbb{E}_{G\in \mathsf{G}} \mathcal{U} (\mathcal{U}^{\dagger}\mathcal{G} \mathcal{U})\Lambda_G (\mathcal{U}^{\dagger} \mathcal{G}^{\dagger}\mathcal{U}) \mathcal{U}^{\dagger}\\
&= \mathbb{E}_{G\in \mathsf{G}} \mathcal{U}\mathcal{G} \Lambda_G  \mathcal{G}^{\dagger}\mathcal{U}^{\dagger}\\
&= \mathcal{U}\Lambda_G \mathcal{U}^{\dagger}.
\end{split}
\end{equation}
The equality in the third line comes from the condition $U\mathsf{G}U^{-1} = \mathsf{G}$.
\end{proof}

\subsection{Proof of Lemma~\ref{lemma:cardconstraint}}\label{appendssc:lemmacardconstraint}
In the main text, we present the lemma that for a finite $n$-qubit unitary subgroup $\mathsf{G}$, if for any quantum channel $\Lambda$, its $\mathsf{G}-$twirled channel, $\Lambda_G$, is diagonal in the Pauli-Liouville representation up to a unitary transformation independent of $\Lambda$, then the Pauli-Liouville representation of $\mathsf{G}$ is multiplicity-free. Also, the cardinality of the twirling group $\abs{\mathsf{G}}\geq 4^n$. Below, we present the proof of this lemma. It is worth mentioning that, in the proof, we consider a more generic scenario: if $\Lambda_G$ is diagonal up to an invertible matrix transformation rather than a unitary transformation, then the Pauli-Liouville representation of $\mathsf{G}$ is multiplicity-free. The result is more generic, and the proof is stronger here.

For convenience, we rewrite the Lemma~\ref{lemma:cardconstraint} below.

\begin{replemma}{lemma:cardconstraint}
For a finite $n$-qubit unitary subgroup, $\mathsf{G}$, if for any quantum channel $\Lambda$, its $\mathsf{G}-$twirled channel, $\Lambda_G$, is diagonal in the Pauli-Liouville representation up to an invertible matrix transformation independent of $\Lambda$, then the Pauli-Liouville representation of $\mathsf{G}$ is multiplicity-free. As a corollary, the cardinality of the twirling group $\abs{\mathsf{G}}\geq 4^n$.
\end{replemma}

\begin{proof}\label{lemmapf:cardconstraint}
Let us begin with the irreducible representation decomposition of the Pauli Liouville representation of $\mathsf{G}$. For each element $G\in \mathsf{G}$, $\mathcal{G}$ can be decomposed as the direct sum of irreducible representations up to an isomorphism $\mathcal{V}$,
\begin{equation}\label{eq:irrepdecompose}
\mathcal{V}\mathcal{G}\mathcal{V}^{-1} = \bigoplus_{i=1}^k \mathbb{I}_{n_{i}\times n_{i}}\otimes \phi_{i}(G),
\end{equation}
where $\mathcal{V}$ is an invertible matrix, $\phi_{i}$ denotes the irreducible representation of $\mathsf{G}$, and $n_{i}$ denotes its multiplicity in $\mathcal{G}$. The term $k$ records the number of inequivalent irreducible representations that $\mathcal{G}$ contains. It is worth mentioning that the basis making $\mathcal{G}$ block-diagonal may not be Pauli operators $\{\mathbb{I},X,Y,Z\}$. That is why we put $\mathcal{V}$ and $\mathcal{V}^{-1}$ in the left side of Eq.~\eqref{eq:irrepdecompose}. As any unitary channel has an invariant subspace $\{\mathbb{I}\}$, the Pauli-Liouville representation of a unitary channel must be block-diagonal as below.
\begin{equation}\label{eq:unitaryptmdecompose}
\mathcal{G} = \begin{pmatrix}
1 & \mathbf{0}\\
\mathbf{0} & \mathbf{T}_\mathcal{G}
\end{pmatrix}.
\end{equation}
Thus, the form of the basis transformation matrix, $\mathcal{V}$, can also be constrained like Eq.~\eqref{eq:unitaryptmdecompose}. That means $\mathcal{V}$ only changes the non-identity basis while keeping the basis $\mathbb{I}$ invariant. In addition, without loss of generality, we set the irreducible representation in the subspace spanned by $\{\mathbb{I}\}$ as the trivial representation, that is, mapping all group elements to 1.

Given a channel, $\Lambda$, its twirling over group $\mathsf{G}$ is
\begin{equation}\label{eq:LambdaGtwirling}
\begin{split}
\Lambda_G &= \mathbb{E}_{G\in \mathsf{G}} \mathcal{G} \Lambda \mathcal{G}^{-1}\\
&= \mathbb{E}_{G\in \mathsf{G}} \mathcal{V}^{-1} (\bigoplus_{i=1}^k \mathbb{I}_{n_{i}\times n_{i}}\otimes \phi_{i}(G)) \mathcal{V}\Lambda \mathcal{V}^{-1}(\bigoplus_{i=1}^k \mathbb{I}_{n_{i}\times n_{i}}\otimes \phi_{i}(G))^{-1} \mathcal{V}\\
&= \mathcal{V}^{-1} [\mathbb{E}_{G\in \mathsf{G}} (\bigoplus_{i=1}^k \mathbb{I}_{n_{i}\times n_{i}}\otimes \phi_{i}(G) ) \mathcal{V}\Lambda \mathcal{V}^{-1}(\bigoplus_{i=1}^k \mathbb{I}_{n_{i}\times n_{i}}\otimes \phi_{i}(G))^{-1}] \mathcal{V}\\
\end{split}
\end{equation}

Focusing on the calculation in the square brackets in Eq.~\eqref{eq:LambdaGtwirling}, we denote $\Lambda' = \mathcal{V}\Lambda \mathcal{V}^{-1}$ and decompose it corresponding to the block partition of the right-hand side in Eq.~\eqref{eq:irrepdecompose}. That is,
\begin{equation}\label{eq:lambda'}
\begin{split}
\Lambda' &= \bigoplus_{i=1}^k \Lambda'^i\\
&= \bigoplus_{i=1}^k \begin{pmatrix}
\Lambda'^i_{11} & \Lambda'^i_{12} & \cdots & \Lambda'^i_{1n_i}\\
\Lambda'^i_{21} & \Lambda'^i_{22} & \cdots & \Lambda'^i_{2n_i}\\
\vdots & \vdots & \ddots & \vdots\\
\Lambda'^i_{n_i1} & \Lambda'^i_{n_i2} & \cdots & \Lambda'^i_{n_in_i}.
\end{pmatrix}
\end{split}
\end{equation}
Then,
\begin{equation}
\begin{split}
&\mathbb{E}_{G\in \mathsf{G}} (\bigoplus_{i=1}^k \mathbb{I}_{n_{i}\times n_{i}}\otimes \phi_{i}(G)) \Lambda'(\bigoplus_{i=1}^k \mathbb{I}_{n_{i}\times n_{i}}\otimes \phi_{i}(G))^{-1}\\
=&\mathbb{E}_{G\in \mathsf{G}} \bigoplus_{i=1}^k
[\begin{pmatrix}
\phi_{i}(G) & 0 & \cdots & 0\\
0 & \phi_{i}(G) & \cdots & 0\\
\vdots & \vdots & \ddots & \vdots\\
0 & 0 & \cdots & \phi_{i}(G)
\end{pmatrix}
\begin{pmatrix}
\Lambda'^i_{11} & \Lambda'^i_{12} & \cdots & \Lambda'^i_{1n_i}\\
\Lambda'^i_{21} & \Lambda'^i_{22} & \cdots & \Lambda'^i_{2n_i}\\
\vdots & \vdots & \ddots & \vdots\\
\Lambda'^i_{n_i1} & \Lambda'^i_{n_i2} & \cdots & \Lambda'^i_{n_in_i}.
\end{pmatrix}
\begin{pmatrix}
\phi_{i}^{-1}(G) & 0 & \cdots & 0\\
0 & \phi_{i}^{-1}(G) & \cdots & 0\\
\vdots & \vdots & \ddots & \vdots\\
0 & 0 & \cdots & \phi_{i}^{-1}(G)
\end{pmatrix}
]\\
=&\mathbb{E}_{G\in \mathsf{G}} \bigoplus_{i=1}^k
\begin{pmatrix}
\phi_{i}(G)\Lambda'^i_{11}\phi_{i}^{-1}(G) & \phi_{i}(G)\Lambda'^i_{12}\phi_{i}^{-1}(G) & \cdots & \phi_{i}(G)\Lambda'^i_{1n_i}\phi_{i}^{-1}(G)\\
\phi_{i}(G)\Lambda'^i_{21}\phi_{i}^{-1}(G) & \phi_{i}(G)\Lambda'^i_{22}\phi_{i}^{-1}(G) & \cdots & \phi_{i}(G)\Lambda'^i_{2n_i}\phi_{i}^{-1}(G)\\
\vdots & \vdots & \ddots & \vdots\\
\phi_{i}(G)\Lambda'^i_{n_i1}\phi_{i}^{-1}(G) & \phi_{i}(G)\Lambda'^i_{n_i2}\phi_{i}^{-1}(G) & \cdots & \phi_{i}(G)\Lambda'^i_{n_in_i}\phi_{i}^{-1}(G).
\end{pmatrix}\\
=&\bigoplus_{i=1}^k
\begin{pmatrix}
\mathbb{E}_{G\in \mathsf{G}}\phi_{i}(G)\Lambda'^i_{11}\phi_{i}^{-1}(G) & \mathbb{E}_{G\in \mathsf{G}}\phi_{i}(G)\Lambda'^i_{12}\phi_{i}^{-1}(G) & \cdots & \mathbb{E}_{G\in \mathsf{G}}\phi_{i}(G)\Lambda'^i_{1n_i}\phi_{i}^{-1}(G)\\
\mathbb{E}_{G\in \mathsf{G}}\phi_{i}(G)\Lambda'^i_{21}\phi_{i}^{-1}(G) & \mathbb{E}_{G\in \mathsf{G}}\phi_{i}(G)\Lambda'^i_{22}\phi_{i}^{-1}(G) & \cdots & \mathbb{E}_{G\in \mathsf{G}}\phi_{i}(G)\Lambda'^i_{2n_i}\phi_{i}^{-1}(G)\\
\vdots & \vdots & \ddots & \vdots\\
\mathbb{E}_{G\in \mathsf{G}}\phi_{i}(G)\Lambda'^i_{n_i1}\phi_{i}^{-1}(G) & \mathbb{E}_{G\in \mathsf{G}}\phi_{i}(G)\Lambda'^i_{n_i2}\phi_{i}^{-1}(G) & \cdots & \mathbb{E}_{G\in \mathsf{G}}\phi_{i}(G)\Lambda'^i_{n_in_i}\phi_{i}^{-1}(G).
\end{pmatrix}
\end{split}
\end{equation}
By Schur's lemma, for any irreducible representation $\phi_{i}$ and matrix $A$, the twirling of $A$ over $\mathsf{G}$ would be proportional to identity,
\begin{equation}
\mathbb{E}_{G\in \mathsf{G}}\phi_{i}(G)A\phi_{i}^{-1}(G) = \tr(A)\mathbb{I}_{d_i},
\end{equation}
where $d_i = \dim \phi_i$. We set
\begin{equation}\label{eq:lambdai'}
\Lambda_i' =
\begin{pmatrix}
\tr(\Lambda'^i_{11}) & \tr(\Lambda'^i_{12}) & \cdots & \tr(\Lambda'^i_{1n_i})\\
\tr(\Lambda'^i_{21}) & \tr(\Lambda'^i_{22}) & \cdots & \tr(\Lambda'^i_{2n_i})\\
\vdots & \vdots & \ddots & \vdots\\
\tr(\Lambda'^i_{n_i1}) & \tr(\Lambda'^i_{n_i2}) & \cdots & \tr(\Lambda'^i_{n_in_i}).
\end{pmatrix}
\end{equation}
Thus we conclude that
\begin{equation}
\mathbb{E}_{G\in \mathsf{G}} (\bigoplus_{i=1}^k \mathbb{I}_{n_{i}\times n_{i}}\otimes \phi_{i}(G)) \Lambda'(\bigoplus_{i=1}^k \mathbb{I}_{n_{i}\times n_{i}}\otimes \phi_{i}(G))^{-1} = \bigoplus_{i=1}^k \Lambda'_i\otimes \mathbb{I}_{d_i},
\end{equation}
and
\begin{equation}
\Lambda_G = \mathcal{V}^{-1} (\bigoplus_{i=1}^k \Lambda'_i\otimes \mathbb{I}_{d_i}) \mathcal{V}.
\end{equation}

Note that $\Lambda_G$ is diagonal up to an invertible matrix transformation independent of $\Lambda$. There exists a fixed invertible matrix transformation $\mathcal{V}'$ such that
\begin{equation}\label{eq:diag2}
\mathcal{V}'\Lambda_G \mathcal{V}^{\prime\dagger} = \sum_i p_i \Pi_i,
\end{equation}
where $\Pi_i$ is a projector independent of $\Lambda$ and $\mathcal{V}'$ is also independent of $\Lambda$.

Notice that Eq.~\eqref{eq:diag2} is a linear function acting on $\Lambda$. If for any channel $\Lambda$, $\Lambda_G$ is diagonal, the linear combination of any set of quantum channels would also be diagonal after $\mathsf{G}$-twirling. Note that the linear span of quantum channels is the set of all trace-preserving (TP) maps~\cite{Robust2014Kimmel}. Thus, for any TP map with form $\Lambda = \begin{pmatrix}
1 & \mathbf{0}\\
\mathbf{t} & \mathbf{T}_{\Lambda}
\end{pmatrix}$, $\mathcal{V}'\Lambda_G \mathcal{V}^{\prime\dagger}$ must be diagonal. It further requires that for any index $i$, $\Lambda'_i$ defined in Eq.~\eqref{eq:lambdai'} is diagonal up to an invertible matrix transformation. Recall that $\Lambda'_i$ is defined from $\Lambda'$, which equals $\mathcal{V}\Lambda \mathcal{V}^{-1}$. As $\mathcal{V}$ has the form
$\begin{pmatrix}
1 & \mathbf{0}\\
\mathbf{0} & \mathbf{T}_\mathcal{V}
\end{pmatrix}$ and  $\Lambda = \begin{pmatrix}
1 & \mathbf{0}\\
\mathbf{t} & \mathbf{T}_{\Lambda}
\end{pmatrix}$ is an arbitrary TP map,
$\Lambda' = \mathcal{V}\Lambda \mathcal{V}^{-1} =
\begin{pmatrix}
1 & \mathbf{0}\\
\mathbf{T}_\mathcal{V}\mathbf{t} & \mathbf{T}_\mathcal{V}\mathbf{T}_{\Lambda}\mathbf{T}_\mathcal{V}^{-1}
\end{pmatrix}
$ is also an arbitrary TP map.  Denote $\mathbf{t}' = \mathbf{T}_\mathcal{V}\mathbf{t}$ and $\mathbf{T}_{\Lambda'} =  \mathbf{T}_\mathcal{V}\mathbf{T}_{\Lambda}\mathbf{T}_\mathcal{V}^{-1}$, the elements of $\mathbf{t}'$ and $\mathbf{T}_{\Lambda'}$ would be arbitrary. Back to considering $\Lambda'_i$, when $i = 1$, $\phi_1$ is the trivial irreducible representation, then $\Lambda_1'=\Lambda'^{i}$ can be written as below,
\begin{equation}
\Lambda_1' =
\begin{pmatrix}
1 & 0 & \cdots & 0\\
\mathbf{t}'_1 & \mathbf{T}_{\Lambda'_{11}} & \cdots & \mathbf{T}_{\Lambda'_{1, n_1-1}}\\
\vdots & \vdots & \ddots & \vdots\\
\mathbf{t}'_{n_1-1} & \mathbf{T}_{\Lambda'_{n_1-1,1}} & \cdots & \mathbf{T}_{\Lambda'_{n_1-1,n_1-1}},
\end{pmatrix}
\end{equation}
where the matrix elements are arbitrary except for the first line. If $n_1 > 1$, there exists a matrix $\begin{pmatrix}
1 & 0\\
1 & 1
\end{pmatrix} \oplus \mathbb{I}_{n_1-2}$ that cannot be diagonalized. Thus, $n_1$ must take $1$. Similarly, when $i \geq 2$,
$\Lambda'_i$ can take any matrix as $\Lambda'$ is an arbitrary TP map. If $n_i > 1$ in this case, we can also find a matrix $\begin{pmatrix}
1 & 0\\
1 & 1
\end{pmatrix} \oplus \mathbb{I}_{n_i-2}$ that cannot be diagonalized. The above arguments indicate that $\forall i$, $n_i = 1$. The irreducible representation decomposition of Pauli-Liouville representation of $\mathsf{G}$ must be multiplicity-free to make $\Lambda_G$ diagonal for any channel $\Lambda$. This completes the proof of the first conclusion in the lemma.

A direct corollary of the non-multiplicity condition of $\mathsf{G}$ in Liouville representation is the cardinality of the twirling group $\abs{\mathsf{G}}\geq 4^n$. This can be obtained by the Burnside theorem, as shown below.
\begin{equation}
\begin{split}
\abs{\mathsf{G}} &= \sum_{\phi_i \in R_G} \dim \phi_i^2\\
&\geq \sum_{i=1}^k \dim \phi_i^2\\
&\geq \sum_{i=1}^k \dim \phi_i\\
&= 4^n,
\end{split}
\end{equation}
where $R_G$ records all inequivalent irreducible representations of $\mathsf{G}$.
\end{proof}

\subsection{Systematic twirling group construction for generic quantum gates}\label{appendssc:construct}
Before we go through the proof of the main theorem, we discuss how to construct twirling groups for generic quantum gates satisfying the conditions in Question~\ref{ques:RB}.

Notice that if $\mathsf{G}$ contains a subgroup $\mathsf{H}$ that would make arbitrary channels diagonal, up to a unitary transformation, via twirling, then $\mathsf{G}$ would also enjoy this property as $\Lambda_G = (\Lambda_G)_H$, which is shown below.

\begin{lemma}\label{lemma:twirlingdiagonal}
If $\mathsf{G}$ contains a subgroup $\mathsf{H}$ that would twirl any noise channel into a diagonal channel, up to a unitary transformation, in the Pauli-Liouville representation, then $\mathsf{G}$ would also enjoy this property. Mathematically,
\begin{equation}
\mathsf{H}\subset \mathsf{G}, \forall \Lambda, \Lambda_H \text{ is diagonal} \Rightarrow \forall \Lambda, \Lambda_G \text{ is diagonal}.
\end{equation}
\end{lemma}
\begin{proof}
\begin{equation}
\begin{split}
(\Lambda_G)_H &= \mathbb{E}_{G_h\in \mathsf{H}} \mathbb{E}_{G\in \mathsf{G}} \mathcal{G}_h\mathcal{G}\Lambda (\mathcal{G}_h\mathcal{G})^{\dagger}\\
&= \mathbb{E}_{G\in \mathsf{G}} \mathcal{G}\Lambda \mathcal{G}^{\dagger}\\
&= \Lambda_G.
\end{split}
\end{equation}
Proof is done.
\end{proof}

As Pauli group $\mathsf{P}_n$ can twirl any channel into a Pauli channel, a simple solution to Question~\ref{ques:RB} is just the smallest group containing Pauli group and normalized by target $U$, which is constructed and spanned by continuously applying $U$ on $\mathsf{P}_n$ until no new element is generated. It is worth noting that, except for Pauli group $\mathsf{P}_n$, local dihedral group $\mathsf{D}^m_{n} = \langle X, Z_{m} \rangle^{\otimes n}$ can also twirl any channel into a Pauli channel. In fact, $\mathsf{P}_n$ is a special case of $\mathsf{D}^m_{n}$ when $m=2$. For specific target gate $U$ like $Z_m$, substituting $\mathsf{P}_n$ with $\mathsf{D}^m_n$ may lead to a smaller twirling group $\mathsf{G}$ though $\mathsf{D}^m_n$ is in general larger than $\mathsf{P}_n$. In reality, one can select an optimal $m$ to obtain a better choice of $\mathsf{G}$. The above discussions can be summarized below.

\begin{corollary}\label{coro:pauli}
If $n$-qubit Pauli group $\mathsf{P}_n\subseteq \mathsf{G}$, then $\Lambda_G$ is a Pauli channel and is diagonal in the Pauli-Liouville representation.
\end{corollary}

\begin{corollary}\label{coro:dihedral}
If $n$-qubit local dihedral group $\mathsf{D}_n^m \subseteq \mathsf{G}$, then $\Lambda_G$ is a Pauli channel and is diagonal in the Pauli-Liouville representation.
\end{corollary}

\begin{example}[Simple twirling group construction]\label{exp:group}
Given an $n$-qubit unitary, $U$, a simple solution to Question~\ref{ques:RB} is the smallest group containing $D_n^m$ and normalized by $U$ where $m \geq 2$ is a positive integer. Concretely, the twirling group $\mathsf{G}$ can be constructed as follows.
\begin{equation}\label{eq:twirlgroup}
\begin{split}
\mathsf{G} &= \langle \bigcup_{l\in \mathbb{N}} U^l\mathsf{D}_n^m (U^{\dagger})^l \rangle\\
&= \{ (U^{l_1}D_1 (U^{\dagger})^{l_1})^{k_1} (U^{l_2}D_2 (U^{\dagger})^{l_2})^{k_2} \cdots, \forall i, l_i, k_i\in \mathbb{Z}, D_i\in \mathsf{D}_n^m\}\\
&= \{ U^{l_1}D_1 (U^{\dagger})^{l_1} U^{l_2}D_2 (U^{\dagger})^{l_2} \cdots, \forall i, l_i\in \mathbb{Z}, D_i\in \mathsf{D}_n^m\}.
\end{split}
\end{equation}
Here, $\mathsf{D}_n^m$ is the $n$-qubit local dihedral group $\langle X, Z_m \rangle^{\otimes n}$. In practice, we select an optimal $m$ to make $\abs{\mathsf{G}}$ as small as possible.
\end{example}

\subsection{Proof of Lemma~\ref{lemma:cosetcardconstraint}}\label{appendssc:lemmacosetcardconstraint}
Below, we focus on the case that the twirling group is a CRU subgroup and prove Lemma~\ref{lemma:cosetcardconstraint} and Theorem~\ref{thm:cnzm} in the main text. We first provide a more formal and more mathematical version of Lemma~\ref{lemma:cosetcardconstraint}.

\begin{replemma}{lemma:cosetcardconstraint}[Formal Version]
For a finite $n$-qubit CRU subgroup $\mathsf{G}$, set the $Z$-basis subgroup of $\mathsf{G}$ as $\mathsf{G}_Z = \{U\in \mathsf{G} | U \text{ is Z basis} \}$, which is a normal subgroup of $\mathsf{G}$ by Lemma~\ref{lemma:ZinvariantCRU2}. If for any quantum channel $\Lambda$, its twirling over group $\mathsf{G}$, $\Lambda_G$, is diagonal up to a unitary transformation in the Pauli-Liouville representation, then the quotient group $\mathsf{G}/\mathsf{G}_Z$ can interchange any two computational basis states. In another word, $\mathsf{G}/\mathsf{G}_Z$ contains set $\mathsf{S} = \{\Pi_\mathbf{i} | \Pi_\mathbf{i} = \Pi^X_\mathbf{i}\Pi^C_\mathbf{i}, \mathbf{i}=i_1i_2...i_n\in \{0,1\}^n, \Pi_\mathbf{i} = X_1^{i_1}X_2^{i_2}...X_n^{i_n} \in \mathsf{X}, \Pi^C_\mathbf{i} \in \mathsf{C}^{[n-1]}\mathsf{X}\}$. Here, $\mathsf{X} = \langle X \rangle$ is the group generated by Pauli $X$ gates on all qubits, $\mathsf{C}^{[n-1]}\mathsf{X} = \langle CX, CCX, C^{n-1} X \rangle$ is the group generated by CNOT and multi-qubit Toffoli gates on all qubits.
\end{replemma}

Note that $\langle \cdot \rangle$ denotes the group generated by $\cdot$. Same with before, we simply use $\langle X \rangle$ to represent group $\langle X_1, X_2, \cdots, X_n \rangle$ and $\langle CX, CCX, C^{n-1} X \rangle$ to represent $\langle CX_{12},\cdots  CCX_{123}, \cdots , C^{n-1} X_{1,2,\cdots n} \rangle$ while subscripts label the qubits acted upon. Below we provide the proof of Lemma~\ref{lemma:cosetcardconstraint}.

\begin{proof}\label{lemmapf:cosetcardconstraint}
From Lemma~\ref{lemma:ZinvariantCRU}, we obtain that the space spanned by $\{\mathbb{I}, Z\}^{\otimes n}$ is an invariant subspace of $\mathsf{G}$. By rearranging Pauli operator bases and put $\{\mathbb{I}, Z\}^{\otimes n}$ forward, the Pauli-Liouville representation of $\mathsf{G}$ would be in a block-diagonal form
\begin{equation}
\mathcal{G} = \mathcal{G}_Z \bigoplus \mathcal{G}_{\perp},
\end{equation}
where $\dim\mathcal{G}_Z = 2^n$ and $\dim\mathcal{G}_{\perp} = 4^n-2^n$. Correspondingly, we represent channel $\Lambda$ in a block-diagonal form in the rearranged basis,
\begin{equation}
\Lambda = \begin{pmatrix}
\Lambda_Z & \Lambda_{Z\perp}\\
\Lambda_{\perp Z} & \Lambda_{\perp}
\end{pmatrix}.
\end{equation}
Thus, the twirling of $\Lambda$ over $\mathsf{G}$ is
\begin{equation}
\begin{split}
\Lambda_G &=
\mathbb{E}_{G\in \mathsf{G}} \mathcal{G} \Lambda \mathcal{G}^{-1}\\
&= \mathbb{E}_{G\in \mathsf{G}}
\begin{pmatrix}
\mathcal{G}_Z\Lambda_Z\mathcal{G}_Z^{-1} & \mathcal{G}_Z\Lambda_{Z\perp}\mathcal{G}_{\perp}^{-1}\\
\mathcal{G}_{\perp}\Lambda_{\perp Z}\mathcal{G}_Z^{-1} & \mathcal{G}_{\perp}\Lambda_{\perp}\mathcal{G}_{\perp}^{-1}
\end{pmatrix}
\end{split}
\end{equation}
Due to the direct sum structure, the irreducible representation decomposition of $\mathcal{G}$ is composed of decompositions of $\mathcal{G}_Z$ and $\mathcal{G}_{\perp}$. From Lemma~\ref{lemma:cardconstraint}, all of the irreducible representations contained in $\mathcal{G}$ must be inequivalent to make $\Lambda_G$ diagonal. Therefore, $\mathcal{G}_Z$ and $\mathcal{G}_{\perp}$ do not have any equivalent irreducible representation in common, which leads to $\mathbb{E}_{G\in \mathsf{G}} \mathcal{G}_Z \Lambda_{Z\perp}\mathcal{G}_{\perp}^{-1} = 0$, $\mathbb{E}_{G\in \mathsf{G}} \mathcal{G}_{\perp} \Lambda_{\perp Z}\mathcal{G}_{Z}^{-1} = 0$, and $\Lambda_G = (\mathbb{E}_{G\in \mathsf{G}} \mathcal{G}_Z \Lambda_Z \mathcal{G}_Z^{-1}) \oplus (\mathbb{E}_{G\in \mathsf{G}} \mathcal{G}_{\perp} \Lambda_{\perp} \mathcal{G}_{\perp}^{-1})$. To make $\Lambda_G$ diagonal, we require that two blocks in it are both diagonal. Below we focus on the first block.

Note that $\mathsf{G}_Z$ is a normal subgroup of $\mathsf{G}$. Consider its coset, or quotient group $\mathsf{G} / \mathsf{G}_Z = \{\Pi_i \mathsf{G}_Z, 1\leq i\leq \abs{\mathsf{G} / \mathsf{G}_Z}\}$, where $\forall i$, $\Pi_i$ is a representative element. Without loss of generality, we set $\Pi_1$ as $\mathbb{I}_{2^n}$. Then the other representative elements must be outside $\mathsf{G}_Z$. Then we separate the twirling of $\Lambda_Z$ over $\mathsf{G}$ into two parts. One is the twirling over $\mathsf{G}_Z$, and the other is the twirling over quotient group $\mathsf{G} / \mathsf{G}_Z$. As $Z$-basis gates are all identity in the basis of $\{\mathbb{I},Z\}^{\otimes n}$ in Liouville representation, they do not influence the twirled channel in that subspace. For instance, the Liouville representation of the CS gate is shown in Fig.~\ref{fig:mapcs}. In the subspace spanned by $\{\mathbb{I},Z_1, Z_2, Z_1Z_2\}$, CS gate is equal to identity. Certainly, if the CS gate is a twirling gate, it contributes nothing to the twirling in this subspace. The same are the gates in $\mathsf{G}_Z$.

\begin{figure}[htbp!]
\centering
\includegraphics[width=.7\textwidth]{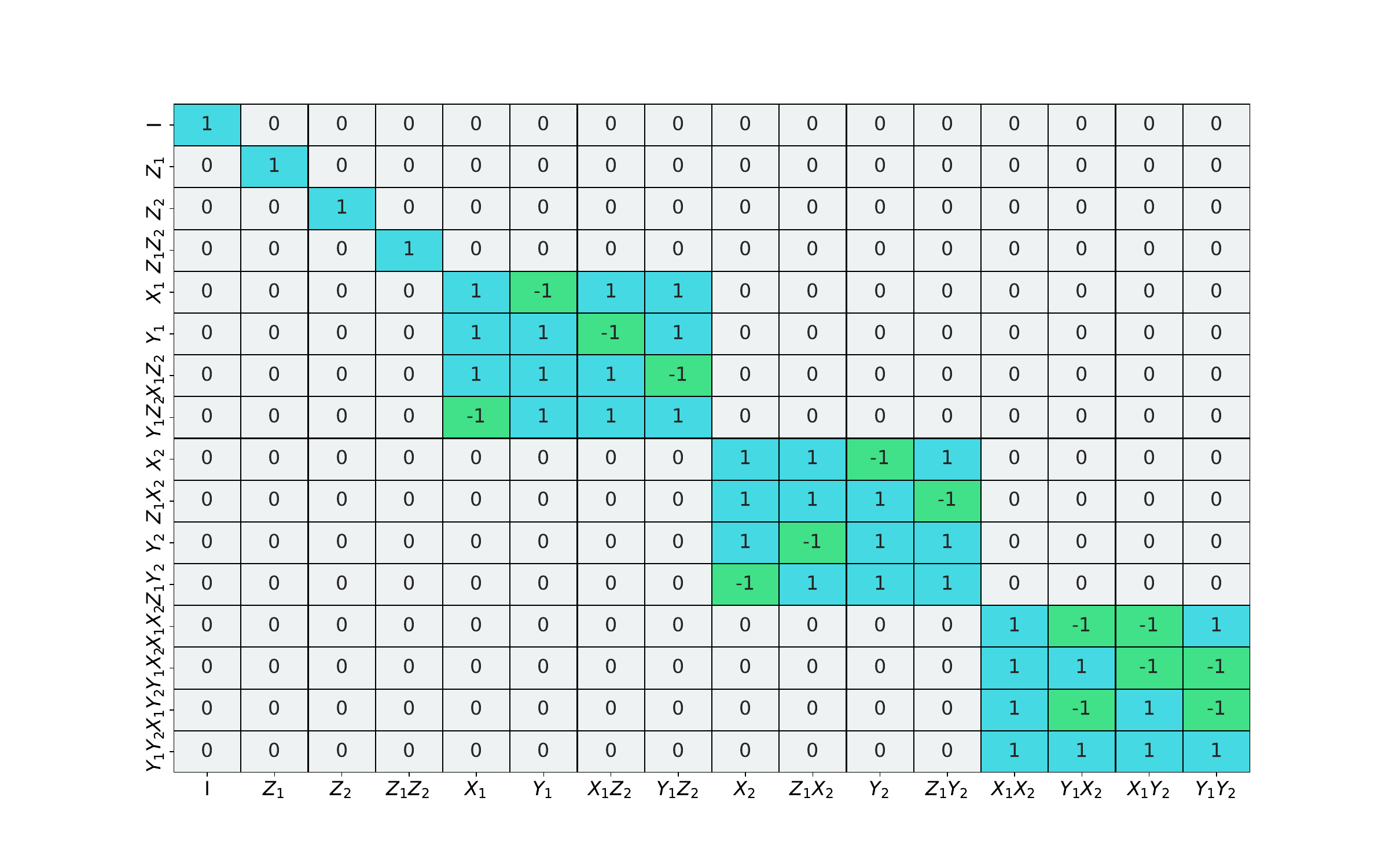}
\caption{Pauli Liouville representation of CS gate.}
\label{fig:mapcs}
\end{figure}

Thus, the only contribution for twirling in this subspace comes from the quotient group. With the similar arguments in Lemma~\ref{lemma:cardconstraint}, we can directly obtain that $\abs{\mathsf{G} / \mathsf{G}_Z} \geq \dim \Lambda_Z = 2^n$, the result in Corollary~\ref{coro:cosetcardconstraint}. Specifically,
\begin{equation}
\begin{split}
\mathbb{E}_{G\in \mathsf{G}}\mathcal{G}_Z\Lambda_Z \mathcal{G}_Z^{-1} &= \mathbb{E}_{1 \leq i \leq \abs{\mathsf{G} / \mathsf{G}_Z}} {\sf \Pi}_{iZ} (\mathbb{E}_{W\in \mathsf{G}_Z} \mathcal{W}_Z\Lambda_Z \mathcal{W}_Z^{-1}) {\sf \Pi}_{iZ}^{-1}\\
&= \mathbb{E}_{1 \leq i \leq \abs{\mathsf{G} / \mathsf{G}_Z}} {\sf \Pi}_{iZ} \Lambda_Z {\sf \Pi}_{iZ}^{-1},
\end{split}
\end{equation}
where the subscript $Z$ denotes the sub-representation of Pauli-Liouville representation in $\{\mathbb{I},Z\}^{\otimes n}$. It can be verified that $\{{\sf \Pi}_{iZ}, 1\leq i \leq \abs{\mathsf{G} / \mathsf{G}_Z}\}$ is a representation of quotient group $\mathsf{G} / \mathsf{G}_Z$. Therefore, one can decompose ${\sf \Pi}_{iZ}$ with irreducible representations of $\mathsf{G} / \mathsf{G}_Z$,
\begin{equation}
{\sf \Pi}_{iZ} = \bigoplus_{j=1}^k \phi_j(\Pi_i).
\end{equation}
The twirled channel $\mathbb{E}_{G\in \mathsf{G}}\mathcal{G}_Z\Lambda_Z \mathcal{G}_Z^{-1}$ would be block diagonal corresponding to the irreducible representation decomposition of ${\sf \Pi}_{iZ}$. With the same arguments in Lemma~\ref{lemma:cardconstraint}, ${\sf \Pi}_{iZ}$ must be multiplicity-free. Then
\begin{equation}
\begin{split}
\abs{\mathsf{G} / \mathsf{G}_Z} &= \sum_{j\in R_{G/Z}} \dim \phi_j^2\\
&\geq \sum_{j=1}^k \dim \phi_j^2\\
&\geq \sum_{j=1}^k \dim \phi_j\\
&= 2^n,
\end{split}
\end{equation}
where $R_{G/Z}$ records all inequivalent irreducible representations of $\mathsf{G} / \mathsf{G}_Z$.

To obtain the result in Lemma~\ref{lemma:cosetcardconstraint}, we further study the irreducible representation decomposition of quotient group $\mathsf{G}/\mathsf{G}_Z$ in the space spanned by $\{\mathbb{I}, Z\}^{\otimes n}$. As mentioned before, ${\sf \Pi}_{iZ}$ must be multiplicity-free, which means each irreducible representation can appear at most once in ${\sf \Pi}_{iZ}$. Focusing on the trivial irreducible representation, its multiplicity in ${\sf \Pi}_{iZ}$ can be obtained via Lemma~\ref{lemma:multiplicity} and is given by
\begin{equation}\label{eq:trimultiplicity}
m_t = \mathbb{E}_{\Pi \in \mathsf{G}/\mathsf{G}_Z} 1\cdot \tr({\sf \Pi}_{Z}),
\end{equation}
where $1$ and $\tr({\sf \Pi}_{Z})$ are characters of trivial irreducible representation and representation in space $\{\mathbb{I}, Z\}^{\otimes n}$, respectively. Through direct calculation, Eq.~\eqref{eq:trimultiplicity} can be simplified to
\begin{equation}\label{eq:trimultiplicity2}
\begin{split}
m_t &= \mathbb{E}_{\Pi} \frac{1}{2^n} \sum_{W\in \{\mathbb{I, Z}^{\otimes n}\}} \tr(W\Pi W\Pi^{\dagger})\\
&= \frac{1}{2^n} \mathbb{E}_{\Pi} \sum_{i_1, i_2, \cdots , i_n = 0}^1 \tr( (\prod_{j=1}^n Z_{j}^{i_j} )\Pi (\prod_{j=1}^n Z_{j}^{i_j} )\Pi^{\dagger})\\
&= \frac{1}{2^n} \mathbb{E}_{\Pi} \sum_{i_1, i_2, \cdots , i_n = 0}^1 \tr( (\prod_{j=1}^n Z_{j}^{i_j} ) (\prod_{j=1}^n \Pi Z_{j}^{i_j}\Pi^{\dagger} ))\\
&= \frac{1}{2^n} \mathbb{E}_{\Pi} \sum_{i_1, i_2, \cdots , i_n = 0}^1 \tr( (\prod_{j=1}^n Z_{j}^{i_j} ) (\prod_{j=1}^n \Pi Z_{j}^{i_j}\Pi^{\dagger} ))\\
&= \frac{1}{2^n} \mathbb{E}_{\Pi} \sum_{i_1, i_2, \cdots , i_n = 0}^1 \tr( (\prod_{j=1}^n Z_{j}^{i_j} \Pi Z_{j}^{i_j}\Pi^{\dagger} ))\\
&= \frac{1}{2^n} \mathbb{E}_{\Pi} \tr( (\prod_{j=1}^n \sum_{i_j = 0}^1 Z_{j}^{i_j} \Pi Z_{j}^{i_j}\Pi^{\dagger} ))\\
&= \frac{1}{2^n} \mathbb{E}_{\Pi} \tr( \prod_{j=1}^n (\mathbb{I} + Z_{j} \Pi Z_{j}\Pi^{\dagger}) ).\\
\end{split}
\end{equation}
As $\Pi \in \mathsf{G}/\mathsf{G}_Z$ is a permutation matrix, $\Pi Z_j \Pi^{\dagger}$, $Z_j\Pi Z_j \Pi^{\dagger}$, and $\mathbb{I} + Z_j\Pi Z_j \Pi^{\dagger}$ are all diagonal matrices. Therefore, Eq.~\eqref{eq:trimultiplicity2} can be further simplified to
\begin{equation}
\begin{split}
m_t &= \frac{1}{2^n} \mathbb{E}_{\Pi} \sum_{\mathbf{i}\in \{0,1\}^n} \prod_{j=1}^n  \bra{\mathbf{i}}\mathbb{I} + Z_{j} \Pi Z_{j}\Pi^{\dagger}\ket{\mathbf{i}} \\
&= \frac{1}{2^n} \mathbb{E}_{\Pi} \sum_{\mathbf{i}\in \{0,1\}^n} \prod_{j=1}^n (1 + \bra{\mathbf{i}} Z_{j} \Pi Z_{j}\Pi^{\dagger}\ket{\mathbf{i}}) \\
&= \frac{1}{2^n} \mathbb{E}_{\Pi} \sum_{\mathbf{i}\in \{0,1\}^n} \prod_{j=1}^n (1 + (-1)^{\mathbf{i}_j}\bra{\mathbf{i}} \Pi Z_{j}\Pi^{\dagger}\ket{\mathbf{i}}) \\
&= \frac{1}{2^n} \mathbb{E}_{\Pi} \sum_{\mathbf{i}\in \{0,1\}^n} \prod_{j=1}^n (1 + (-1)^{\mathbf{i}_j+\pi(\mathbf{i})_j})\\
&= \frac{1}{2^n} \mathbb{E}_{\Pi} \sum_{\mathbf{i}\in \{0,1\}^n} 2^n \delta_{\pi(\mathbf{i})=\mathbf{i}}\\
&= \mathbb{E}_{\Pi} N_{\Pi}\\
&= O_{\mathsf{G}/\mathsf{G}_Z \rightarrow \{0,1\}^n}.\\
\end{split}
\end{equation}
In the fourth line, we utilize that a permutation matrix would transform a bit string into another bit string and denote $\ket{\pi(\mathbf{i})} = \Pi^{\dagger}\ket{\mathbf{i}}$. Here, $\pi$ is a permutation acting on $\{0,1\}^n$ and can be viewed as a representation of $\Pi$. In the sixth line, $N_{\Pi} = \sum_{\mathbf{i}\in \{0,1\}^n} \delta_{\pi(\mathbf{i})=\mathbf{i}}$ denotes the number of fixed points for $\pi$ acting on $\{0,1\}^n$. Via Burnside lemma, $\mathbb{E}_{\Pi} N_{\Pi}$ equals to the number of orbits for group $\mathsf{G}/\mathsf{G}_Z$ acting on $\{0,1\}^n$ which we denote as $O_{\mathsf{G}/\mathsf{G}_Z \rightarrow \{0,1\}^n}$.

As trivial irreducible representation can appear at most once, $m_t$ cannot be larger than 1, or the number of orbits cannot be larger than 1. It implies that gates of $\mathsf{G}/\mathsf{G}_Z$ can interchange any two computational states or any two bit strings in $\{0,1\}^n$. Then, we obtain the results in the informal version of Lemma~\ref{lemma:cosetcardconstraint}.

Now, we obtain that $\mathsf{G}/\mathsf{G}_Z$ can transform $0^n$ to any other bit string in $\{0,1\}^n$. Recall that $\mathsf{G}/\mathsf{G}_Z$ is a subgroup of $\langle X, C^{n-1}X \rangle $, each element in $\mathsf{G}/\mathsf{G}_Z$ can be written as $\Pi^X_i\Pi^C_i$, where $\Pi^X_i\in \mathsf{X} = \langle X\rangle$, $\Pi^C_i \in \langle CX, CCX, \cdots , C^{n-1} X \rangle$. As any element in $\langle CX, CCX, \cdots , C^{n-1} X \rangle$ has no effect on bit string $0^n$, $\Pi^X_i\Pi_i^C$ would simply transform $\ket{0^n}$ to $\Pi^X_i \ket{0^n}$. Thus, $\Pi^X_i$ must take over all elements in $\mathsf{X}$ to transform $0^n$ to all bit strings in $\{0,1\}^n$. Here we complete the proof of Lemma~\ref{lemma:cosetcardconstraint}.
\end{proof}

In the main text, we have discussed some advantages of choosing the CRU subgroup as a twirling group for tailoring a diagonal gate on the computational bases. Below we discuss that it might be sufficient only to consider CRU twirling gates for tailoring diagonal gates. We conjecture that if any finite group $\mathsf{G}$ can make arbitrary noise channels diagonal via twirling, then under the equivalence of unitary transformation, it has a CRU subgroup $\mathsf{G}_C$ also achieving that. A concrete example is $\mathsf{G}$ as the Clifford group and $\mathsf{G}_C$ as the Pauli group. Thus, for diagonal gate $U$ in the computational basis, considering $\mathsf{G}$ in CRU highly likely suffices to find the optimal solution. It is worth noting that the above discussion is only conjecture, and we hope in the future, people can find the smallest group in the whole unitary group for tailoring generic quantum gates.

From Lemma~\ref{lemma:cosetcardconstraint}, we can directly obtain a lower bound for the cardinality of the quotient group, as shown in the following corollary.

\begin{corollary}\label{coro:cosetcardconstraint}
For a finite $n$-qubit CRU subgroup $\mathsf{G}$, set the computational basis subgroup of $\mathsf{G}$ as $\mathsf{G}_Z = \{U\in \mathsf{G} | U \text{ is computational basis} \}$. By Lemma~\ref{lemma:ZinvariantCRU}, $\mathsf{G}_Z$ is a normal subgroup of $\mathsf{G}$. If for any quantum channel $\Lambda$, its twirling over group $\mathsf{G}$, $\Lambda_G$, is diagonal up to a unitary transformation in the Pauli-Liouville representation, then the cardinality of the quotient group $\abs{\mathsf{G}/\mathsf{G}_Z}\geq 2^n$.
\end{corollary}

\subsection{Proof of Theorem~\ref{thm:cnzm}}\label{appendssc:thmcnzm}
Below, we provide the proof of Theorem~\ref{thm:cnzm}. The key point is utilizing that $\mathsf{G}$ contains specific permutation gates and $\mathsf{G}$ is normalized by target gate $U$. We first present a lemma, telling us what is the smallest group containing a given permutation group and normalized by a given computational basis diagonal gate $U$.

\begin{lemma}\label{lemma:smallestgroup}
Given a permutation group, ${\sf \Pi}$, and a diagonal gate, $U$, in the computational basis, the smallest group containing ${\sf \Pi}$ and normalized by $U$ is given by
\begin{equation}
\mathsf{G} = {\sf \Pi} \ltimes \mathsf{W},
\end{equation}
where $\ltimes$ denotes semi-product and $\mathsf{W} = \langle\{\Pi^{\dagger} U \Pi U^{\dagger}, \Pi \in {\sf \Pi}\} \rangle$.
\end{lemma}

\begin{proof}
We first clarify the meaning of semi-product. It means that
\begin{enumerate}
\item  Any element in ${\sf \Pi} \ltimes \mathsf{W}$ has an expression of $\Pi W$ where $\Pi\in {\sf \Pi}$ and $W\in \mathsf{W}$.

\item $\mathsf{W}$ is a normal subgroup of ${\sf \Pi} \ltimes \mathsf{W}$.
\end{enumerate}
It is easy to verify the second condition that $\mathsf{W}$ is a normal subgroup with the first condition. We only need to verify that for any $\Pi$ and $\Pi'$, $\Pi^{\prime \dagger}\Pi^{\dagger} U \Pi U^{\dagger} \Pi' \in \mathsf{W}$. This can be seen via the following equation.
\begin{equation}
\Pi^{\prime \dagger}\Pi^{\dagger} U \Pi U^{\dagger} \Pi' = (\Pi\Pi')^{\dagger} U (\Pi \Pi') U^{\dagger} (\Pi^{\prime \dagger} U \Pi' U^{\dagger})^{\dagger}.
\end{equation}
As $(\Pi\Pi')^{\dagger} U (\Pi \Pi') U^{\dagger}$ and $\Pi^{\prime \dagger} U \Pi' U^{\dagger}$ are both elements in $\mathsf{W}$, we successfully show the soundness of the second condition.

Now we turn to the proof of the lemma. Obviously, ${\sf \Pi} \ltimes \mathsf{W}$ contains ${\sf \Pi}$. It is also normalized by $U$:

Given an element $\Pi W$ where $\Pi\in {\sf \Pi}$ and $W\in \mathsf{W}$,
\begin{equation}
\begin{split}
U\Pi W U^{\dagger} &= U\Pi U^{\dagger} W\\
&= \Pi (\Pi^{\dagger}U\Pi U^{\dagger}) W.
\end{split}
\end{equation}
The first line comes from the fact that $W$ and $U$ are both diagonal gates. As $\Pi^{\dagger}U\Pi U^{\dagger}$ and $W$ both belong to $\mathsf{W}$, $U\Pi W U^{\dagger}$ has an expression of $\Pi W'$ where $W'\in \mathsf{W}$. Combining the condition 1, we obtain that ${\sf \Pi} \ltimes \mathsf{W}$ is normalized by $U$.

Meanwhile, for any group $\mathsf{G}$ containing ${\sf \Pi}$ and normalized by $U$, $U \Pi U^{\dagger}\in \mathsf{G}$ and $\Pi^{\dagger}U \Pi U^{\dagger}\in \mathsf{G}$. Then we know ${\sf \Pi}$ and $\mathsf{W}$ both belong to $\mathsf{G}$. As a consequence, $\mathsf{G}$ must contain ${\sf \Pi} \ltimes \mathsf{W}$. Combining with the arguments before, we prove that ${\sf \Pi} \ltimes \mathsf{W}$ is the smallest group containing ${\sf \Pi}$ and normalized by $U$.
\end{proof}

With Lemma~\ref{lemma:smallestgroup}, we can easily prove Theorem~\ref{thm:cnzm} in the main text. For convenience, we rewrite the theorem below.

\begin{reptheorem}{thm:cnzm}
The optimal twirling group $\mathsf{G}$ in CRU for the multi-qubit controlled phase gate, $U = C^n Z_m$ with $n\geq 1, m\geq 2$ is the smallest group containing $\mathsf{X}$ and normalized by $U$, given by
\begin{equation}
\mathsf{G} = \{\Pi  (\prod_{i=1}^t (\Pi_i^{\dagger} U\Pi_i U^{\dagger})^{l_i} ) | \Pi\in\mathsf{X}, t\in \mathbb{Z}_+, \forall i, l_i\in {\pm 1}, \Pi_i \in \mathsf{X} \}.
\end{equation}
\end{reptheorem}

\begin{proof}
From Lemma~\ref{lemma:cosetcardconstraint}, we can deduce that $\mathsf{G}$ at least contains permutation group $\langle \mathsf{S}\rangle$ where $\mathsf{S} = \{\Pi^X_i\Pi^C_i, \forall \Pi^X_i\in \mathsf{X}, \exists \Pi^C_i \in \mathsf{C}^{[n-1]}\mathsf{X}\}$. Using Lemma~\ref{lemma:smallestgroup}, we obtain that any twirling group $\mathsf{G}$ as a solution to Question~\ref{ques:RB} must satisfy
\begin{equation}
\langle \mathsf{S}\rangle \ltimes \langle \{\Pi^{\dagger} U \Pi U^{\dagger}, \Pi \in \langle \mathsf{S}\rangle\} \rangle \leq \mathsf{G}.
\end{equation}
If $U=C^{n-1}Z_{m} =  \begin{pmatrix}
\mathbb{I}_{2^{n}-1} & \mathbf{0}\\
\mathbf{0} & e^{i\frac{2\pi}{m}}
\end{pmatrix}$, we can show that $\langle \{\Pi^{\dagger} U \Pi U^{\dagger}, \Pi \in \langle \mathsf{X}\rangle\}\rangle\leq \langle \{\Pi^{\dagger} U \Pi U^{\dagger}, \Pi \in \langle \mathsf{S}\rangle\}\rangle$. For brevity, we denote $\mathsf{W} = \langle \{\Pi'^{\dagger} U \Pi' U^{\dagger}, \Pi' \in \langle \mathsf{S}\rangle\} \rangle$ and $\mathsf{W}_X = \langle \{\Pi^{\dagger} U \Pi U^{\dagger}, \Pi \in \langle \mathsf{X}\rangle\}\rangle$. For any generator $\Pi^{\dagger} U \Pi U^{\dagger}$ in $\mathsf{W}_X$ where $\Pi \in \mathsf{X}$, $\Pi^{\dagger} U \Pi$ is a diagonal matrix while only one diagonal element is not 1 but equals $e^{i\frac{2\pi}{m}}$. As the permutation elements in $\mathsf{S}$ can interchange any two computational bases, there must exist an element $\Pi' \in  \mathsf{S} $ such that
\begin{equation}
\Pi^{\prime \dagger} U \Pi' = \Pi^{\dagger} U \Pi.
\end{equation}
Then,
\begin{equation}
\Pi^{\prime \dagger} U \Pi' U^{\dagger} = \Pi^{\dagger} U \Pi U^{\dagger}.
\end{equation}
Thus, any generator of $\mathsf{W}_X$ belongs to $\mathsf{W}$. Then, we obtain that $\mathsf{W}_X \leq \mathsf{W}$ for $C^{n-1}Z_m$. 
This result also applies to the diagonal matrix in which only one diagonal element differs from the others.

Note that the smallest group containing $\mathsf{X}$ and normalized by $U$ is $\mathsf{X}\ltimes \mathsf{W}_X$. As $\mathsf{X}$ is no larger and no global than $\langle \mathsf{S} \rangle$ and $\mathsf{W}_X$ is a subset of $\mathsf{W}$, $\mathsf{X}\ltimes \mathsf{W}_X$ is obviously smaller and more easily implementable than $\langle \mathsf{S}\rangle \ltimes \langle \{\Pi^{\dagger} U \Pi U^{\dagger}, \Pi \in \langle \mathsf{S}\rangle\} \rangle$, which implies that the twirling group cannot be better than $\mathsf{X}\ltimes \mathsf{W}_X$. Below, we would show that for $U=C^nZ_m$ with $n\geq 1, m\geq 2$, $\mathsf{X}\ltimes \mathsf{W}_X$ suffices to tailor $C^nZ_m$. Combining the two sides, we would obtain that the optimal twirling group for $C^nZ_m$ is $\mathsf{X}\ltimes \mathsf{W}_X$, just shown in Theorem~\ref{thm:cnzm}.

In the following, we show $\mathsf{X}\ltimes \mathsf{W}_X$ satisfies Question~\ref{ques:RB} for multi-qubit controlled phase gate $U=C^nZ_{m}$. From the definition of $\mathsf{X}\ltimes \mathsf{W}_X$, we know that $\mathsf{X}\ltimes \mathsf{W}_X$ is normalized by $U$. Thus, we only need to verify that $\mathsf{X}\ltimes \mathsf{W}_X$ can make arbitrary channels diagonal via twirling. With corollaries~\ref{coro:pauli} and~\ref{coro:dihedral}, it is sufficient to show that Pauli $Z$ gate or $Z_m$ gate on each qubit belongs to $\mathsf{W}_X$. Below we show that this is right for $U=C^nZ_{m}$ with $n\geq 1, m\geq 2$.

Recall that $\mathsf{W}_X = \langle \{\Pi^{\dagger} U \Pi U^{\dagger}, \Pi \in \langle \mathsf{X}\rangle\}\rangle$, in the following we analyze the generators of $\mathsf{W}_X$, or $\Pi^{\dagger} U \Pi U^{\dagger}$, in detail. For certain $\Pi\in\langle X\rangle$, let us define its pattern to be a $0$-$1$ bit string, $s_j^{\Pi}$, such that $s_j^{\Pi}=1$ if the $j$-th qubit of $\Pi$ is $I$ and $s_j^{\Pi}=0$ if the $j$-th qubit of $\Pi$ is $X$. Then, the matrix representation of $\Pi^{\dagger}U\Pi U^{\dagger}$ for $U=C^{n}Z_{m}$ is
\begin{equation}
\Pi^{\dagger}C^{n}Z_{m}\Pi C^{n}Z_{m}^{\dagger}=\begin{pmatrix}
1\\&\ddots\\&&e^{i\frac{2\pi}{m}}\\&&&\ddots\\&&&&e^{-i\frac{2\pi}{m}}
\end{pmatrix},
\end{equation}
where the $(s^{\Pi},s^{\Pi})$ entry is $e^{i\frac{2\pi}{m}}$, the $(2^{n+1}-1,2^{n+1}-1)$ entry is $e^{-i\frac{2\pi}{m}}$, and other diagonal entries are $1$. As generators are all diagonal, all elements in $\mathsf{W}_X$ are diagonal in computational basis. Moreover, the diagonal elements would be power of $e^{i\frac{2\pi}{m}}$.  Define an injective map $\phi:\mathsf{W}_X\rightarrow\mathbb{Z}_{m}^{2^{n+1}}$ where $\mathbb{Z}_m = \{0,1,\cdots, m-1\}$,
\begin{equation}\label{eq:Wxmap}
\phi(w)= -i \frac{m}{2\pi} \begin{pmatrix}
\ln [w]_{0,0}\\[10pt]
\vdots\\[10pt]
\ln [w]_{2^{n+1}-1,2^{n+1}-1}
\end{pmatrix}.
\end{equation}

Note that the entries of $\phi(w)$ will always belong to $\{0,1,\cdots, m-1\}$. With map $\phi$, we express the gate in $\mathsf{W}_X$ with a vector. For example, $\phi(\Pi^{\dagger}C^{n}Z_m\Pi C^{n}Z_{m}^{\dagger})=e_{s^{\Pi}}-e_{2^{n+1}-1}$, where we define $e_i$ to be the basis vector whose $i$-th entry is $1$ and other entries are $0$. In this way, the group multiplication is turned into integer vector addition, and the group generation is equivalent to the linear combination of vectors. As $\Pi$ can take any element in $\mathsf{X}$, vectors $v_0=e_0-e_{2^{n+1}-1},\cdots,v_{2^{n+1}-2}=e_{2^{n+1}-2}-e_{2^{n+1}-1}$, are all basis vectors. Note that the overall phase of a quantum gate is not important, so any vectors differing by multiples of $v_{2^{n+1}-1}=\begin{pmatrix}1&\cdots&1\end{pmatrix}^T$ are equivalent. In the next, we show how to construct vectors associated with $Z$ or $Z_m$ gates with vectors $v_0,\cdots,v_{2^{n+1}-1}$.

As different qubits are symmetric under arbitrary permutation, we only need to construct $Z$ gate or $Z_m$ gate on first qubit. If $m$ is odd, we show $Z_m\in \mathsf{W}_X$. Gate $Z_m$ on first qubit corresponds to vector $\phi((Z_{m})_1) = \begin{pmatrix} 0 & \cdots & 0 & 1 & \cdots & 1 \end{pmatrix}^T$. Define
\begin{equation}
u=-(v_0+v_1+\cdots+v_{2^{n+1}-2})+(2^{n+1}-1)v_{2^{n+1}-1}=\begin{pmatrix}
0&\cdots&2^{n+1}
\end{pmatrix}^T.
\end{equation}
When $m$ is odd, $\gcd(2^{n+1},m)=1$, vector $u$ is equivalent to $\begin{pmatrix}0&\cdots&1\end{pmatrix}^T = e_{2^{n+1}-1}$. Then, we can also obtain $e_i = v_i+u$ for $0\leq i\leq 2^{n+1}-2$. With $e_i$ for $0\leq i\leq 2^{n+1}-1$, $\phi((Z_{m})_1)$ can certainly be constructed by
\begin{equation}
\phi((Z_{m})_1) = e_{2^{n+1}-2^n} + e_{2^{n+1}-2^n+1} + \cdots + e_{2^{n+1}-1}.
\end{equation}
Thus, if $m$ is odd, we can ensure that phase gates $Z_m$ on all qubits belong to $\mathsf{W}_X$. In this case, $\mathsf{G}=\mathsf{X}\ltimes \mathsf{W}_X$ definitely satisfies the requirements in Question~\ref{ques:RB}.

When $m$ is even, we show $Z\in \mathsf{W}_X$. We first express $m=q2^k$ where $q\geq 1$ is odd and $k$ is a positive integer. Gate $Z$ on first qubit corresponds to vector $\phi(Z_1) = \begin{pmatrix} 0 & \cdots & 0 & q2^{k-1} & \cdots & q2^{k-1} \end{pmatrix}^T$. As $\gcd(2^{n+1},m)=2^{\min(n+1, k)}$, $u$ can be expressed as $\begin{pmatrix}
0&\cdots&2^{\min(n+1, k)}
\end{pmatrix}^T$. Then, $\phi(Z_1)$ can be constructed by
\begin{equation}
\begin{split}
\phi(Z_1) &= q2^{k-1}(e_{2^{n+1}-2^n} + e_{2^{n+1}-2^n+1} + \cdots + e_{2^{n+1}-1})\\
&=q2^{k-1}(v_{2^{n+1}-2^n} + v_{2^{n+1}-2^n+1} + \cdots + v_{2^{n+1}-2}) + q2^{k-1}(2^n-1)e_{2^{n+1}-1} + q2^{k-1}e_{2^{n+1}-1}\\
&=q2^{k-1}(v_{2^{n+1}-2^n} + v_{2^{n+1}-2^n+1} + \cdots + v_{2^{n+1}-2}) + q2^{n+k-1}e_{2^{n+1}-1}\\
&= q2^{k-1}(v_{2^{n+1}-2^n} + v_{2^{n+1}-2^n+1} + \cdots + v_{2^{n+1}-2}) + q2^{n+k-1-\min(n+1, k)}u.
\end{split}
\end{equation}
As $n\geq 1$, $n+k-1-\min(n+1, k) = \max(n-1,k-2)$ is always a non-negative integer. Thus, $q2^{n+k-1-\min(n+1, k)}$ is an integer. As a consequence, $\phi(Z_1)$ can be constructed with linear combination of $v_0,\cdots,v_{2^{n+1}-2}$, and $u$, which is equivalent to $Z_1\in \mathsf{W}_X$. Thus, if $m$ is even, we can ensure that Pauli $Z$ gates on all qubits belong to $\mathsf{W}_X$. In this case, $\mathsf{G}=\mathsf{X}\ltimes \mathsf{W}_X$ also satisfies the requirements in Question~\ref{ques:RB}. Proof is done.
\end{proof}

\subsection{Twirling groups for multi-qubit controlled phase gates}\label{appendssc:uconjchannel}
Below we prove the following theorem, showing the twirling groups constructed for multi-qubit controlled phase gates in the previous subsection not only satisfies $U\mathsf{G}U^{\dagger} = \mathsf{G}$, but also satisfies $\mathcal{U}\Lambda_G\mathcal{U}^{\dagger}=\Lambda_G$.

\begin{theorem}\label{thm:ucommute}
For $U=C^nZ_{m}$ with $n\geq 1, m\geq 2$ and $nm\neq 2$, the twirling group $\mathsf{G} = \mathsf{X}\ltimes \mathsf{W}_X$ satisfies $\mathcal{U}\Lambda_G\mathcal{U}^{\dagger}=\Lambda_G$.
\end{theorem}

For such quantum gates, our protocol can provide their fidelity estimation accurately instead of providing lower bounds. To prove Theorem~\ref{thm:ucommute}, we present the following lemma.

\begin{lemma}\label{lemma:czdihedral}
As long as $\langle X, CZ, S\rangle\leq \mathsf{G}$ and $U$ is diagonal in the computational basis, the equation $\mathcal{U}\Lambda_G\mathcal{U}^{\dagger}=\Lambda_G$ holds.
\end{lemma}
\begin{proof}
In Pauli-Liouville representation, $\mathcal{U}$ is block-diagonal. Below we investigate the diagonal blocks of $\mathcal{U}$. By definition,
\begin{equation}\label{eq:PauliOfU}
\begin{split}
\mathcal{U}_{ij}=&\tr(\sigma_iU(\sigma_j))\\
=&d\tr(\sigma_iU\sigma_jU^{\dagger}\sigma_j^{\dagger}\sigma_j)\\
=&d\tr(\sigma_j\sigma_iU\sigma_jU^{\dagger}\sigma_j^{\dagger}).
\end{split}
\end{equation}
As $U$ is diagonal in the computational basis or $Z$-basis and $\sigma_j$ is proportional to a Pauli operator, the right part $U\sigma_jU^{\dagger}\sigma_j^{\dagger}$ is also diagonal in the computational basis. Thus, when the left part $\sigma_j\sigma_i$ is not diagonal, the trace of the above expression is $0$. In general, we can express that $\sigma_i = \frac{1}{\sqrt{d}} X_{\mathbf{a}_i}Z_{\mathbf{b}_i}$ and $\sigma_j = \frac{1}{\sqrt{d}} X_{\mathbf{a}_j}Z_{\mathbf{b}_j}$ where $X_{\mathbf{a}_i}, X_{\mathbf{a}_j}\in \{\mathbb{I}, X\}^{\otimes n}$ and $Z_{\mathbf{b}_i}, Z_{\mathbf{b}_j}\in \{\mathbb{I}, Z\}^{\otimes n}$. The notation $\mathbf{a}_i$ is a bit string from $\{0,1\}^n$ meaning $X_{\mathbf{a}_i} = \prod_{k=1}^{n} X_k^{(\mathbf{a}_i)_k}$. The same are for $\mathbf{a}_j$, $\mathbf{b}_i$, and $\mathbf{b}_j$. If $X_{\mathbf{a}_i}\neq X_{\mathbf{a}_j}$, then $\sigma_j\sigma_i$ is not diagonal and hence $\mathcal{U}_{ij}=0$. Thus, for any operator $X_{\mathbf{a}}\in \{\mathbb{I},X\}^{\otimes n}$, there is a corresponding block with dimension $2^n\times 2^n$ in the Liouville representation of the diagonal matrix $U$. The linear space associated with the block is spanned by bases $\{X_{\mathbf{a}}\cdot I^{\otimes n},X_{\mathbf{a}}\cdot(I^{\otimes n-1}\otimes Z),\cdots, X_{\mathbf{a}}\cdot Z^{\otimes n}\}$ and we denote the set of this bases as $X_{\mathbf{a}}\mathsf{Z}$. It is worth noting that the block spanned by $\mathsf{Z}$ is equal to identity with dimension $2^n$, which means this block can be further decomposed into the direct sum of $2^n$ small blocks with dimensions $1$.

Recall that $\Lambda_{G}$ is also diagonal in Liouville representation and all of its blocks are proportional to identities. If $\Lambda_G$ has the same blocks as $\mathcal{U}$, they will commute with each other and hence $\mathcal{U}\Lambda_G\mathcal{U}^{-1}=\Lambda_G$ holds. Now we show that $\langle X,CZ,S\rangle \leq \mathbf{G}$ is enough for $\Lambda_G$ to have the same blocks as $\mathcal{U}$. Since $\langle X,Z\rangle\leq\langle X,CZ,S\rangle \leq \mathsf{G}$, $\Lambda_G$ must be diagonal in Pauli-Liouville representation, $\Lambda_G = \sum_i \lambda_i \lketbra{\sigma_i}{\sigma_i}$. Note that if there exists an element $G\in \mathsf{G}$ such that $\sigma_i=G\sigma_jG^{\dagger}$, then $\sigma_i$ and $\sigma_j$ are symmetric under the group action of $\mathsf{G}$, which results in $\lambda_i=\lambda_j$, or equivalently, $\sigma_i$ and $\sigma_j$ are in the same diagonal block~\cite{Cross2016dihedral}. Note that any $G\in \langle CZ,S\rangle \leq \mathsf{G}$ can be decomposed as $G= WV$ where $W\in\langle CZ\rangle, V\in\langle S\rangle$, then for any $X_{\mathbf{a}}\in \{\mathbb{I},X\}^{\otimes n} / \mathbb{I}^{\otimes n}$,
\begin{equation}\label{eq:ULambdaU=Lambda}
\begin{split}
GX_{\mathbf{a}} G^{\dagger}=& WVX_{\mathbf{a}}V^{\dagger}W^{\dagger}\\
=&X_{\mathbf{a}}  (X_{\mathbf{a}}^{\dagger} WX_{\mathbf{a}}) (X_{\mathbf{a}}^{\dagger}VX_{\mathbf{a}})V^{\dagger}W^{\dagger}\\
=&X_{\mathbf{a}}  (X_{\mathbf{a}}^{\dagger} WX_{\mathbf{a}}W^{\dagger}) (X_{\mathbf{a}}^{\dagger}VX_{\mathbf{a}}V^{\dagger}).
\end{split}
\end{equation}
Below we express $X_{\mathbf{a}} = \bigotimes_{i=1}^n X_i^{a_i}$ where $a_i\in \{0,1\}$ and we simply denote $\mathbf{a} = (a_1, a_2, \cdots, a_n)^T \in \{0,1\}^n$. As $X_{\mathbf{a}}$ is not equal to identity $\mathbb{I}^{\otimes n}$, there is at least an element in $\{a_1, a_2, \cdots, a_n\}$ to be nonzero. Suppose $a_i=1$. For any $Z_{\mathbf{b}}\in \{\mathbb{I},Z\}^{\otimes n}$ where $\mathbf{b} = (b_1, b_2, \cdots, b_n)^T \in \{0,1\}^n$, we can choose $W = \prod_{1\leq j\leq n, j\neq i} CZ_{ij}^{b_j}$, $V = S_i^{\sum_{1\leq j\leq n} a_jb_j}$ such that
\begin{equation}
\begin{split}
X_{\mathbf{a}}(X_{\mathbf{a}}^{\dagger} WX_{\mathbf{a}}W^{\dagger}) (X_{\mathbf{a}}^{\dagger}VX_{\mathbf{a}}V^{\dagger}) &= X_{\mathbf{a}}(\prod_{1\leq j\leq n, j\neq i} Z_j^{b_j}Z_i^{a_jb_j})  (Z_i^{\sum_{1\leq j\leq n} a_jb_j})\\
&= X_{\mathbf{a}}(\prod_{1\leq j\leq n, j\neq i} Z_j^{b_j}) (Z_i^{a_ib_i+ 2\sum_{1\leq j\leq n, j\neq i} a_jb_j})\\
&= X_{\mathbf{a}}\prod_{1\leq j\leq n} Z_j^{b_j}\\
&= X_{\mathbf{a}}Z_{\mathbf{b}}.
\end{split}
\end{equation}
Here we utilize the following two identities,
\begin{align}
X_iCZ_{i,j}X_iCZ_{i,j}^{\dagger}&=Z_j,\\
X_iS_iX_iS_i^{\dagger}&=Z_i.
\end{align}
Then, for any $X_{\mathbf{a}}\in \{\mathbb{I},X\}^{\otimes n} / \mathbb{I}^{\otimes n}$ and $Z_{\mathbf{b}}\in \{\mathbb{I},Z\}^{\otimes n}$, we can find a unitary $G=WV$, such that
\begin{equation}
GX_{\mathbf{a}} G^{\dagger} = X_{\mathbf{a}}Z_{\mathbf{b}}.
\end{equation}
It means that the linear space spanned by $X_{\mathbf{a}}\mathsf{Z}$ is a diagonal block for $\Lambda_G$, which is the same as $\mathcal{U}$. In the space spanned by $\mathsf{Z}$, $\Lambda_G$ and $\mathcal{U}$ are both fully diagonal with $2^n$ one-dimensional blocks. Thus, we prove that $\langle X,CZ,S\rangle \leq \mathsf{G}$ suffices to make $\Lambda_G$ commute with the diagonal matrix $\mathcal{U}$ and complete the proof of Lemma~\ref{lemma:czdihedral}.
\end{proof}

Now we present the proof of Theorem~\ref{thm:ucommute}.
\begin{proof}
The situation of $m = 2$ and that of $m\geq 3$ are different. Below we first discuss the case that $m\geq 3$. The main technique is following the discussion in the proof of Theorem~\ref{thm:cnzm} and mapping $\mathsf{W}_X$ to $\mathbb{Z}_{m}^{2^{n+1}}$ via injective map $\phi$ in Eq.~\eqref{eq:Wxmap}.

Suppose $m\geq 3$. In the case that $m$ is odd for $U = C^nZ_m$, from the proof of Theorem~\ref{thm:cnzm}, we obtain that $\phi(\mathsf{W}_X)$ is spanned by $\{e_i, 0\leq i\leq 2^{n+1}-1\}$. It means that $\phi(\mathsf{W}_X)$ is equal to $\mathbb{Z}_{m}^{2^{n+1}}$. Note that $\phi(C^nZ_m) = e_{2^{n+1}-1}$, so in this case the target gate $U$ itself belongs to $\mathsf{W}_X$ and hence $U$ belongs to the twirling group $\mathsf{G} = \mathsf{X}\ltimes \mathsf{W}_X$. When $U\in \mathsf{G}$, then
\begin{equation}
\begin{split}
\mathcal{U}\Lambda_G \mathcal{U}^{\dagger} &= \mathbb{E}_{G\in \mathsf{G}} \mathcal{U}\mathcal{G}\Lambda \mathcal{G}^{\dagger}\mathcal{U}^{\dagger}\\
&= \mathbb{E}_{G\in \mathsf{G}} \mathcal{G}\Lambda \mathcal{G}^{\dagger}\\
&= \Lambda_G.
\end{split}
\end{equation}
The commutation between $\mathcal{U}$ and $\Lambda_G$ is satisfied automatically.

Below we discuss the case that $m$ is even and $m\geq 3$ for $U=C^nZ_m$. Based on Lemma~\ref{lemma:czdihedral}, we study whether $CZ$ and $S$ gates belong to $\mathsf{W}_X$ or not. By symmetry, we only need to study that for the phase gate on the first qubit $S_1$ and the controlled-phase gate on the first two qubits $CZ_{12}$. Note that the phase gate $S$ can be expressed as $S=Z_m^{m/4}$ and $CZ$ can be expressed as $CZ=CZ_{m}^{m/2}$. We express $m=q2^k$ where $q\geq 1$ is odd and $k\geq 1$ is a positive integer. Note that the bases of $\phi(\mathsf{W}_X)$ are $\{v_i=e_i-e_{2^{n+1}-1}, 0\leq i\leq 2^{n+1}-2\}$ along with $u = 2^{\min(n+1,k)}e_{2^{n+1}-1}$. Then, $\phi(S_1)$ can be constructed by
\begin{equation}\label{eq:S1}
\begin{split}
\phi(S_1) &= q2^{k-2}(e_{2^{n+1}-2^n} + e_{2^{n+1}-2^n+1} + \cdots + e_{2^{n+1}-1})\\
&=q2^{k-2}(v_{2^{n+1}-2^n} + v_{2^{n+1}-2^n+1} + \cdots + v_{2^{n+1}-2}) + q2^{k-2}(2^n-1)e_{2^{n+1}-1} + q2^{k-2}e_{2^{n+1}-1}\\
&=q2^{k-2}(v_{2^{n+1}-2^n} + v_{2^{n+1}-2^n+1} + \cdots + v_{2^{n+1}-2}) + q2^{n+k-2}e_{2^{n+1}-1}\\
&= q2^{k-2}(v_{2^{n+1}-2^n} + v_{2^{n+1}-2^n+1} + \cdots + v_{2^{n+1}-2}) + q2^{n+k-2-\min(n+1, k)}u\\
&= q2^{k-2}(v_{2^{n+1}-2^n} + v_{2^{n+1}-2^n+1} + \cdots + v_{2^{n+1}-2}) + q2^{\max(k-3, n-2)}u.
\end{split}
\end{equation}
Meanwhile, $\phi(CZ_{12})$ can be constructed by
\begin{equation}\label{eq:CZ12}
\begin{split}
\phi(CZ_{12}) &= q2^{k-1}(e_{2^{n+1}-2^{n-1}} + e_{2^{n+1}-2^{n-1}+1} + \cdots + e_{2^{n+1}-1})\\
&=q2^{k-1}(v_{2^{n+1}-2^{n-1}} + v_{2^{n+1}-2^{n-1}+1} + \cdots + v_{2^{n+1}-2}) + q2^{k-1}(2^{n-1}-1)e_{2^{n+1}-1} + q2^{k-1}e_{2^{n+1}-1}\\
&=q2^{k-1}(v_{2^{n+1}-2^{n-1}} + v_{2^{n+1}-2^{n-1}+1} + \cdots + v_{2^{n+1}-2}) + q2^{n+k-2}e_{2^{n+1}-1}\\
&= q2^{k-1}(v_{2^{n+1}-2^{n-1}} + v_{2^{n+1}-2^{n-1}+1} + \cdots + v_{2^{n+1}-2}) + q2^{n+k-2-\min(n+1, k)}u\\
&= q2^{k-1}(v_{2^{n+1}-2^{n-1}} + v_{2^{n+1}-2^{n-1}+1} + \cdots + v_{2^{n+1}-2}) + q2^{\max(k-3, n-2)}u.
\end{split}
\end{equation}
The constructions Eq.~\eqref{eq:S1} and Eq.~\eqref{eq:CZ12} are valid only when the coefficients of the bases are integers. Obviously, $k\geq 3$, and $k = 2$ and $n\geq 2$ can make the coefficients to be integers. In this case, $S_1$ and $CZ_{12}$ belong to $\mathsf{W}_X$ and furthermore, $\langle CZ, X, S\rangle \leq \mathsf{G}$. Based on Lemma~\ref{lemma:czdihedral}, $\mathcal{U}\Lambda_G \mathcal{U}^{\dagger} = \Lambda_G$ is satisfied. The exceptions are the case that $k = 1$ and that $n = 1$ and $k = 2$. Below we discuss the two kinds of exceptions in detail.

When $k = 1$, the target gate $U = C^nZ_m$ and $m=2q$. In this case, $u=2e_{2^{n+1}-1}$, and $CZ_{m/2}$ and $Z_m$ gates belong to $\mathsf{G}$, which can be derived through the following equations.
\begin{equation}
\begin{split}
\phi((Z_m)_1) &= e_{2^{n+1}-2^n} + e_{2^{n+1}-2^n+1} + \cdots + e_{2^{n+1}-1}\\
&= v_{2^{n+1}-2^n} + v_{2^{n+1}-2^n+1} + \cdots + v_{2^{n+1}-2} + (2^n-1)e_{2^{n+1}-1} + e_{2^{n+1}-1}\\
&= v_{2^{n+1}-2^n} + v_{2^{n+1}-2^n+1} + \cdots + v_{2^{n+1}-2} + 2^ne_{2^{n+1}-1}\\
&= v_{2^{n+1}-2^n} + v_{2^{n+1}-2^n+1} + \cdots + v_{2^{n+1}-2} + 2^{n-1}u.
\end{split}
\end{equation}
\begin{equation}
\begin{split}
\phi((CZ_{m/2})_{12}) &= 2(e_{2^{n+1}-2^{n-1}} + e_{2^{n+1}-2^{n-1}+1} + \cdots + e_{2^{n+1}-1})\\
&=2(v_{2^{n+1}-2^{n-1}} + v_{2^{n+1}-2^{n-1}+1} + \cdots + v_{2^{n+1}-2}) + 2(2^{n-1}-1)e_{2^{n+1}-1} + 2e_{2^{n+1}-1}\\
&=2(v_{2^{n+1}-2^{n-1}} + v_{2^{n+1}-2^{n-1}+1} + \cdots + v_{2^{n+1}-2}) + 2^ne_{2^{n+1}-1}\\
&= 2(v_{2^{n+1}-2^{n-1}} + v_{2^{n+1}-2^{n-1}+1} + \cdots + v_{2^{n+1}-2}) + 2^{n-1}u.
\end{split}
\end{equation}
Thus, when $k=1$, $\langle CZ_{m/2}, Z_m \rangle \leq \mathsf{W}_X$ and $\mathsf{X}\ltimes \langle CZ_{m/2}, Z_m \rangle \leq \mathsf{G}$. Notice that $m\geq 3$ and $m/2$ is odd, so $m/2\geq 3$ holds. It can be verified that the irreducible representation decomposition of the Liouville representation for group $\langle CZ_m, Z_m, X \rangle$, where $m\geq 3$ and $m$ is odd, and group $\langle CZ, S, X \rangle$ are the same. As $\mathsf{X}\ltimes \langle CZ_{m/2}, Z_{m/2} \rangle \leq \mathsf{X}\ltimes \langle CZ_{m/2}, Z_m \rangle \leq \mathsf{G}$ and $m/2\geq 3$, after being twirled by $\mathsf{G}$, the twirled noise channel $\Lambda_G$ would have the same blocks as $\mathcal{U}$, and $\mathcal{U}\Lambda_G \mathcal{U}^{\dagger} = \Lambda_G$ is satisfied.

When $n = 1$ and $k = 2$, the target gate $U = CZ_m$ and $m=4q$. The bases for $\phi(\mathsf{W}_X)$ are $v_0 = (1, 0, 0, -1)^T$, $v_1 = (0, 1, 0, -1)^T$, $v_2 = (0, 0, 1, -1)^T$, and $u = (0, 0, 0, 4)$. Note that $\phi(CZ\cdot S_1) = (0, 0, q, 3q)^T = v_2+u$ and $\phi(CZ\cdot S_2) = (0, q, 0, 3q)^T = v_1+u$. Thus, in this case, $\langle CZ\cdot S_1, CZ\cdot S_2 \rangle \in \mathsf{W}_X \leq \mathsf{G}$. In the proof of Theorem~\ref{thm:cnzm}, we have already shown that Pauli group $\mathsf{P}_2\in \mathsf{G}$. Thus, the twirled noise channel would be diagonal in Liouville representation, $\Lambda_G = \sum_i \lambda_i \lketbra{\sigma_i}{\sigma_i}$. Also, we have the following identities,
\begin{align}
CZ\cdot S_1 X_1 (CZ\cdot S_1)^{\dagger} &= X_1 Z_1 Z_2;\\
CZ\cdot S_2 X_1 (CZ\cdot S_2)^{\dagger} &= X_1 Z_2;\\
(CZ\cdot S_2 \cdot CZ\cdot S_1) X_1 (CZ\cdot S_2 \cdot CZ\cdot S_1)^{\dagger} &= X_1 Z_1.
\end{align}
Under the twirling of $\langle CZ\cdot S_1, CZ\cdot S_2 \rangle$, $X_1$, $X_1 Z_1$, $X_1 Z_2$, and $X_1 Z_1Z_2$ would be symmetric. Their corresponding Pauli fidelities $\lambda_i$ would be the same after the twirling. As the twirling group contains $\langle CZ\cdot S_1, CZ\cdot S_2 \rangle$, the twirled noise channel $\Lambda_G$ would be diagonal and proportional to identity in the space spanned by $X_1\mathsf{Z}$. The cases for spaces spanned by $X_2\mathsf{Z}$ or $X_1X_2\mathsf{Z}$ are the same. In summary, $\Lambda_G$ has diagonal blocks in spaces spanned by $X_1\mathsf{Z}$, $X_2\mathsf{Z}$ and $X_1X_2\mathsf{Z}$, respectively. Thus, in this case, $\mathcal{U}\Lambda_G \mathcal{U}^{\dagger} = \Lambda_G$ is also satisfied.

At last, we analyze the case that $m=2$ and $n\geq 2$. The case that $m=2$ and $n=1$ is just CZ and obviously in this case $\mathcal{U}\Lambda_G \mathcal{U}^{\dagger} \neq \Lambda_G$. When the target gate is $C^{n}Z$ with $n\geq 2$, the twirling group is $\mathsf{G}=X\ltimes W_X=\langle C^{n-1}Z,\cdots, Z, X\rangle$. Different from a generic diagonal gate, the diagonal blocks spanned by $X_{\mathbf{a}}\mathsf{Z}$ of $C^{n}Z$ gate in Pauli-Liouville representation can be further divided into two small blocks. We will show that any channel twirled by group $\langle CZ, Z, X\rangle$ would have the same diagonal blocks as $C^{n}Z$ in Liouville representation. As $\langle C^{n-1}Z,\cdots, Z, X\rangle$ contains $\langle CZ, Z, X\rangle$, $\Lambda_{\langle C^{n-1}Z,\cdots, Z, X\rangle}$ would commute with $C^nZ$.

Similar to Eq.~\eqref{eq:PauliOfU}, by setting $\sigma_i = \frac{1}{\sqrt{d}} X_{\mathbf{a}_i}Z_{\mathbf{b}_i}$ and $\sigma_j = \frac{1}{\sqrt{d}} X_{\mathbf{a}_j}Z_{\mathbf{b}_j}$ where $X_{\mathbf{a}_i}, X_{\mathbf{a}_j}\in \{\mathbb{I}, X\}^{\otimes n}$ and $Z_{\mathbf{b}_i}, Z_{\mathbf{b}_j}\in \{\mathbb{I}, Z\}^{\otimes n}$, we get the matrix element of the Liouville representation of $U=C^nZ$,
\begin{equation}
\begin{split}
\mathcal{U}_{ji}&=d\tr(\sigma_i\sigma_jC^{n}Z\sigma_iC^{n}Z\sigma_i^{\dagger})\\
&=\frac{1}{d}\tr(X_{\mathbf{a}_i}Z_{\mathbf{b}_i}X_{\mathbf{a}_j}Z_{\mathbf{b}_j}W)\\
&\propto \frac{1}{d}\tr(X_{\mathbf{a}_i}X_{\mathbf{a}_j}Z_{\mathbf{b}_i}Z_{\mathbf{b}_j}W);\\
W&=C^{n}Z X_{\mathbf{a}_i} C^{n}Z X_{\mathbf{a}_i}.
\end{split}
\end{equation}
Notice that to make $\mathcal{U}_{ji}$ nonzero, we at least require $X_{\mathbf{a}_i}=X_{\mathbf{a}_j}$. Define support of $X_{\mathbf{a}_i}$ to be the set, $\text{supp}(X_{\mathbf{a}_i})=\{i|X_{\mathbf{a}_i}\text{ on qubit $i$ is }X\}$, which is essentially the location of $1$ in $\mathbf{a}_i$. Notice that two of the diagonal elements of $W$ are $-1$ while others are all $1$. Suppose these two elements are $W_{i_1,i_1}$ and $W_{i_2,i_2}$. Without loss of generality, $i_2$ can be set as $2^{n+1}-1$ as the element of the last row and last column is $-1$. To make $\mathcal{U}_{ji}$ nonzero, we furthermore require that $(Z_{\mathbf{b}_i}Z_{\mathbf{b}_j})_{i_1,i_1}=(Z_{\mathbf{b}_i}Z_{\mathbf{b}_j})_{i_2,i_2}$, or $\bra{i_1}Z_{\mathbf{b}_i}Z_{\mathbf{b}_j}\ket{i_1}=\bra{i_2}Z_{\mathbf{b}_i}Z_{\mathbf{b}_j}\ket{i_2}$. Note that $\ket{i_2} = X_{\mathbf{a}_i}\ket{i_1}$. Thus, to make $\mathcal{U}_{ji}$ nonzero, we require that $X_{\mathbf{a}_i}$ and $Z_{\mathbf{b}_i}Z_{\mathbf{b}_j}$ commute, which is equivalent to that $Z_{\mathbf{b}_i}Z_{\mathbf{b}_j}$ has even number of $Z$ gates on qubits of $\text{supp}(X_{\mathbf{a}_i})$. Hence, the block spanned by $X_{\mathbf{a}_i}\mathsf{Z}$ can be further split into two blocks when $m=2$. One block is spanned by $X_{\mathbf{a}_i}\mathsf{Z}_e = \{X_{\mathbf{a}_i}Z_{\mathbf{b}_k} , \mod(\abs{\mathbf{a}_i\cap \mathbf{b}_k}, 2)=0 \}$ and the other block is spanned by $X_{\mathbf{a}_i}\mathsf{Z}_o = \{X_{\mathbf{a}_i}Z_{\mathbf{b}_k} , \mod(\abs{\mathbf{a}_i\cap \mathbf{b}_k}, 2)=1\}$. Here, $\mathbf{a}_i\cap \mathbf{b}_k$ is the bitwise and operation between $\mathbf{a}_i$ and $\mathbf{b}_k$, $\abs{\mathbf{a}_i\cap \mathbf{b}_k}$ is the weight of $\mathbf{a}_i\cap \mathbf{b}_k$, and $\mod(\abs{\mathbf{a}_i\cap \mathbf{b}_k}, 2)=0$ means that there are even number of $Z$ gates of $Z_{\mathbf{b}_k}$ acting on qubits of $\text{supp}(X_{\mathbf{a}_i})$. Notice that $\mod(\abs{\mathbf{a}_i\cap \mathbf{b}_k}, 2)=0$ is equivalent to that $\sum_{1\leq j\leq n} (\mathbf{a}_i)_j (\mathbf{b}_k)_j$ is even and $\mod(\abs{\mathbf{a}_i\cap \mathbf{b}_k}, 2)=1$ is equivalent to that $\sum_{1\leq j\leq n} (\mathbf{a}_i)_j (\mathbf{b}_k)_j$ is odd. These relations are useful in the following proof.

Now we obtain that the Liouville representation of $\mathcal{U}$ is spanned by blocks with bases $X_{\mathbf{a}}\mathsf{Z}_e$ and $X_{\mathbf{a}}\mathsf{Z}_o$. To show $\mathcal{U}\Lambda_G \mathcal{U}^{\dagger} = \Lambda_G$, we also need to check whether $\Lambda_G$ has the same diagonal blocks with $\mathcal{U}$ and whether $\Lambda_G$ is proportional to identity when restricting in these blocks. The proof is similar to that in Lemma~\ref{lemma:czdihedral}. Recall that $\Lambda_G = \sum_i \lambda_i \lketbra{\sigma_i}{\sigma_i}$ in Liouville representation. Also, if there exists an element $G\in \mathsf{G}$ such that $\sigma_i=G\sigma_jG^{\dagger}$, then $\lambda_i=\lambda_j$, and $\sigma_i$ and $\sigma_j$ would be in the same blocks. Then, we only need to check whether any element in $X_{\mathbf{a}}\mathsf{Z}_e$ can be generated from $X_{\mathbf{a}}$ by conjugate action of elements in $\mathsf{G}$.

Given $X_{\mathbf{a}}\in \{\mathbb{I},X\}^{\otimes n} / \mathbb{I}^{\otimes n}$, choose $i$ from $\text{supp}(X_{\mathbf{a}})$. For any $Z_{\mathbf{b}}\in \{\mathbb{I},Z\}^{\otimes n}$ where $X_{\mathbf{a}}Z_{\mathbf{b}}\in X_{\mathbf{a}}\mathsf{Z}_e$, we can choose $W=\prod_{j\neq i}CZ_{ij}^{b_j}$, such that
\begin{equation}
\begin{split}
WX_{\mathbf{a}}W^{\dagger}=&X_{\mathbf{a}}(X_{\mathbf{a}}^{\dagger}WX_{\mathbf{a}}W^{\dagger})\\
=&X_{\mathbf{a}}(\prod_{1\leq j\leq n, j\neq i}Z_j^{b_j}) Z_i^{\sum_{1\leq j\leq n, j\neq i} a_jb_j} \\
=&X_{\mathbf{a}}(\prod_{1\leq j\leq n, j\neq i}Z_j^{b_j}) Z_i^{(\sum_{1\leq j\leq n} a_jb_j)-a_ib_i}\\
=&X_{\mathbf{a}}(\prod_{1\leq j\leq n, j\neq i}Z_j^{b_j}) Z_i^{-a_ib_i}\\
=&X_{\mathbf{a}}(\prod_{1\leq j\leq n, j\neq i}Z_j^{b_j}) Z_i^{b_i}\\
=&X_{\mathbf{a}}Z_{\mathbf{b}}.
\end{split}
\end{equation}
Here, in the fourth line, we utilize $X_{\mathbf{a}}Z_{\mathbf{b}}\in X_{\mathbf{a}}\mathsf{Z}_e$, which means $\sum_{1\leq j\leq n} a_jb_j$ is even. In the fifth line, we use the condition that $a_i=1$. Thus, we show that $\Lambda_G$ is proportional to identity when restricting in the block spanned by $X_{\mathbf{a}}\mathsf{Z}_e$. From a similar argument, it can be verified that this is also true for block spanned by $X_{\mathbf{a}}\mathsf{Z}_o$. Then, we successfully show that $\mathcal{U}$ and $\Lambda_G$ have the same diagonal blocks, and $\mathcal{U}\Lambda_G\mathcal{U}^{-1}=\Lambda_G$ when $m=2$.

As a conclusion, the twirling group $X\ltimes \mathsf{W}_X$ satisfies $\mathcal{U}\Lambda_G\mathcal{U}^{\dagger}=\Lambda_G$ for $U=C^nZ_{m}$ with $n=1, m\geq 3$ or $n\geq 2, m\geq 2$. Proof is done.
\end{proof}

\subsection{Structure of $\mathsf{W}_X$}
In this part, we present the concrete form of $\mathsf{W}_X$. Following the discussion of $\mathsf{W}_X$ in previous subsections, we map $\mathsf{W}_X$ into $\mathbb{Z}_{m}^{2^{n+1}}$ via Eq.~\eqref{eq:Wxmap}.

We first discuss the cardinality of $\mathsf{W}_X$. The image
$\phi(\mathsf{W}_X)$ is the linear span of bases $\{v_0=e_0-e_{2^{n+1}-1},\cdots,v_{2^{n+1}-2}=e_{2^{n+1}-2}-e_{2^{n+1}-1}\}$. If considering the factor of global phase, the image would be the linear span of bases $\{v_0=e_0-e_{2^{n+1}-1},\cdots,v_{2^{n+1}-2}=e_{2^{n+1}-2}-e_{2^{n+1}-1}, u = 2^{\min(n+1, k)}e_{2^{n+1}-1}\}$ and we denote this enlarged image to be $\phi'(\mathsf{W}_X)$. In $\phi'(\mathsf{W}_X)$, any two vectors differing by $v_{2^{n+1}-1}=\begin{pmatrix}1&\cdots&1\end{pmatrix}^T$ correspond to the same gate. If $n+1\geq k$, $u$ is $\begin{pmatrix}0&\cdots&2^{k}\end{pmatrix}^T$. Now we consider the cardinality of $\phi'(\mathsf{W}_X)$. One can take arbitrary values from $\mathbb{Z}_m$ to fix the coefficients of the first $2^{n+1}-1$ bases. After that, the coefficient of the last basis $u$ only has $\frac{m}{2^k}$ inequivalent choices from $\mathbb{Z}_m$ as $(a+\frac{m}{2^k})u=au$. Hence, the cardinality of $\phi'(\mathsf{W}_X)$ is
\begin{equation}
|\phi'(\mathsf{W}_X)|=m^{2^{n+1}-1}\cdot\frac{m}{2^k}=m^{2^{n+1}}/2^k.
\end{equation}
For any $\phi(w)$ belongs to $\phi'(\mathsf{W}_X)$, we can also find $\phi(w)+\begin{pmatrix}1&\cdots&1\end{pmatrix}^T,\cdots,\phi(w)+(m-1)\cdot\begin{pmatrix}1&\cdots&1\end{pmatrix}^T$ in $\phi'(\mathsf{W}_X)$. All these $m$ terms correspond to the same gate. Thus, the cardinality of $\phi(\mathsf{W}_X)$ is just $|\phi'(\mathsf{W}_X)|$ divided by $m$, excluding the redundacy of the global phase,
\begin{equation}\label{eq:cardWx}
|\mathsf{W}_X|=|\phi(\mathsf{W}_X)|=\frac{1}{m}|\phi'(\mathsf{W}_X)|=m^{2^{n+1}-1}/2^k.
\end{equation}

Now we investigate the structure of $\mathsf{W}_X=\langle\Pi^{\dagger} C^{n}Z_{m}\Pi C^{n}Z_{m}^{\dagger}\rangle$. We first analyze the simple case of $m=q$ and $m=2q$ where $q$ is an odd number. And then we use induction to analyze the general case that $m=q2^k$.

\begin{enumerate}[1)]
\item
Case 1: $m=q$.
\\
In this case,  $k = 0$, $\gcd(m, 2^{n+1})=1$, the vector $u=e_{2^{n+1}-1}$, and $\phi'(\mathsf{W}_X) = \mathbb{Z}_{m}^{2^{n+1}}$. Clearly, for any $0\leq l\leq n$, $\phi(C^l Z_m)\in \phi'(\mathsf{W}_X)$, so we can obtain that $\langle C^{n}Z_{m}, C^{n-1}Z_{m}, \cdots, CZ_{m}, Z_{m} \rangle \leq \mathsf{W}_X$. Since $\langle C^{n}Z_{m}, C^{n-1}Z_{m}, \cdots, CZ_{m}, Z_{m} \rangle = \langle C^{n}Z_{m}\rangle\times\langle C^{n-1}Z_{m}\rangle\times\cdots\times \langle Z_{m}\rangle$, $\abs{\langle C^{n}Z_{m}, C^{n-1}Z_{m}, \cdots, CZ_{m}, Z_{m} \rangle} = m^{\binom{n+1}{n+1}}\times m^{\binom{n+1}{n}}\times\cdots\times m^{\binom{n+1}{1}}=m^{2^{n+1}-1} = \abs{\mathsf{W}_X}$. The cardinalities of $\mathsf{W}_X$ and its subgroup $\langle C^{n}Z_{m}, C^{n-1}Z_{m}, \cdots, CZ_{m}, Z_{m} \rangle$ are the same, hence, we can conclude that
\begin{equation}
\mathsf{W}_X = \langle C^{n}Z_{m}, C^{n-1}Z_{m}, \cdots, CZ_{m}, Z_{m} \rangle = \langle C^{n}Z_{m}\rangle\times\langle C^{n-1}Z_{m}\rangle\times\cdots\times \langle Z_{m}\rangle.
\end{equation}

\item\label{item:case2}
Case 2: $m=2q$.
\\
In this case, $k=1$, $\gcd(m, 2^{n+1})=2$, and $u=2e_{2^{n+1}-1}$. We can use the following formula to construct all $C^lZ_{m}$ with $0\leq l\leq n-1$,
\begin{equation}
C^lZ_{m}=\phi^{-1}\Bigl(v_{(2^{l+1}-1)*2^{n-l}}+v_{(2^{l+1}-1)*2^{n-l}+1}+\cdots+v_{2^{n+1}-2}+2^{n-l-1}\cdot u\Bigr).
\end{equation}
But this construction is not valid for $C^{n}Z_{m}$ as it requires $\frac{1}{2}u$. We can only construct $C^{n}Z_{m}^2=C^{n}Z_{m/2}$. This implies that $\langle C^{n}Z_{m/2},C^{n-1}Z_{m},C^{n-2}Z_{m},\cdots\rangle\leq \mathsf{W}_X$. Since $\langle C^{n}Z_{m/2},C^{n-1}Z_{m},\cdots\rangle=\langle C^{n}Z_{m/2}\rangle\times\langle C^{n-1}Z_{m}\rangle\times\cdots\times \langle Z_{m}\rangle$, $\abs{\langle C^{n}Z_{m/2}, C^{n-1}Z_{m}, \cdots, CZ_{m}, Z_{m} \rangle} = (m/2)^{\binom{n+1}{n+1}}\times m^{\binom{n+1}{n}}\times\cdots\times m^{\binom{n+1}{1}}=m^{2^{n+1}-1}/2 = \abs{\mathsf{W}_X}$. Thus, in this case,
\begin{equation}
\mathsf{W}_X = \langle C^{n}Z_{m/2}, C^{n-1}Z_{m}, \cdots, CZ_{m}, Z_{m} \rangle = \langle C^{n}Z_{m/2}\rangle\times\langle C^{n-1}Z_{m}\rangle\times\cdots\times \langle Z_{m}\rangle.
\end{equation}

\item
Case 3.1: $m=q2^k$ with $k\leq n$.
\\
We first give the result and prove it by induction.
\begin{theorem}
If $m=q2^k$ and $k\leq n$, the structure of $\mathsf{W}_X$ can be expressed in the following recursion,
\begin{equation}
\begin{split}
\mathsf{W}_X&=\langle C^{n-k}Z_{m}, C^{n-k-1}Z_{m}, \cdots, Z_m \rangle\times \langle A_k\rangle;\\
A_0&=\{\mathbb{I}\}, A_1=\{C^{n}Z_{m/2}\},\\ \forall 2\leq k\leq n, A_k&=\{C^{n-k+1} Z_{m/2}, gC^{n-k+1}Z_m| g\in A_{k-1} \}.
\end{split}
\end{equation}
\end{theorem}

\begin{proof}
The proof idea is showing $|\mathsf{W}_X| = |\langle C^{n-k}Z_{m}, C^{n-k-1}Z_{m}, \cdots, Z_m \rangle\times \langle A_k\rangle|$ and $\langle C^{n-k}Z_{m}, C^{n-k-1}Z_{m}, \cdots, Z_m \rangle\times \langle A_k\rangle \leq \mathsf{W}_X$.

Let us start with studying the relation between $\langle A_{k-1}\rangle$ and $\langle A_{k}\rangle$. For any $g\in \langle A_{k-1}\rangle$, it can be decomposed as $g=g_1^{\lambda_1}g_2^{\lambda_2}\cdots$ where $g_1,g_2,\cdots$ all belong to $A_{k-1}$ and $\lambda_1,\lambda_2,\cdots$ are integers. Since $g_1C^{n-k+1}Z_{m},g_2C^{n-k+1}Z_{m},\cdots\in A_k$, we can conclude that $g'=(g_1C^{n-k+1}Z_{m})^{\lambda_1}(g_2C^{n-k+1}Z_{m})^{\lambda_2}\cdots=gC^{n-k+1}Z_{m}^{\sum_{i}\lambda_i}\in\langle A_{k}\rangle$. Suppose $\sum_{i}\lambda_i$ is an even number, since $C^{n-k+1}Z_{m/2}=C^{n-k+1}Z_{m}^2\in A_{k}$, the set $\{g, gC^{n-k+1}Z_{m}^2, gC^{n-k+1}Z_{m}^4,\cdots\}\subseteq\langle A_k\rangle$. Now we want to argue that any gate has the form of $gC^{n-k+1}Z_{m}^{2p+1}$ cannot be in $\langle A_{k}\rangle$. This claim can be proved by contradiction. If $gC^{n-k+1}Z_{m}^{2p+1}\in\langle A_k\rangle$, combining with $gC^{n-k+1}Z_{m}^{2p}\in\langle A_k\rangle$, it is clear that $C^{n-k+1}Z_{m}=(gC^{n-k+1}Z_{m}^{2p})^{\dagger}gC^{n-k+1}Z_{m}^{2p+1}\in\langle A_k\rangle$. It leads to a contradiction as $C^{n-k+1}Z_{m}$ cannot be constructed with $\{v_i, 0\leq i\leq 2^{n+1}-2\}$ and $u = 2^ke_{2^{n+1}-1}$. Similarly, if $\sum_{i}\lambda_i$ is an odd number, we can see that all elements has the form $gC^{n-k+1}Z_{m}^{2p+1}$ will belong to $\langle A_k\rangle$ while elements within form of $gC^{n-k+1}Z_{m}^{2p}$ will not. In either case, one element $g$ in $\langle A_{k-1}\rangle$ will contribute to $|\langle C^{n-k+1}Z_{m/2}\rangle|=\frac{1}{2}|\langle C^{n-k+1}Z_{m}\rangle|$ elements in $\langle A_k\rangle$, which implies that $|\langle A_k\rangle|=\frac{1}{2}|\langle A_{k-1}\rangle|\times|\langle C^{n-k+1}Z_{m}\rangle|$. By iteration, the cardinality of $\langle C^{n-k}Z_{m}, C^{n-k-1}Z_{m},\cdots\rangle\times \langle A_k\rangle$ is
\begin{equation}\label{eq:cardequal}
\begin{split}
|\langle C^{n-k}Z_{m}, C^{n-k-1}Z_{m},\cdots\rangle\times \langle A_k\rangle|=&\frac{1}{2}|\langle C^{n-k+1}Z_{m}, C^{n-k}Z_{m},\cdots\rangle\times \langle A_{k-1}\rangle|\\
=&\frac{1}{2^k}|\langle C^{n}Z_{m}, C^{n-1}Z_{m},\cdots\rangle|\\
=&m^{2^{n+1}-1}/2^k.
\end{split}
\end{equation}
This is exactly equal to $|\mathsf{W}_X|$. To complete the proof, we only need to argue that $\langle C^{n-k}Z_{m}, C^{n-k-1}Z_{m},\cdots\rangle\times \langle A_k\rangle\leq \mathsf{W}_X$. Note that $\phi'(\mathsf{W}_X)$ is spanned by bases ${v_i, 1\leq i\leq 2^{n+1}-2}$ and $u$. We only need to investigate whether the generators of $\langle C^{n-k}Z_{m}, C^{n-k-1}Z_{m},\cdots\rangle\times \langle A_k\rangle$ belong to the preimage set of $\phi$.

For $C^{n-k-l}Z_{m}$ with $0\leq l\leq n-k$, it can be expressed in the form of
\begin{equation}
\begin{split}
C^{n-k-l}Z_{m} &=\phi^{-1}\Bigl(2^{l}\begin{pmatrix}0&\cdots&2^{k}\end{pmatrix}^T+\sum_{i=(2^{n+1-k-l}-1)\cdot2^{k+l}}^{2^{n+1}-2}v_i\Bigr)\\
&=\phi^{-1}\Bigl(2^{l}u+\sum_{i=(2^{n+1-k-l}-1)\cdot2^{k+l}}^{2^{n+1}-2}v_i\Bigr).
\end{split}
\end{equation}
Thus, $C^{n-k-l}Z_{m}$ belongs to $\mathsf{W}_X$. Next we prove $\langle A_k \rangle \leq \mathsf{W}_X$, or equivalently, $A_k \subseteq \mathsf{W}_X$, by induction. Denote $\mathsf{W}_X^k = \langle\Pi^{\dagger} C^{n}Z_{q2^k}\Pi C^{n}Z_{q2^k}^{\dagger}\rangle$. Suppose $A_{k-1}\subseteq \mathsf{W}_X^{k-1}$. Then, any element $g\in A_{k-1}$ has the following expression,
\begin{equation}
g=\phi^{-1}\Bigl(\begin{pmatrix}0&\cdots& 2^{k-1}\end{pmatrix}^T+\sum_{i=0}^{2^{n+1}-2}c_iv_i\Bigr).
\end{equation}
Note that $C^{n-k+1}Z_{m}$ can also be expressed in the above form as
\begin{equation}
C^{n-k+1}Z_{m}=\phi^{-1}\Bigl(\begin{pmatrix}0&\cdots&2^{k-1}\end{pmatrix}^T+\sum_{i=(2^{n-k+2}-1)\cdot 2^{k-1}}^{2^{n+1}-2}v_i\Bigr).
\end{equation}
This implies that $C^{n-k+1}Z_{m}^2$ and $C^{n-k+1}Z_{m}g$ can be expressed as below.
\begin{align}
C^{n-k+1}Z_{m}^2=&\phi^{-1}\Bigl(\begin{pmatrix}0&\cdots&2^{k}\end{pmatrix}^T+\sum_{i=(2^{n-k+2}-1)\cdot 2^{k-1}}^{2^{n+1}-2}2v_i\Bigr);\\
gC^{n-k+1}Z_{m}=&\phi^{-1}\Bigl(\begin{pmatrix}0&\cdots&2^{k}\end{pmatrix}^T+\sum_{i=0}^{(2^{n-k+2}-1)\cdot 2^{k-1}-1}c_iv_i+\sum_{i=(2^{n-k+2}-1)\cdot 2^{k-1}}^{2^{n+1}-2}(c_i+1)v_i\Bigr).
\end{align}
Therefore, any element $g'\in A_k$ has the expression,
\begin{equation}
g'=\phi^{-1}\Bigl(\begin{pmatrix}0&\cdots& 2^{k}\end{pmatrix}^T+\sum_{i=0}^{2^{n+1}-2}c_iv_i\Bigr).
\end{equation}
It means that all elements of $A_{k}$ can be constructed with $u$ and $v_i$, and belong to $\mathsf{W}_X^k$. 
Now we have shown that $A_{k-1}\subseteq \mathsf{W}_X^{k-1}$ implies $A_{k}\subseteq \mathsf{W}^k_X$. For the induction argument to hold, we only need to check the initial condition, i.e., $A_1 \subseteq \mathsf{W}_X^{1}$. This reduces to Case~\ref{item:case2}. Thus, we prove that $A_k\subseteq \mathsf{W}_X$ and $\langle A_k\rangle\leq \mathsf{W}_X$. Hence, $\langle C^{n-k}Z_{m}, C^{n-k-1}Z_{m},\cdots,Z_m\rangle\times \langle A_k\rangle\leq \mathsf{W}_X$. Combining with Eq.~\eqref{eq:cardequal}, we obtain that $\mathsf{W}_X = \langle C^{n-k}Z_{m}, C^{n-k-1}Z_{m},\cdots\rangle\times \langle A_k\rangle$ and proof is done.
\end{proof}

\item
Case 3.2: $m=q2^k$ with $k>n$.
\\
In this case, $\gcd(2^{n+1},m)=2^{n+1}$ and $\mathsf{W}_X=\langle A_{n+1}\rangle$ as described in the previous case.

Below we present examples of the twirling groups for target gate $U$ to be $C^nZ$ and $CZ_m$, which have been mentioned in the main text. In the next subsection, we will analyze their sample complexity and computational complexity.

\end{enumerate}

\begin{example}\label{exp:cnz}
The optimal twirling group $\mathsf{G}$ in CRU for multi-qubit controlled $Z$ gate $C^n Z$ is
\begin{equation}
\langle C^{n-1}Z, C^{n-2}Z, \cdots, CZ, Z, X \rangle = \mathsf{X}\ltimes \langle C^{n-1}Z, C^{n-2}Z, \cdots, CZ, Z\rangle.
\end{equation}
\end{example}

\begin{example}\label{exp:czm}
The optimal twirling group $\mathsf{G}$ in CRU for controlled phase gate $CZ_m$ is
\begin{equation}
\langle CZ_m, Z_m, X \rangle = \mathsf{X}\ltimes \langle CZ_m, Z_m\rangle,
\end{equation}
if $m$ is odd, and
\begin{equation}
\langle CZ_{m/2}, Z_m, X \rangle = \mathsf{X}\ltimes \langle CZ_{m/2}, Z_m\rangle,
\end{equation}
if $m$ is even and $m/2$ is odd, and
\begin{equation}
\langle CZ_{m/4}, Z_{m/2}, X, CZ_{m/2}Z^1_{m}, CZ_{m/2}Z^2_{m} \rangle = \mathsf{X}\ltimes \langle CZ_{m/4}, Z_{m/2}, CZ_{m/2} Z^1_{m}, CZ_{m/2}Z^2_{m} \rangle,
\end{equation}
if $m/2$ is even. It is worth mentioning that $\langle CZ_{m/4}, Z_{m/2}, CZ_{m/2} Z^1_{m}, CZ_{m/2}Z^2_{m} \rangle$ is a subgroup of $\langle CZ_{m/2}, Z_m\rangle$.
\end{example}

\subsection{Complexity analysis}\label{appendssc:complexity}
To enable randomized benchmarking or any other quantum information tasks with twirling groups, one needs to sample from the group and also needs to compute the multiplication of group elements. It is necessary to analyze the sample complexity and the computational complexity of the twirling group. The sample complexity is directly related to the cardinality of the corresponding twirling group. The computational complexity is related to both the group structure and the algorithm for computing the multiplication. We will provide a group multiplication algorithm and present its complexity as the upper bound of the group computational complexity. In the discussion below, we distinguish the number of qubits $N$ and the number of controlled qubits $n$. In general, $n\le N-1$.

To benchmark $C^{n}Z$ gate, the twirling group is $\mathsf{G}_{C^{n}Z} = \langle C^{n-1}Z,C^{n-2}Z\cdots,CZ,Z,X\rangle$ shown in Example~\ref{exp:cnz}. Note that this group has a semi-direct product structure,
\begin{equation}
\langle C^{n-1}Z, C^{n-2}Z, \cdots, CZ, Z, X\rangle=\langle X\rangle \rtimes(\langle C^{n-1}Z\rangle\times\langle C^{n-2}Z\rangle\times\cdots\langle Z\rangle).
\end{equation}
This means that any element in $\mathsf{G}_{C^{n}Z}$ can be written in the form of $\Pi W_{n-1}W_{n-2}\cdots W_{1}W_{0}$ where $\Pi\in\langle X\rangle$ and $W_{l}\in\langle C^{l}Z\rangle$. To sample from this twirling group $\mathsf{G}_{C^{n}Z}$, we just need to sample independently from $\langle C^{n-1}Z\rangle$, $\langle C^{n-2}Z\rangle$, $\cdots$, $\langle Z\rangle$, and $\langle X\rangle$. Since the orders of all generators are $2$, we can use a binary string whose length equals the number of generators of the group to represent an arbitrary group element. Then, we can sample the group element by sampling the binary string. For instance, the generators of $\langle CZ\rangle$ are $\{CZ_{1,2},CZ_{1,3},\cdots,CZ_{1,N},\cdots,CZ_{2,3},\cdots,CZ_{N-1,N}\}$. So a binary string with length $\binom{N}{2}$ is enough for sampling from $\langle CZ\rangle$. The total length needed to sample $\mathsf{G}_{C^{n}Z}$ is given by
\begin{equation}
\log |\mathsf{G}_{C^{n}Z}| = \log{\bigl|\langle X\rangle\rtimes\bigl(\langle C^{n-1}Z\rangle\times\cdots\times\langle CZ\rangle\times\langle Z\rangle\bigr)\bigr|}=n+\sum_{i=1}^{n}\binom{N}{i}=O(N^{n}),
\end{equation}
which is indeed the sample complexity. For computing the inverse gates, we will utilize the following identity. Given a $C^{n}Z$ gate acting on qubits $i_1,\cdots,i_{n+1}$ and a subset of $\{i_1,\cdots,i_{n+1}\}$, $I$, then
\begin{equation}\label{eq:swapidentity}
C^{n}Z_{i_1,i_2,\cdots,i_{n+1}}\Pi_{i\in I}X_i=\Pi_{i\in I}X_i\Pi_{S\subseteq I}C^{n-|S|}Z_{\{i_1,\cdots,i_{n+1}\}\backslash S}.
\end{equation}
So the multiplication of two elements in the group is calculated by
\begin{equation}
\begin{split}
&\Pi^{(1)}W_{n-1}^{(1)}\cdots W_{0}^{(1)} \Pi^{(2)}W_{n-1}^{(2)}\cdots W_{0}^{(2)}\\
=&(-1)^{\Pi^{(2)}\cdot W_{0}^{(1)}}\Pi^{(1)}W_{n-1}^{(1)}\cdots W_{1}^{(1)}\Pi^{(2)}W_{0}^{(1)}W_{n-1}^{(2)}\cdots W_{0}^{(2)}\\
=&(-1)^{\Pi^{(2)}\cdot W_{0}^{(1)}}\Pi^{(1)}W_{n-1}^{(1)}\cdots W_{2}^{(1)}\Pi^{(2)}W'_{0}W_{1}^{(1)}W_{0}^{(1)}W_{n-1}^{(2)}\cdots W_{0}^{(2)}\\
=&(-1)^{\Pi^{(2)}\cdot W_{0}^{(1)}}\Pi^{(1)}\Pi^{(2)} W'_{n-2}W_{n-1}^{(1)}\cdots W'_{0}W_{1}^{(1)}W_{0}^{(1)}W_{n-1}^{(2)}W_{n-2}^{(2)}\cdots W_{1}^{(2)}W_{0}^{(2)}.
\end{split}
\end{equation}
The $W'_0, \cdots, W'_{n-2}$ in the third and fourth lines are indeed the additional controlled-Z gates and multi-qubit controlled-Z gates in Eq.~\eqref{eq:swapidentity}.

Now we give a brief analysis of the total complexity. The key point is utilizing Eq.~\eqref{eq:swapidentity} to shift $\Pi^{(2)}$ gate to the left. We first fix an integer $l$ and focus on $C^{l}Z$. Now we try to study the swapping between $C^lZ$ and $\Pi^{(2)}$. By Eq.~\eqref{eq:swapidentity}, the total complexity for shifting $\Pi^{(2)}$ across a $C^{l}Z$ gate is $O(2^{l+1})$ since the cardinality of subset $I$ is at most $l+1$ and the number of subsets of $I$ is at most $2^{l+1}$. As the number of $C^{l}Z$ gates on $N$ qubits is at most $\binom{N}{l+1}$, the total complexity for shifting $\Pi^{(2)}$ to the leftmost is
\begin{equation}\label{eq:multiplycomplexity}
\text{group multiplication complexity}=\sum_{l=0}^{n-1}\binom{N}{l+1}\cdot 2^{l+1}.
\end{equation}
The right part is bounded by $N^{n}$. Hence the total group multiplication complexity is $O(N^{n})$. Notice that this bound is quite untight. For example if we let $n=N$, the right part of Eq.~\eqref{eq:multiplycomplexity} is indeed $3^{n}-1$.

The inverse of an element $\Pi W_{n-1}\cdots W_{0}$ is $W_{n-1}\cdots W_{0}\Pi$, which equals the multiplication of element $W_{n-1}\cdots W_{1}1_{\langle X\rangle}$ and $1_{\langle C^{n-1}Z\rangle}\cdots 1_{\langle Z\rangle}\Pi$, where $1_{\mathsf{G}}$ represents the identity in group $\mathsf{G}$ and is indeed the all-zero string in the binary string representation of group elements. The computational complexity for computing inverse gate is then equal to the complexity for group element multiplication, i.e. $O(N^{n})$.

To benchmark $CZ_m$ gate, the twirling group by our protocol would always be a subgroup of $\mathsf{G}_{CZ_m}=\langle CZ_m, Z_m, X\rangle$. For simplicity, we only analyze the complexity of $\langle CZ_m, Z_m, X\rangle$. This complexity would be the upper bound of the complexity of the twirling group shown in Example~\ref{exp:czm}. $\langle CZ_m, Z_m, X\rangle$ has a similar direct product structure like $\langle C^{n-1}Z, C^{n-2}Z, \cdots, CZ, Z, X\rangle$,
\begin{equation}
\mathsf{G}_{CZ_{m}}=\langle X\rangle\rtimes(\langle CZ_m\rangle\times\langle Z_m\rangle).
\end{equation}
Thus, each element in $\mathsf{G}_{CZ_m}$ can be expressed as $W_2W_1\Pi$ where $W_2\in\langle CZ_m\rangle,W_1\in\langle Z_m\rangle$ and $\Pi\in\langle X\rangle$. In this case, the orders of generators $CZ_m$ and $Z_m$ are both $m$. We need to use $\lceil\log{m}\rceil$ bits to record their multiplicities and to uniquely identify one group element. Then, to represent an element in $\langle CZ_m\rangle$, we should use a bit string with length $\lceil\log{m}\rceil N(N-1)/2$. The total number of bits to express a group element is
\begin{equation}
\log\bigl|\langle X\rangle\rtimes(\langle CZ_m\rangle\times\langle Z_m\rangle)\bigr|=\left(\frac{N(N-1)}{2}+N\right)\lceil\log{m}\rceil+N=O(N^2\log{m}),
\end{equation}
which is indeed the sample complexity. The equation Eq.~\eqref{eq:swapidentity} can be extended for $CZ_m$ in the following way
\begin{equation}	CZ^{k}_{m}X_i=X_iCZ_{m}^{-k}(Z_{m})_j^{k},
\end{equation}
where the $CZ_{m}$ here is on qubit $i,j$. By similar argument like $C^nZ$, we can see that the group multiplication complexity, and thus inverse computation complexity, is $O(N^2\log m)$.

The group size and computational complexity for the CNOT dihedral group have been elaborated in \cite{Cross2016dihedral}. So, we omit the details here. Note that for $CZ_m$, their results only apply to $m=2^k$, and ours admit $m$ taking arbitrary positive integers. The scaling of the complexity results has been summarized and listed in the main text. In Fig.~\ref{fig:card}, we also show an accurate result of the size comparison between our group and the CNOT dihedral group by considering $N = n + 1$. The size of our group increases slower with respect to the qubit number than that of the CNOT dihedral group.

\begin{figure}[htbp!]
\centering
\includegraphics[width=.48\textwidth]{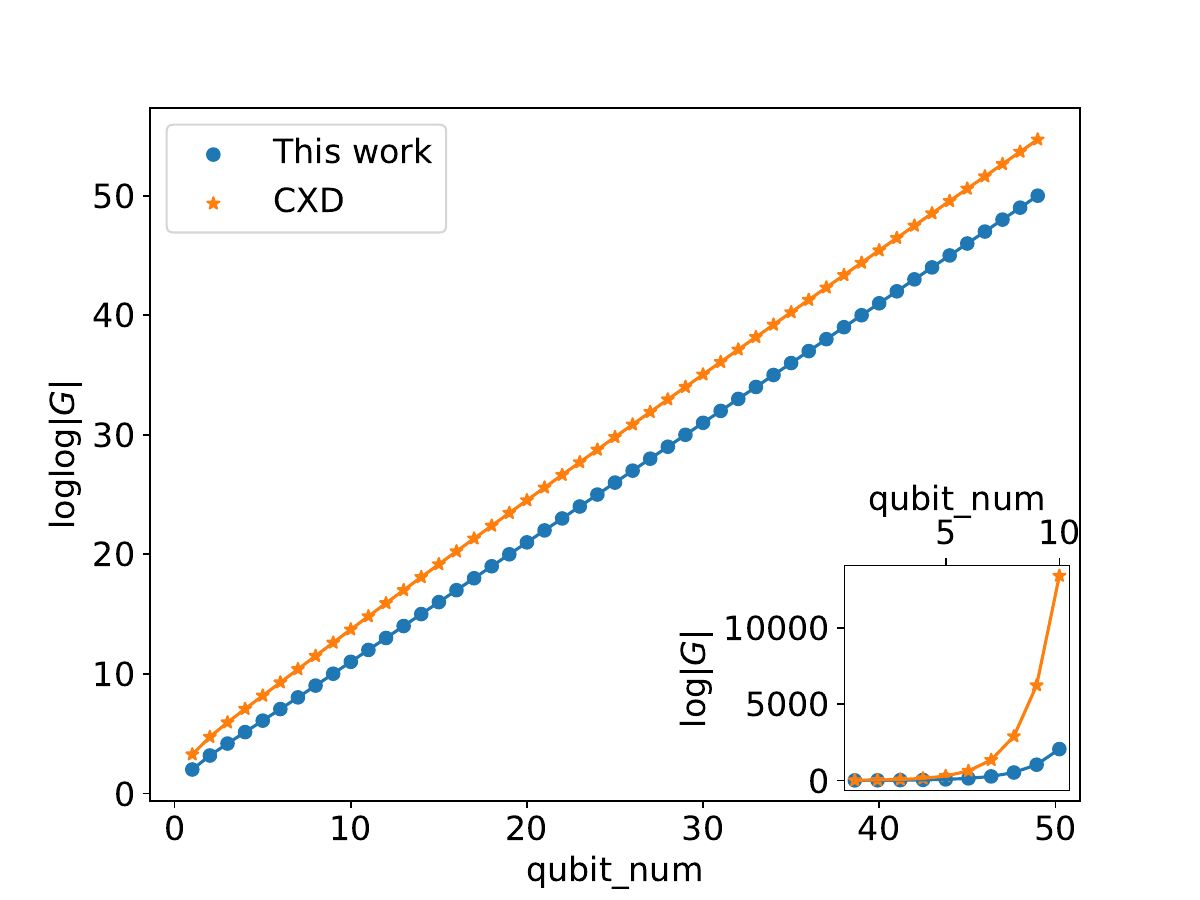}
\caption{Size comparison between the twirling groups for $C^nZ$ in our method and in~\cite{Cross2016dihedral}, associated with blue and orange patterns, respectively. $\abs{\textsf{G}}$ represents the group size. If taking a double logarithmic scale, both two kinds of groups increased nearly linearly with respect to the qubit number, but the CNOT dihedral group increases faster. The smaller figure demonstrates the results in the logarithm scale.}
\label{fig:card}
\end{figure}


Below, we discuss the hierarchy of twirling groups, which may be additionally interesting. In the main text, we have mentioned that by choosing the twirled gate to be $C^nZ_m$ with increasing $n$ and $m$, the size of the twirling group and the computational complexity of the group multiplication would keep increasing. The two complexities are closely related to the classical simulability of the twirling group, and the hierarchy of the groups actually forms a computational complexity hierarchy. The results may be useful in other fields like quantum complexity theory and quantum simulation.

Here, we discuss the hierarchy of groups themselves rather than the hierarchy of their sample and computational complexities. Specifically, we focus on the group of $\langle C^nZ, C^{n-2}Z, \cdots, CZ, X, Z_{2^k}\rangle$. Starting from Pauli group $\langle X, Z\rangle $, there are two ways to build up the hierarchy. One way is adding the number of sides of dihedral groups, i.e., increasing the number $k$ in $\langle X, Z_{2^k}\rangle$. The other way is adding multi-qubit controlled $Z$ gates $C^n Z$ or even adding multi-qubit controlled $X$ gates $C^n X$. Interestingly, $C^n Z$ can be contained in the CNOT-dihedral group $\langle CX, X, Z_{2^k}\rangle$ by choosing $k = n+1$~\cite{Cross2016dihedral}. The hierarchy is summarized below.
\begin{align}
\langle C^{n-1}Z, C^{n-2}Z, \cdots, CZ, X, Z_{2^k}\rangle &\leq \langle C^{n-1}Z, C^{n-2}Z, \cdots, CZ, X, Z_{2^{k+1}}\rangle \cdots \leq \langle C^{n-1}Z, C^{n-2}Z, \cdots, CZ, X, Z(\theta)\rangle \cdots;\\
\langle C^{n-1}X, X, Z_{2^k}\rangle &\leq \langle C^{n-1}X, X, Z_{2^{k+1}}\rangle \cdots \leq \langle C^{n-1}X, X, Z(\theta)\rangle \cdots \leq \mathrm{CRU}_n;\\
\langle X, Z_{2^k}\rangle &\leq \langle CZ, X, Z_{2^k}\rangle \leq \langle CCZ, CZ, X, Z_{2^k}\rangle \cdots;\\
\langle X, Z_{2^k}\rangle &\leq \langle CX, X, Z_{2^k}\rangle \leq \langle CCX, X, Z_{2^k}\rangle \cdots \leq \mathrm{CRU};\\
\langle C^{k-1}Z, C^{k-2}Z, \cdots, CZ, X, Z_{2^k}\rangle &\leq \langle CX, X, Z_{2^k}\rangle.
\end{align}
Here, $\mathrm{CRU}_n$ represents the CRU on $n$ qubits and $\mathrm{CRU}$ represents $\cup_{n} \mathrm{CRU}_n$. As shown in the above equations, the upper limit of the hierarchy is CRU. The key point is that $\langle C^{n-1}X, X\rangle$ contains all permutations over the computational bases $\{0,1\}^n$ based on Lemma~\ref{lemma:permutation}. Thus, $\langle C^{n-1}X, X, Z(\theta)\rangle$ contains all quantum gates like $\Pi W$ where $\Pi$ is an arbitrary permutation gate and $W$ is an arbitrary diagonal gate. This is exactly the necessary and sufficient condition for CRU. Also, from the group hierarchy, it is clearer that the twirling group of our protocol is smaller than that in~\cite{Cross2016dihedral} for multi-qubit controlled-$Z$ gates $C^nZ$. It is straightforward by investigating $\langle C^{n-1}Z, C^{n-2}Z, \cdots, CZ, X, Z\rangle \leq \langle C^{n-1}Z, C^{n-2}Z, \cdots, CZ, X, Z_{2^k}\rangle \leq \langle CX, X, Z_{2^k}\rangle$.


\section{Random compiling}\label{appendsc:RC}
In this part, we discuss the application of our results in random compiling. Recall that in the task of random compiling, we need to turn the noisy quantum gate $\widetilde{\mathcal{U}} = \mathcal{U} \Lambda$ to $\mathcal{U} \Lambda_G$, where the twirled noise channel $\Lambda_G$ is supposed to be a Pauli channel. Here, we consider a more general case than that in the main text. Assuming we implement a quantum circuit, $\widetilde{C} = \cdots \widetilde{\mathcal{U}}_2 \widetilde{\mathcal{U}}_1 = \cdots \mathcal{U}_2 \Lambda_2 \mathcal{U}_1 \Lambda_1$, as shown in Fig.~\ref{fig:RCappend}(a), the task of random compiling is tailoring $\widetilde{C}$ into $ \cdots \mathcal{U}_2 \Lambda_{2G} \mathcal{U}_1 \Lambda_{1G}$ as shown in Fig.~\ref{fig:RBappend}(c). The quantum gates $U_i, i\in \mathbb{Z}_+$ are tailored gates and, in general, are very noisy in experiments. They are always non-local, acting on several qubits. To tailor them, we add twirling gates $G_i, G_i'$ beside them as shown in Fig.~\ref{fig:RCappend}(b). The twirling gates are always noiseless and more easily implementable than the tailored gates.

\begin{figure}[htbp!]
\centering
\includegraphics[width=.5\textwidth]{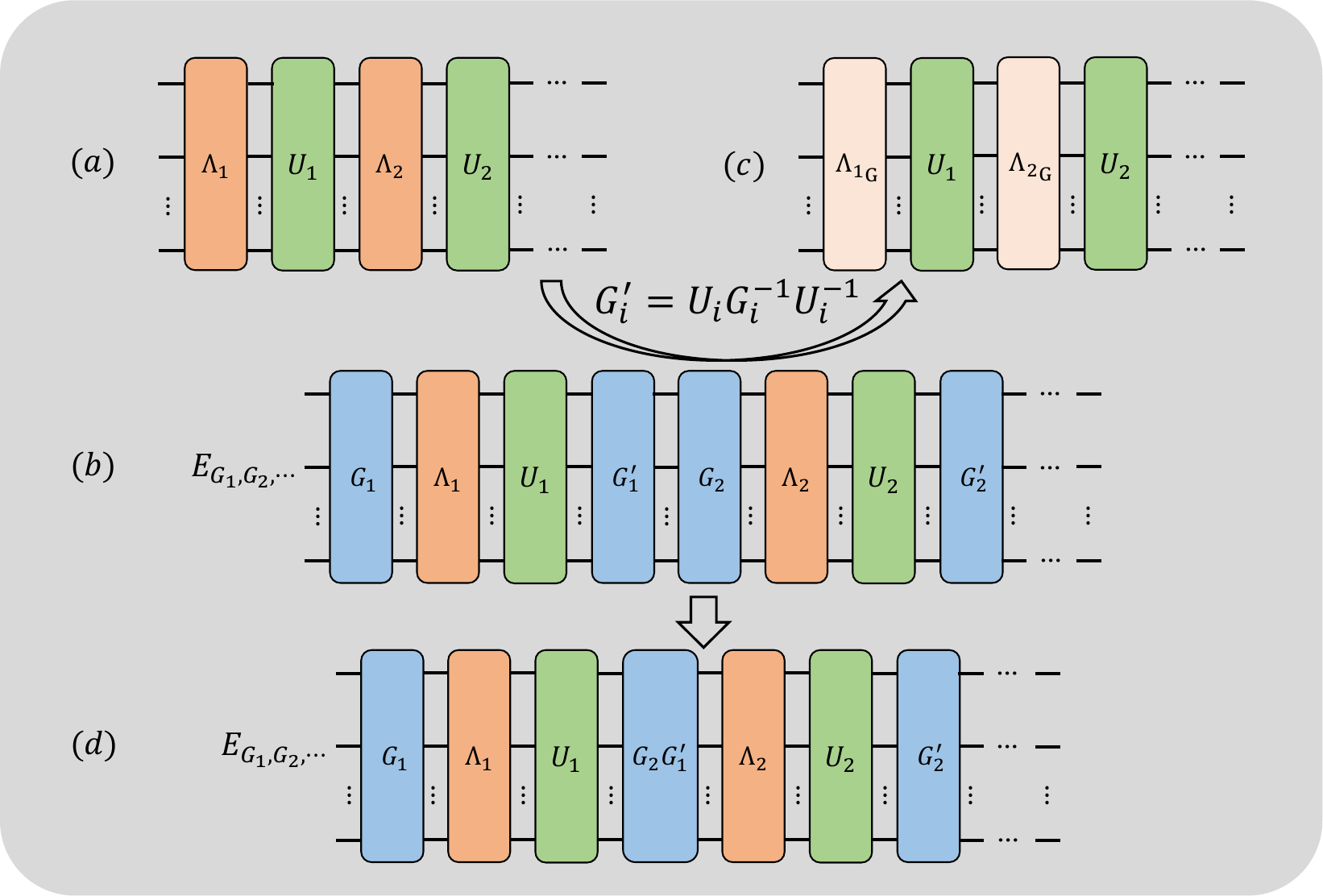}
\caption{Random compiling for $U_1$, $U_2$, $\cdots$, by inserting twirling gates $G_i$ and $G'_i = U_iG_i^{-1}U_i^{-1}$ beside $U_i$ with $i\in \mathbb{Z}_+$. $\Lambda_{iG}$ is the $\mathsf{G}$-twirled noise channel required to be a Pauli channel. In reality, $G_i$ and $G'_{i-1}$ are merged and implemented together as quantum gate $G_iG'_{i-1}$.}
\label{fig:RCappend}
\end{figure}

In the original work of~\cite{Wallman2016compiling}, the tailored multi-qubit quantum gates are restricted to be Clifford gates, and the twirling gates are chosen as Pauli gates. Nonetheless, in experiments, some multi-qubit non-Clifford gates, like CS and CCZ, are also important and native. They can be directly realized without gate composition. A random compiling scheme for these non-Clifford gates is conducive to their applications in quantum information tasks. Below, we show that our twirling groups constructed for benchmarking target gate $U$ can also be used to tailor $U$ in the setting of random compiling.

Similar to the discussion in RB, we investigate the requirements of the twirling gates for tailoring a gate set $\mathsf{U} = \{U_1, U_2, \cdots\}$. Note that in Figs.~\ref{fig:RCappend}(a),~\ref{fig:RCappend}(b), and~\ref{fig:RCappend}(c) we have shown the process to tailor $\cdots \mathcal{U}_2 \Lambda_2 \mathcal{U}_1 \Lambda_1$ into $ \cdots \mathcal{U}_2 \Lambda_{2G} \mathcal{U}_1 \Lambda_{1G}$. The key point is inserting twirling gates $G_i$ and $G'_i=U_iG_i^{-1}U_i^{-1}$ beside $U_i$ where $G_i$ is randomly sampled from twirling group $\mathsf{G}$. In reality, the gates $G_i$ and $G'_{i-1}$ are merged and implemented together to reduce quantum resource consumption. Thus, the twirling gates in general belong to $\mathsf{V} = \bigcup_{U\in \mathsf{U}} \mathsf{G}U\mathsf{G}U^{\dagger}$. Thus, in random compiling to tailor gate set $\mathsf{U}$, we should optimize the choice of $\mathsf{G}$ to simplify $\mathsf{V}$ as much as possible. It is equivalent to solving the following question.

\begin{question}\label{ques:RC}
Given a gate set, $\mathsf{U} = \{U_1, U_2, \cdots\}$, minimize $\abs{\mathsf{V}}$, such that $\mathsf{V} = \bigcup_{U\in \mathsf{U}} \mathsf{G}U\mathsf{G}U^{\dagger}$ and the twirling group $\mathsf{G}$ satisfies
\begin{equation}
\Lambda_G = \mathbb{E}_{G\in \mathsf{G}} \mathcal{G} \Lambda \mathcal{G}^{\dagger} \ \text{is a Pauli channel.}
\end{equation}
\end{question}

It can be seen that if the gate set $\mathsf{U}$ only contains a single gate $U$, then the solution to Question~\ref{ques:RB} is at least a feasible solution to Question~\ref{ques:RC}. The solution may not be optimal but will provide an upper bound on $\abs{\mathsf{V}}$. Now we suppose that the gate set $\mathsf{U}$ only contains a multi-qubit controlled phase gate $U=C^nZ_m$. Question~\ref{ques:RC} reduces to finding the smallest set of $\mathsf{G}U\mathsf{G} U^{\dagger}$ while $\Lambda_G$ is a Pauli channel. Similar with the results in RB, if we assume $\mathsf{G}$ belonging to CRU, then $\mathsf{G}U \mathsf{G} U^{\dagger}$ cannot be smaller than $\mathsf{X}U\mathsf{X}U^{\dagger}$ for $U=C^nZ_m$ based on Lemma~\ref{lemma:cosetcardconstraint}. Meanwhile, by choosing $\mathsf{G}$ as Pauli group $\mathsf{P}_n$, we would obtain a solution $\mathsf{V}$ to Question~\ref{ques:RC}, which is $\mathsf{P}_nU\mathsf{P}_n U^{\dagger} = \mathsf{X}U\mathsf{X}U^{\dagger} \mathsf{Z}$. Note that $\log |\mathsf{V}| - \log |\mathsf{X}U\mathsf{X}U^{\dagger}| \leq \log |\mathsf{Z}| = n$. It means that the size of $\mathsf{P}_nU\mathsf{P}_n U^{\dagger}$ would not be much larger than the optimal solution since the optimal solution contains $\mathsf{X}U\mathsf{X}U^{\dagger}$. Then, we found the nearly optimal solution for tailoring a multi-qubit controlled phase gate, $C^nZ_m$, in the setting of random compiling. It is simpler than the twirling group associated with $C^nZ_m$ in the setting of RB. Generally speaking, the requirements for the twirling gates in random compiling are much simpler than those in RB. We leave the problem of tailoring more kinds of quantum gates in the setting of random compiling to the future.

\section{Simulation}\label{appendsc:simulation}
In this section, we present the details of the simulation, including the setting of the noise model, the benchmarking protocol using the proposed twirling groups, and an additional benchmarking result for CS gate, which equals $\begin{pmatrix}
1  & 0\\
0 & i
\end{pmatrix}$. We denote the total number of qubits as $N$, which is set to be $2$ for the CS gate and $n+1$ for the $C^{n}Z$ gate.

\subsection{Error model}
In our simulation, we assume the twirling gates are always noiseless and only consider the noise of the target gate since the noise of twirling groups is easy to measure and eliminate by various methods that we mentioned before.
We generate this noise channel by composing three kinds of noise channels:
\begin{enumerate}[1.]
\item
Depolarize channel $\Lambda_d(\rho) = p\rho + (1-p)\frac{\mathbb{I}}{2^N}$, in which $\mathbb{I}$ represents the identity operator with dimension $2^N\times 2^N$. In our simulation experiments, $p$ is prefixed to be a constant.

\item
Local amplitude damping channel $\Lambda_a$. The amplitude damping channel on qubit $i$, where $1\leq i\leq N$, is defined as
\begin{equation}
\Lambda_a^i(\rho) = K_0^i\rho K_0^{i\dagger} + K_1^i\rho K_1^{i\dagger}.
\end{equation}
Here,
\begin{equation}
K_0^i \equiv \begin{pmatrix}
1 &  0\\
0 & \sqrt{p_i}
\end{pmatrix},
K_1^i \equiv \begin{pmatrix}
0 & \sqrt{1-p_i} \\
0 & 0
\end{pmatrix}.
\end{equation}
The parameter $p_i$ denotes the noise strength of the amplitude damping channel and is uniformly and randomly sampled from $[0, 0.02]$. Then, the local amplitude damping noise on $N$ qubits is defined as $\Lambda_a = \bigotimes_{i=1}^N \Lambda_a^i$.

\item
Unitary channel $\Lambda_u$. We consider that there exist two types of unitary noise for the target gate at the same time. One is SWAP-coupling noise, and the other is multi-Z-coupling noise. The SWAP-coupling noise is defined as
\begin{equation}
U_{\text{SWAP}}\equiv e^{i \sum_{1 \leq j < k \leq N} \alpha_{jk}\text{SWAP}_{jk}},
\end{equation}
where $\text{SWAP}_{jk}$ denotes the SWAP gate between qubits $j$ and $k$. The parameter $\alpha_{jk}$ is used to adjust the noise strength of SWAP-coupling noise and is uniformly and randomly sampled from $[0, 0.01]$. The ZZ-coupling noise is defined as
\begin{equation}
U_{\text{ZZ}}\equiv e^{i \sum_{z\in \{0,1\}^N, \abs{z}\geq 2} \beta_z \ketbra{z}}.
\end{equation}
The parameter $\beta_{z}$ represents the coupling strength and is uniformly and randomly sampled from $[0, 0.1]$. This noise records the phase fluctuation of the computational basis with a weight of no less than 2. The total unitary noise is set as $U = U_{\text{ZZ}}U_{\text{SWAP}}$, and the corresponding channel representation is $\Lambda_u$.
\end{enumerate}
Based on the three kinds of noise, the noise of the target gate is generated by $\Lambda = \Lambda_u \circ \Lambda_a \circ \Lambda_d$.

Besides the noise of the gate, we also set errors for state preparation and measurement (SPAM). In our simulation experiments, the initial states only contain $\ket{0}^{\otimes N}$ and $\ket{+}^{\otimes N} = H^{\otimes N} \ket{0}^{\otimes N}$ where $H$ is the Hadamard gate. In the simulation, we suppose the Hadamard gate is noiseless, and the error of the initial state only comes from the noise when preparing $\ket{0}^{\otimes n}$. We assume that on each qubit, the state $\ket{0}$ would turn into $\ket{1}$ with probability 0.02, resulting from the thermal excitation. We take the measurement to be noiseless as there is always a measurement error correction in experiments.

\subsection{Benchmarking protocol and additional simulation results}
Below, we present the benchmarking protocol for CCZ and CS with the twirling group $\text{CZD} = \langle CZ, X, S\rangle$. The benchmarking protocol for $C^nZ$ is almost the same except for substituting twirling group $\langle CZ, X, S\rangle$ with $\langle C^{n-1}, \cdots, CZ, X, S\rangle$.

To help readers more deeply understand the choice of the twirling groups and the benchmarking protocols, we take CS gate for an example and present the Liouville representation of its twirled noise channel under several different twirling groups, including Pauli group $\mathsf{P} = \langle X, Z \rangle$, CZ Pauli group $\text{CZP} = \langle CZ, X, Z \rangle$, CZ dihedral group $\text{CZD} = \langle CZ, X, S \rangle$, and CNOT dihedral group $\text{CXD} = \langle CX, X, T \rangle$. The results are shown in Figs.~\ref{fig:channelp},~\ref{fig:channelczp},~\ref{fig:channelcz}, and~\ref{fig:channelcx}.

\begin{figure}[htbp!]
\centering
\includegraphics[width=.85\textwidth]{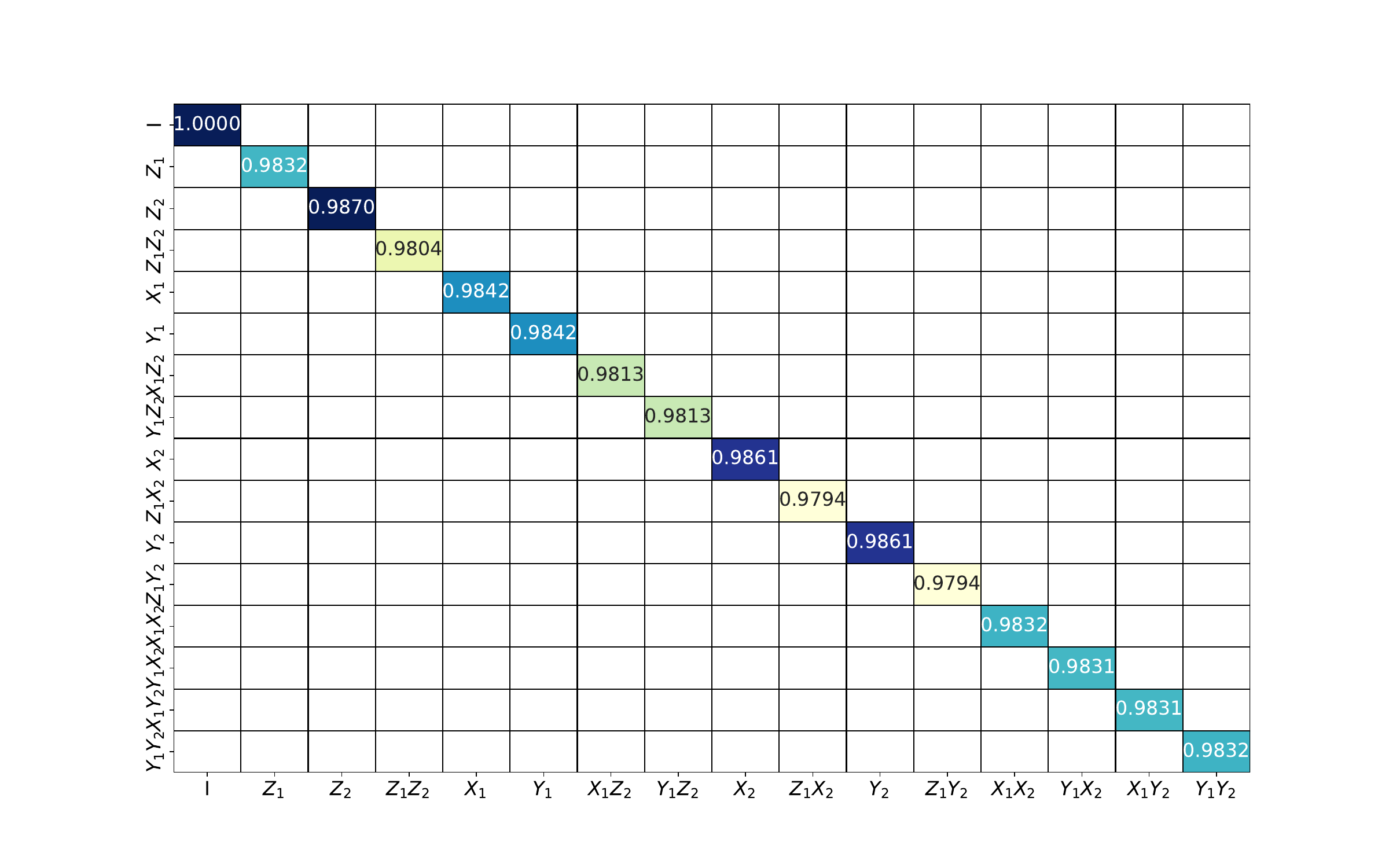}
\caption{Pauli Liouville representation of $\Lambda_{\mathsf{P}}$ where $\Lambda$ is the noise channel of CS gate and $\mathsf{P}$ denotes the Pauli group. The values in vacant squares are all 0 and we omit them. The two axes record the bases of the Pauli-Liouville representation.}
\label{fig:channelp}
\end{figure}

\begin{figure}[htbp!]
\centering
\includegraphics[width=.85\textwidth]{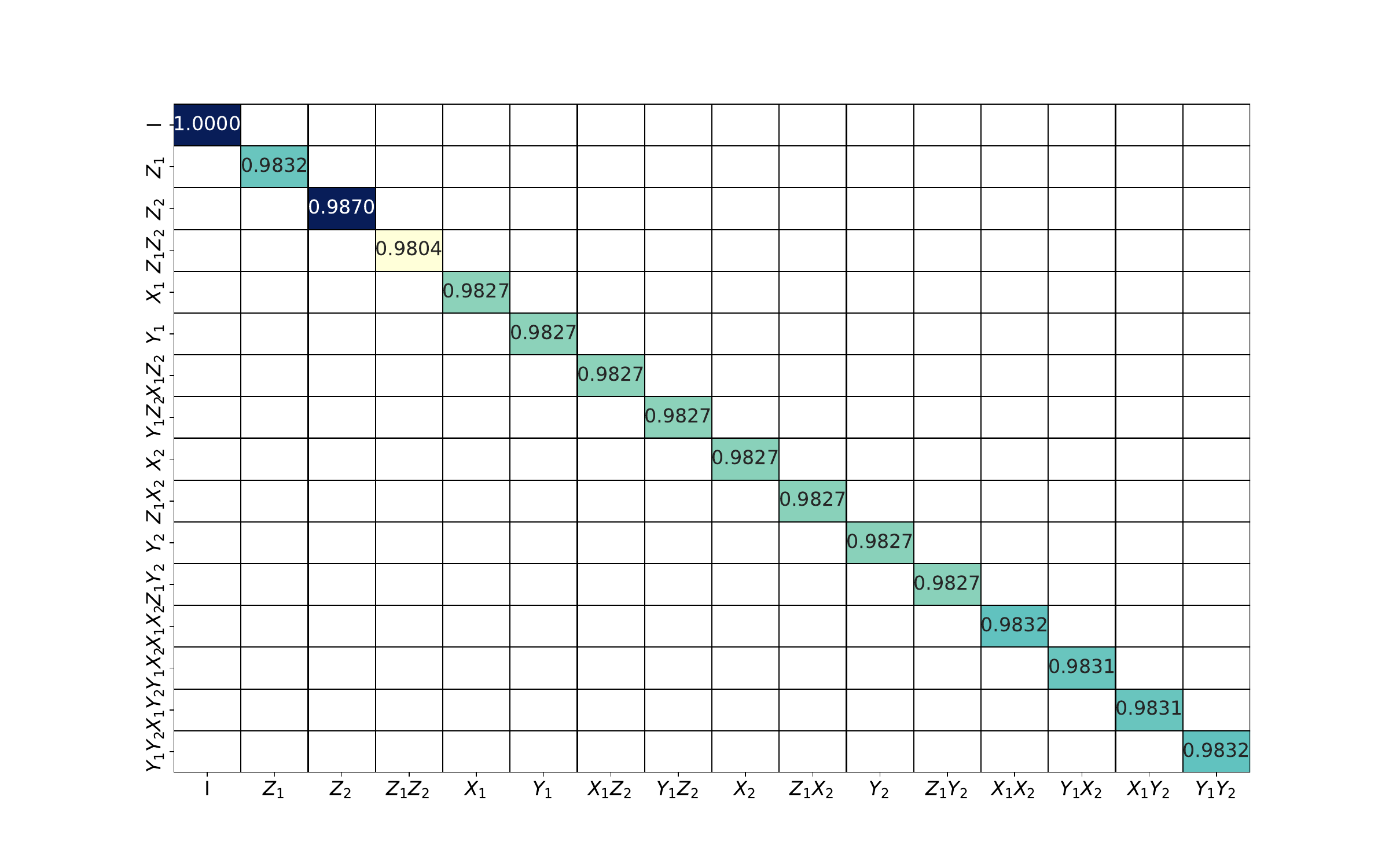}
\caption{Pauli Liouville representation of $\Lambda_{\textsf{CZP}}$ where $\Lambda$ is the noise channel of CS gate and \textsf{CZP} denotes the twirling group $\langle CZ, Z, X\rangle$. The values in vacant squares are all 0, and we omit them. Note that the elements corresponding to $X_1X_2$ and $Y_1Y_2$ are the same, and the elements corresponding to $X_1Y_2$ and $X_1Y_2$ are the same. But the two values are different. Generally, the elements corresponding to $X_1$ and $Y_1$ are different, and the elements corresponding to $X_2$ and $Y_2$ are also different. However, in this case, the noise channel is very special, and the difference can hardly be seen by accident.}
\label{fig:channelczp}
\end{figure}

\begin{figure}[htbp!]
\centering
\includegraphics[width=.85\textwidth]{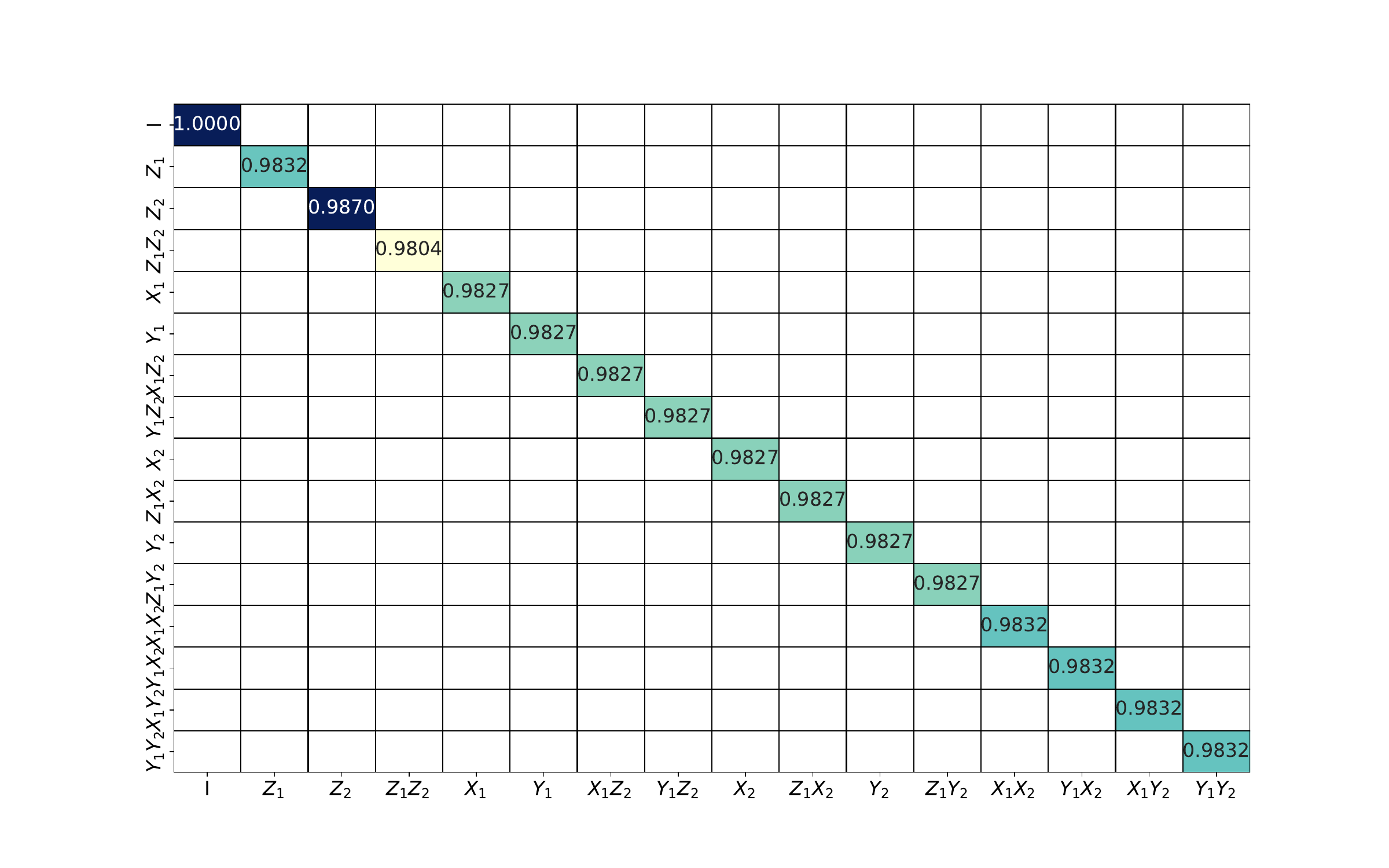}
\caption{Pauli Liouville representation of $\Lambda_{\textsf{CZD}}$ where $\Lambda$ is the noise channel of CS gate and \textsf{CZD} denotes the twirling group $\langle CZ, Z, S\rangle$. The values in vacant squares are all 0, and we omit them. The matrix elements corresponding to $X_a\mathsf{Z}$ are all averaged respectively for $X_a\in \{X_1, X_2, X_1X_2\}$.}
\label{fig:channelcz}
\end{figure}

\begin{figure}[htbp!]
\centering
\includegraphics[width=.85\textwidth]{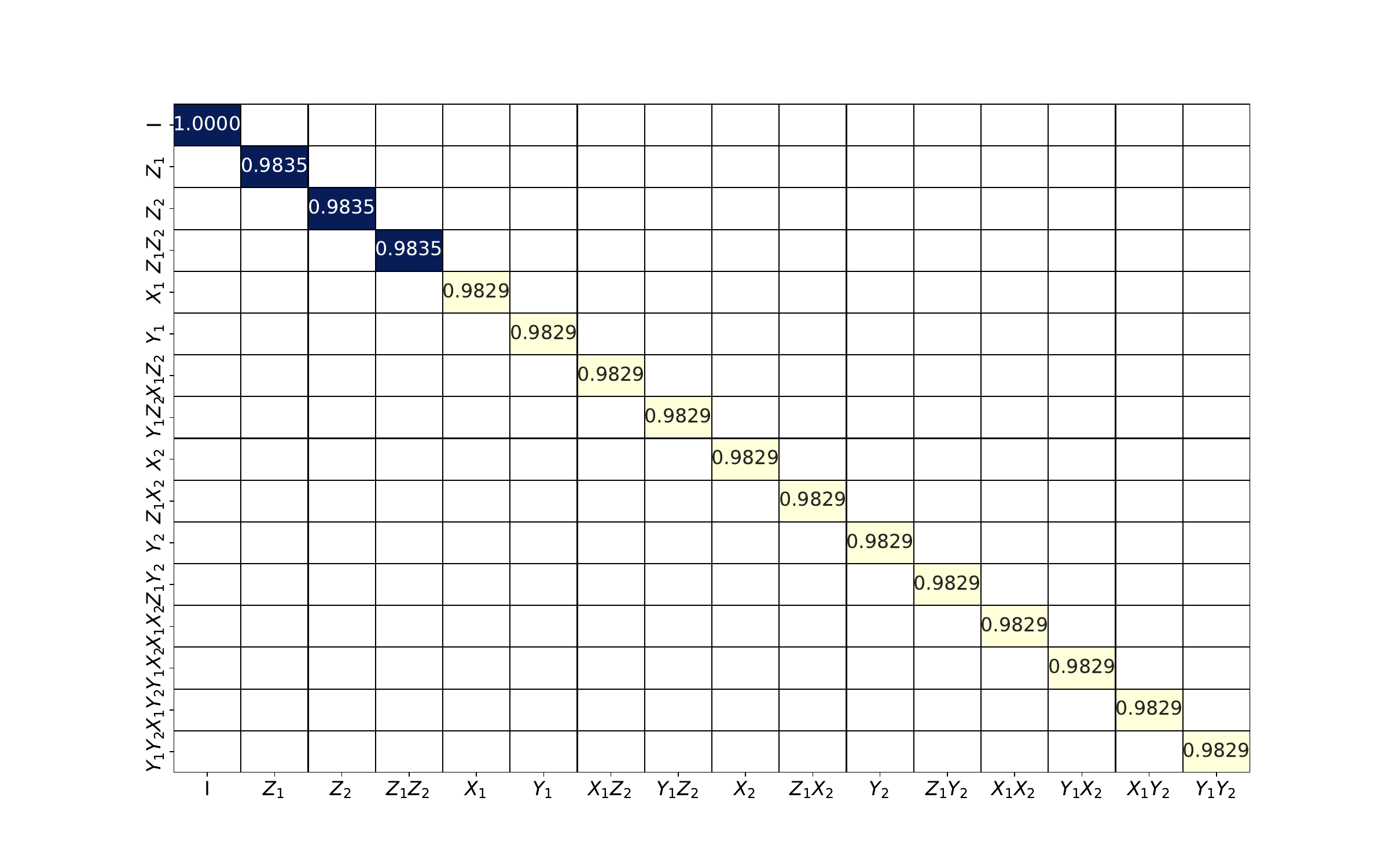}
\caption{Pauli Liouville representation of $\Lambda_{\textsf{CXD}}$ where $\Lambda$ is the noise channel of CS gate and \textsf{CXD} denotes the twirling group $\langle CX, X, T\rangle$. The values in vacant squares are all 0, and we omit them. The second to the fourth diagonal elements are averaged, and the last 12 diagonal elements are also averaged.}
\label{fig:channelcx}
\end{figure}

The matrix elements of the original noise channel $\Lambda$ of $CS$ gate in the Pauli-Liouville representation are all non-zero and contain complex numbers. After twirled by the Pauli group, there are only diagonal terms left to be nonzero, as shown in Fig.~\ref{fig:channelp}. If we use a larger twirling group, the CZ Pauli group, for any $X_a\in \{X_1, X_2, X_1X_2\}$, the matrix elements corresponding to $X_a$ and $X_aZ_1Z_2$ would be averaged, and the elements of $X_aZ_1$ and $X_aZ_2$ would be also averaged. The matrix elements of $\Lambda_{\textsf{CZP}}$ are shown in Fig.~\ref{fig:channelczp}. Its diagonal blocks are associated with the bases set in $\{ \{\mathbb{I}\}, \{Z_1\}, \{Z_2\}, \{Z_1Z_2\}, \{X_aZ_1, X_aZ_2\}, \{X_a, X_aZ_1Z_2\} | X_a\in \{X_1, X_2, X_1X_2\}\}$. Unfortunately, this group does not satisfy $\Lambda_{\textsf{CZP}}\mathcal{U} = \mathcal{U}\Lambda_{\textsf{CZP}}$ where $U$ is the target gate CS. The Pauli Liouville representation of the CS gate is shown in Fig.~\ref{fig:mapcs}. The diagonal blocks of the CS gate are associated with the bases set in $\{ \{\mathbb{I}\}, \{Z_1\}, \{Z_2\}, \{Z_1Z_2\}, \{X_a, X_aZ_1, X_aZ_2, X_aZ_1Z_2\} | X_a\in \{X_1, X_2, X_1X_2\}\}$ and are inconsistent with that of $\Lambda_{\textsf{CZP}}$. If we further use a bit larger twirling group, CZ dihedral group, the matrix elements corresponding to $X_a\mathsf{Z}$ would be all averaged. Thus, the diagonal blocks of $\Lambda_{\textsf{CZD}}$ are consistent with that of $\mathcal{U}$ and $\Lambda_{\textsf{CZD}}\mathcal{U} = \mathcal{U}\Lambda_{\textsf{CZD}}$ holds. Using the CZ dihedral group for twirling, we can obtain $\Lambda_{\textsf{CZD}}^m$. Combining $Z$-basis and $X$-basis measurements, we can extract the diagonal terms of $\Lambda_{\textsf{CZD}}^m$ and hence obtain the diagonal terms of  $\Lambda_{\textsf{CZD}}$ via fitting. The full protocol is shown in detail in Box~\ref{box:CnZmRB}. The situation for the CCZ gate is similar. Although CZ Pauli group $\langle CZ, Z, X\rangle$ can make $\Lambda_{\textsf{CZP}}\mathcal{U} = \mathcal{U}\Lambda_{\textsf{CZP}}$ for $U=\text{CCZ}$, the number of different parameters of $\Lambda_{\textsf{CZP}}$ is $3*(2^N-1)$ and is more than $2*(2^N-1)$. The parameters cannot be extracted only with $Z$-basis and $X$-basis measurements. On the other hand, the number of parameters of $\Lambda_{\textsf{CZD}}$ is equal to $2*(2^N-1)$. The parameters can be directly extracted only with $Z$-basis and $X$-basis measurements. Hence, we utilize the CZ dihedral group for tailoring CCZ to simplify the SPAM settings.\\

\begin{mybox}{box:CnZmRB}{Procedures for benchmarking CCZ or CS with CZ dihedral group}
\begin{enumerate}
\item
First initialize the state, $\ket{\psi_0} = \ket{0}^{\otimes N}$ or $\ket{\psi_1} = \ket{+}^{\otimes N}$ where $N$ is the number of qubits. $N=2$ for CS gate and $N=3$ for CCZ gate.

\item
Choose a positive integer, $M$, and choose two sets of positive integers $\{m_1, m_2, ..., m_M\}$ and $\{K_1, K_2, ..., K_M\}$. Here, $\{m_1, m_2, ..., m_M\}$ is the set of circuit depths and for $1\leq i\leq m$, $K_i$ is the number of sampled sequences when circuit depth equals $m_i$.

\item
For each integer $1\leq i\leq M$, uniformly and randomly sample $2m_i$ gates from $\langle CZ, S, X\rangle$ for $K_i$ times, which we denote $\{G_{j,1}, G_{j,2}, \cdots, G_{j,2m_i}\}$, $1\leq j\leq K_i$.

\item
For each integer $1\leq i\leq M$ and $1\leq j\leq K_i$, implement gate sequence
\begin{equation}
\widetilde{S}(j, m_i) = \widetilde{U}_{inv}\prod_{t=1}^{m_i} U^{\dagger}G_{j, 2t}UG_{j, 2t-1},
\end{equation}
where $\widetilde{\cdot}$ represents the noisy version of the quantum gate and $U_{inv} = (\prod_{t=1}^{m_i} U^{\dagger}G_{j, 2t}UG_{j, 2t-1})^{-1}$. $U$ is the target gate CCZ or CS.

\item
For initial state $\ket{\psi_0} = \ket{0}^{\otimes N}$, measuring all observables from $\{\mathbb{I}, Z\}^{\otimes N}$ of the final state via $Z$-basis measurement. For initial state $\ket{\psi_1} = \ket{+}^{\otimes N}$, measuring all observables from $\{\mathbb{I}, X\}^{\otimes N}$ of the final state via $X$-basis measurement. That is, for each $Q_k \in \{\mathbb{I}, Z\}^{\otimes N}$, estimate
\begin{equation}
f_{Z}(j, m_i, k) = \tr(\widetilde{Q}_k \widetilde{S}(j, m_i)(\widetilde{\rho}_0)).
\end{equation}
For each $P_k \in \{\mathbb{I}, X\}^{\otimes N}$, estimate
\begin{equation}
f_{X}(j, m_i, k) = \tr(\widetilde{P}_k \widetilde{S}(j, m_i)(\widetilde{\rho}_1)).
\end{equation}
Here, $\widetilde{\rho}_0$ and $\widetilde{\rho}_1$ are noisy versions of $\ket{0}^{\otimes N}$ and $\ket{+}^{\otimes N}$, respectively.

\item
Average the results of different gate sequences and obtain
\begin{align}
f_Z(m_i, k) &= \frac{1}{K_i}\sum_{j=1}^{K_i}f_{Z}(j, m_i, k)\\
f_X(m_i, k) &= \frac{1}{K_i}\sum_{j=1}^{K_i}f_{X}(j, m_i, k).
\end{align}

\item
For each $Q_k$, fit $f_Z(m_i, k)$ to function
\begin{equation}
f(m_i) = A\lambda_{Z,k}^{m_i}
\end{equation}
and obtain $\lambda_{Z,k}$. For each $P_k$, fit $f_X(m_i, k)$ to function
\begin{equation}
f(m_i) = A\lambda_{X,k}^{m_i}
\end{equation}
and obtain $\lambda_{X,k}$.

\item
Estimate the fidelity of the target gate via
\begin{equation}
F = \frac{\sum_k\lambda_{Z,k} + 2^N(\sum_k\lambda_{X,k} - 1)}{4^N}.
\end{equation}

\item
If one wants to separate further the noise of the twirling gates and the target gate, execute the following step. Replace the target gate with identity and repeat the above processes to estimate the fidelity of the twirling groups, $F_G$. Then, estimate the fidelity of $U$ with $F$ and $F_G$ by the technique of interleaved RB~\cite{Magesan2012interleavedRB}. In our simulation, we omit this step by assuming the twirling gates are noiseless. This assumption is also taken for protocols using the CNOT dihedral group.
\end{enumerate}
\end{mybox}

Note that in the above benchmarking procedure, we use the circuit structure in~\cite{Zhang2023Scalable} instead of~\cite{Erhard2019CB}. That means we implement random twirling gates interleaved with $U$ and $U^{\dagger}$ instead of being just interleaved with $U$. For $C^nZ$ gate, $U=U_{inv}$ and this modification does not influence anything. For more general case that $U$ is $C^nZ_m$ gate like CS gate, under a mild assumption that $U$ and $U^{\dagger}$ has the same noise $\Lambda$, the above procedure can estimate the fidelity of $\sqrt{\mathcal{U}^{\dagger}\Lambda_G\mathcal{U}\Lambda_G}$, which is equal to the fidelity of $\Lambda$. In this modified procedure, the circuit depth only needs to be chosen as the multiples of 2. While within the circuit structure in Fig.~\ref{fig:RBappend}, just like~\cite{Erhard2019CB}, the circuit depth needs to be chosen as the multiples of the order of the target gate to ensure the inverse gate belonging to the twirling group. For instance, the order of $C^nZ_m$ is $m$ and can be large for $C^nZ_m$ with large $m$. In this case, the modified benchmarking procedure in Box~\ref{box:CnZmRB} can offer an advantage of shorter circuit depth in circuit implementation. One can also use the circuit structure in Fig.~\ref{fig:RBappend} to benchmark the channel.

If using CNOT dihedral group for twirling, the twirled noise channel $\Lambda_{\textsf{CXD}}$ would be more symmetric than $\Lambda_{\textsf{CZD}}$ and $\Lambda_{\textsf{CXD}}\mathcal{U} = \mathcal{U}\Lambda_{\textsf{CXD}}$ also holds. The elements corresponding to $Z_1$, $Z_2$, and $Z_1Z_2$ will be averaged and the elements corresponding to $\{X_a\mathsf{Z}| X_a\in \{X_1, X_2, X_1X_2\}\}$ will also be averaged. Thus, $\Lambda_{\textsf{CXD}}$ only has two parameters as shown in Fig.~\ref{fig:channelcx}. With the protocol in~\cite{Cross2016dihedral}, one can obtain these two parameters. Note that in~\cite{Cross2016dihedral}, there are two ways to extract the two parameters. One way is extracting the two parameters with double exponential fitting, which performs badly, and we do not use it. We simulate the other way that using two SPAM settings extracts the parameters:
\begin{enumerate}[a.]
\item \label{item:sma}
preparing and measuring $\ket{0}^{\otimes N}$;
\item \label{item:smb}
preparing and measuring $\ket{+}^{\otimes N}$,
\end{enumerate}
respectively. Within the first SPAM setting, one can nearly get $\bra{0}^{\otimes N}\Lambda^m_{\textsf{CXD}}\ket{0}^{\otimes N} \approx A\lambda_Z^m+B$ and obtain $\lambda_Z = \frac{\sum_k \lambda_{Z,k}-1}{2^N-1}$. Here, $\lambda_{Z,k}$ is the diagonal term of $\Lambda_{\textsf{CZD}}$ mentioned in Box~\ref{box:CnZmRB}. Similarly, with the second SPAM setting, one can nearly get $\bra{+}^{\otimes N}\Lambda^m_{\textsf{CXD}}\ket{+}^{\otimes N} \approx A\lambda_X^m+B$ and obtain $\lambda_X = \frac{\sum_k \lambda_{X,k}-1}{2^N-1}$. One can also use the state preparation and measurement settings and the post-processing methods in Box~\ref{box:CnZmRB} to obtain these two parameters.

Below, we show the simulation result of CS gate as a supplement. The result is similar to that of CCZ and $C^nZ$ gates in the main text, as shown in Fig.~\ref{fig:CS}. The result indicates that the benchmarking protocol with the CZ dihedral group can perform as well as that with the CNOT dihedral group if using the procedure in Box~\ref{box:CnZmRB}. Both two are better than the protocol using the CNOT dihedral group using SPAM settings~\ref{item:sma} and~\ref{item:smb}. It is reasonable as the protocol in Box~\ref{box:CnZmRB} utilizes all the bases of the measurements while settings~\ref{item:sma} and~\ref{item:smb} are only involved with one basis in one measurement.

\begin{figure}[htbp!]
\centering
\includegraphics[width=.5\textwidth]{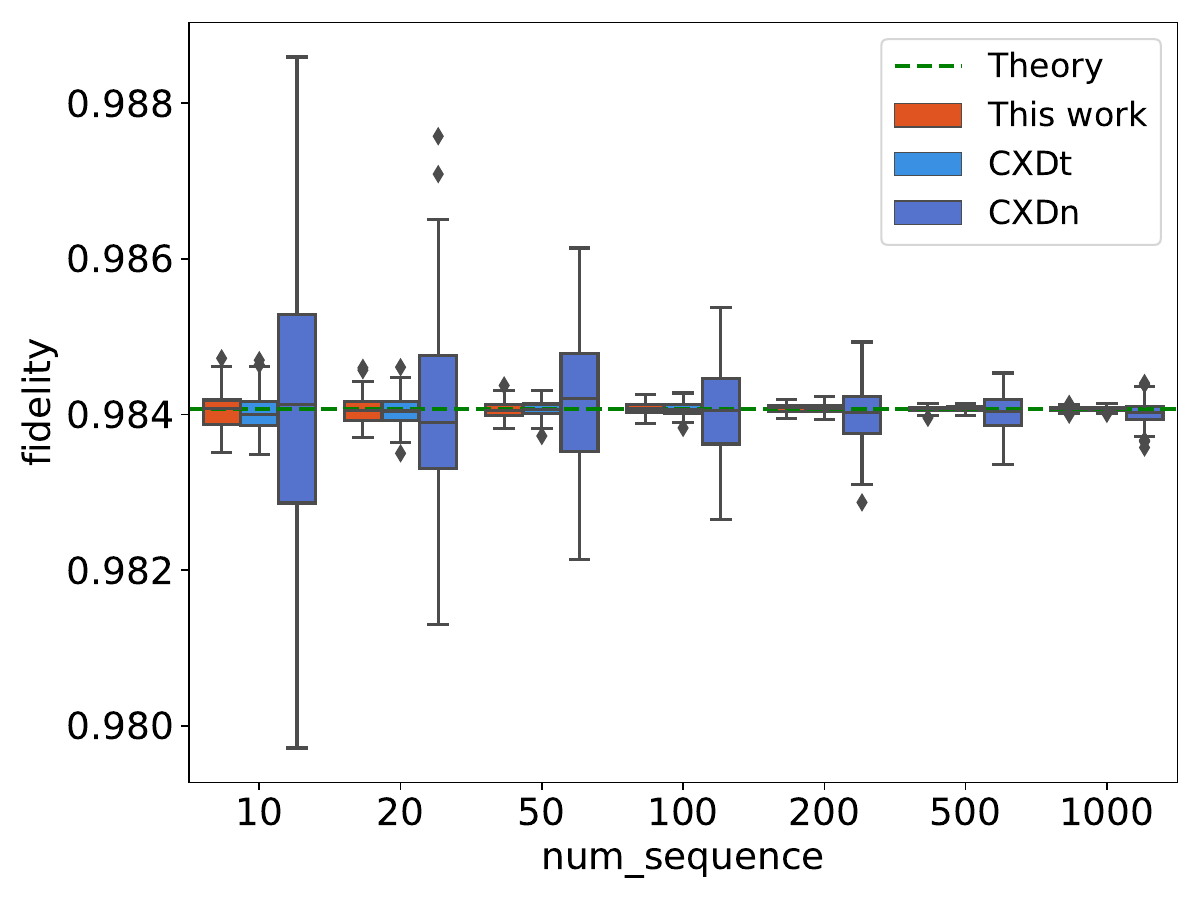}
\caption{Benchmarking results for CS with CZ dihedral group $\langle CZ, X, S\rangle$ and CNOT dihedral group $\langle CX, X, T\rangle$. There are two processing procedures for the CNOT dihedral group. The original procedure in~\cite{Cross2016dihedral}, or the procedure with SPAM settings~\ref{item:sma} and~\ref{item:smb}, is denoted as CXDn. The procedure using settings in Box~\ref{box:CnZmRB} is denoted as CXDt. The green dashed line is the theoretical value of the noise channel fidelity, and the other three box plots are the estimated fidelities of 100 groups. Each group employs circuits with depth list $\{m_1, m_2, ..., m_M\} = \{3, 6,\cdots, 30\}$. For each depth, we take the same number of sampled sequences, which means $K_1 = K_2 = ... = K_M$. The horizontal axis gives the value of $K_i, 1\leq i\leq M$.}
\label{fig:CS}
\end{figure}

\bibliographystyle{apsrev4-1}

\bibliography{bibDihedral}


\end{document}